\documentclass[10pt, conference]{IEEEtran}

\usepackage{graphicx}
\usepackage[table,xcdraw]{xcolor}

\usepackage[colorlinks=true,linkcolor=black,citecolor=black,urlcolor=blue]{hyperref}

\usepackage{amsmath,amssymb,amsthm}

\usepackage{subfig}

\usepackage{textcomp} % \textregistered \textcopyrightdddd
\usepackage{stmaryrd} % ⟦ and ⟧
\newcommand{\plusplus}{+\!\!\!+}

\usepackage[capitalize,nameinlink]{cleveref} % must be loaded after hyperref and amsmath
\crefformat{definition}{Defn.~#2#1#3}
\crefformat{equation}{Eq.~(#2#1#3)}
\crefformat{lemma}{Lem.~#2#1#3}
\crefformat{theorem}{Thm.~#2#1#3}
\crefformat{section}{Sec.~#2#1#3}

\usepackage{doi}

\usepackage{aliascnt}

\usepackage{thmtools,thm-restate}% CTAN:macros/experimental/thmtools

\newtheorem{theorem}{Theorem}%[section]

\newtheorem{lemma}[theorem]{Lemma}

\theoremstyle{definition} % All newtheorems after this will be non-italicized
\newtheorem{definition}[theorem]{Definition}

\usepackage{url}
\theoremstyle{remark} % The title Case will be italicized
\newtheorem{proofpart}{Case}[theorem]

\usepackage{mathpartir}

\usepackage{booktabs}

\usepackage{xcolor}
\definecolor{ltblue}{rgb}{0,0.4,0.4}
\definecolor{dkblue}{rgb}{0,0.1,0.6}
\definecolor{dkgreen}{rgb}{0,0.35,0}
\definecolor{dkviolet}{rgb}{0.3,0,0.5}
\definecolor{dkred}{rgb}{0.5,0,0}

\usepackage[noend]{algpseudocode}

\usepackage{tikz} %commutativity diagrams
\usetikzlibrary{cd}
\usetikzlibrary{positioning}
\usetikzlibrary{shapes.multipart}
\usetikzlibrary{%
  arrows,%
  shapes.misc,% wg. rounded rectangle
  shapes.arrows,%
  shapes.callouts,
  chains,%
  matrix,%
  positioning,% wg. " of "
  scopes,%
  decorations.pathmorphing,% /pgf/decoration/random steps | erste Graphik
  decorations.text,
  shadows%
}
\usetikzlibrary{quantikz}
\usepackage{adjustbox}

\usepackage{textcomp} % textquotesingle

\usepackage{xparse}
\NewDocumentCommand{\optionalParens}{s m m}{
    \IfBooleanTF{#2}{\left(#3\right)}{\IfBooleanTF{#1}{~#3}{#3}}
}
\NewDocumentCommand{\apply}{m s O{} m}{
    #1#3 \optionalParens*{#2}{#4}
}
\NewDocumentCommand{\applytwo}{m O{} s m s m}{ 
   #1#2~\optionalParens{#3}{#4}~\optionalParens{#5}{#6}
}

\usepackage{xspace}

\usepackage{braket}
\usepackage{blochsphere}

\usepackage{comment}
\usepackage{style/draft}

\drafttrue

\newnote[JP]{jennifer}{red}
\newnote[AS]{albert}{blue}
\newnote[AM]{anne}{purple}
\newnote{fixme}{green}
\newnote[Q]{question}{purple}
\newnote{note}{green}
\newnote{todo}{green}
\newnote[Cite:]{tocite}{green}

\newcommand{\ShortVersion}
    { \includecomment{ShortVersion}
      \excludecomment{ExtendedVersion}
    }
    
\usepackage{enumerate}
\usepackage[shortlabels,inline]{enumitem}
\includecomment{DraftVersion} % exclude for submission version
\ShortVersion

\renewcommand{\textsf}{\text}

\newcommand{\registered}{$^{\text{\textregistered}}$\xspace}

\newcommand{\PCAST}{\text{PCOAST}\xspace}
\newcommand{\tket}{t$\ket{\text{ket}}$\xspace}
\newcommand{\PCOAST}{\text{PCOAST}\xspace}
\newcommand{\PCOASTone}{\text{PCOAST}\ensuremath{_1}\xspace}
\newcommand{\PCOASTFT}{\text{PCOAST}\ensuremath{_{\text{FT}}}\xspace}
\newcommand{\tketone}{\textit{tket}\ensuremath{_1}\xspace}
\newcommand{\tkettwo}{\textit{tket}\ensuremath{_2}\xspace}
\newcommand{\qiskitthree}{\textit{qiskit}\ensuremath{_3}\xspace}
\newcommand{\qiskittwo}{\textit{qiskit}\ensuremath{_2}\xspace}

\newcommand{\companion}{\citep{Schmitz2023PCOAST}\xspace}

\newcommand{\ifExtended}[2]{}
\begin{ExtendedVersion}
    \renewcommand{\ifExtended}[2]{#1}
\end{ExtendedVersion}
\begin{ShortVersion}
    \renewcommand{\ifExtended}[2]{#2}
\end{ShortVersion}

\newcommand{\PrepZ}{\texttt{PrepZ}}
\newcommand{\MeasZ}{\texttt{MeasZ}}

\newcommand{\RX}[1]{\texttt{RX}(#1)}

\newcommand{\support}{\textsf{supp}}
\newcommand{\Pauli}{\mathcal{P}}

\newcommand{\poprgate}{\gate[style={rounded corners}]}
\newcommand{\startingpoprgate}{\gate[nwires={1},style={rounded corners}]}

\newcommand{\startingframegate}[2]{\gate[nwires={1,2}, wires=#1, style={rounded corners}]{#2} }
\newcommand{\poprgrouping}[2]{\gategroup[#1, steps=#2, style={dashed, rounded corners, fill=blue!20, inner xsep=2pt}, background]{}}
\newcommand{\pcastgraph}[2]{\gategroup[#1, steps=#2, style={rounded corners, fill=blue!20, inner xsep=2pt}, background]{}}

\newcommand{\cqstate}[2]{#1 \mapsto #2}
\NewDocumentCommand{\CQSTATE}{O{}}{\textsf{CQ}^{#1}}

\newcommand{\cLookup}[2]{#1[#2]} % lookup for classical states
\newcommand{\cAssign}[2]{#1 \gets #2}
\newcommand{\hold}{\textit{hold}\xspace}
\newcommand{\release}{\textit{release}\xspace}

\NewDocumentCommand\interp{sm}
        {\IfBooleanTF{#1}
                {\conjugate{[ #2 ]}}
                {[ #2 ]}
        }
\NewDocumentCommand\interpM{m}{\llbracket #1 \rrbracket}
\usepackage{MnSymbol} % \upmodels
\NewDocumentCommand\commute{smm}
        {#2 \IfBooleanT{#1}{\not}\upmodels #3}
\NewDocumentCommand\anticommute{}{\commute*}
\NewDocumentCommand{\commutativity}{mm}{\lambda(#1,#2)}

\newcommand{\CZ}{\texttt{CZ}}
\newcommand{\CX}{\texttt{CNOT}}
\newcommand{\CNOT}{\textsf{CNOT}}

\newcommand{\RZ}{\texttt{RZ}}
\newcommand{\RY}{\texttt{RY}}
\newcommand{\RXY}{\texttt{RXY}}
\newcommand{\X}{\text{X}}
\newcommand{\eff}[1]{\textsf{eff}#1}
\newcommand{\effX}{\eff{X}}
\newcommand{\effZ}{\eff{Z}}
\newcommand{\frameof}[1]{F^{#1}}
\newcommand{\unitaryof}[1]{U^{#1}}
\newcommand{\termof}[1]{t^{#1}}
\newcommand{\tr}{\textsf{tr}}
\newcommand{\dom}{\textsf{dom}}

\newcommand{\fv}{\textsf{freevars}}

\newcommand{\hmult}{\odot}
\newcommand{\fwdAction}[1]{\overrightarrow{#1}}
\newcommand{\bwdAction}[1]{\overleftarrow{#1}}

\NewDocumentCommand{\Rot}{t2 m m o o}
    {\IfBooleanTF{#1}
        {\textsf{Rot}^2(#2, #3, #4, #5)}
        {\textsf{Rot}(#2,#3)}}
\newcommand{\Prep}[2]{\textsf{Prep}(#1,#2)}
\NewDocumentCommand{\Meas}{O{} m}{
        \textsf{Meas}^{#1}(#2)
}
\NewDocumentCommand{\conjugate}{sm}{{\optionalParens{#1}{#2}}^\ast}

\NewDocumentCommand{\pauliToFrame}{t2 m}
        {\IfBooleanTF{#1}
                {F^{#2, \tfrac{\pi}{2}}}
                {F^{#2}}
        }

\NewDocumentCommand\cosBy{t2 m}
        {\IfBooleanTF{#1}
                {\cos{(\tfrac{#2}2)}}
                {\cos{#2}}
        }
\NewDocumentCommand\sinBy{t2 m}
        {\IfBooleanTF{#1}
                {\sin{(\tfrac{#2}2)}}
                {\sin{#2}}
        }

\NewDocumentCommand{\commuteTTRule}{o}
        {\IfNoValueTF{#1}
                {\ensuremath{\bot^t}}
                {\textsc{#1-\ensuremath{\upmodels^t}}}}
\NewDocumentCommand{\commuteTPRule}{o}
        {\IfNoValueTF{#1}
                {\ensuremath{\bot^p}}
                {\textsc{#1-\ensuremath{\upmodels^p}}}}

\NewDocumentCommand{\mergestep}{s}
        {\longrightarrow_M \IfBooleanT{#1}{^\ast}}
\NewDocumentCommand{\framestep}{s}
        {\longrightarrow_F \IfBooleanT{#1}{^\ast}}

\newcommand\ADDTERM{\todo{change to addNode}\xspace}
\newcommand{\addTerm}[2]{\ADDTERM(#1,#2)}
\newcommand\ADDNODE{\textsc{addNode}\xspace}
\newcommand{\addNode}[2]{\ADDNODE(#1,#2)}
\newcommand\ADDVERTEX{\textsc{addVertex}\xspace}
\newcommand\REMOVEVERTEX{\textsc{removeVertex}\xspace}
\newcommand\ADDEDGE{\textsc{addEdge}\xspace}
\newcommand{\circuitToPoPR}[1]{\textsc{circToGraph}\xspace(#1)}
\newcommand{\nodeCost}{\textsc{nodeCost}\xspace}
\newcommand{\reduceNode}{\textsc{reduceNode}\xspace}
\newcommand{\gateCost}{\textsc{gateCost}\xspace}
\newcommand{\addGate}{\textsc{addGate}\xspace}
\newcommand{\BEGIN}{\textsf{BEGIN}\xspace}
\newcommand{\END}{\textsf{END}\xspace}
\newcommand{\gateof}[1]{g^{#1}}

\DeclareMathOperator*{\argmin}{arg\,min} % Packages and macros specific to this project
\title{\PCOAST: A Pauli-based Quantum Circuit Optimization Framework
}

\author{
    \IEEEauthorblockN{
        Jennifer Paykin\IEEEauthorrefmark{1}\IEEEauthorrefmark{3},
        Albert T. Schmitz\IEEEauthorrefmark{1}\IEEEauthorrefmark{3},
        Mohannad Ibrahim\IEEEauthorrefmark{1},
        Xin-Chuan Wu\IEEEauthorrefmark{2},
        and
        A. Y. Matsuura\IEEEauthorrefmark{1}
        }
    \IEEEauthorblockA{\IEEEauthorrefmark{1}
    \textit{Intel Labs, Intel Corporation}, 
    Hillsboro, OR, USA \\
    }
    \IEEEauthorblockA{\IEEEauthorrefmark{2}
    \textit{Intel Labs, Intel Corporation},
    Santa Clara, CA, USA
    }
    \IEEEauthorblockA{
        \IEEEauthorrefmark{3}
        These authors contributed equally. Email: \{jennifer.paykin, albert.schmitz\}@intel.com
    }
}

\usepackage[numbers]{natbib}
\begin{document}

\maketitle

\begin{DraftVersion}
% Add page numbers, do not include on submission
\thispagestyle{plain}
\pagestyle{plain}
\end{DraftVersion}

\begin{abstract}
    This paper presents the Pauli-based Circuit Optimization, Analysis, and Synthesis Toolchain (\PCOAST), a framework for quantum circuit optimizations based on the commutative properties of Pauli strings. Prior work has demonstrated that commuting Clifford gates past Pauli rotations can expose opportunities for optimization in unitary circuits. \PCOAST extends that approach by adapting the technique to mixed unitary and non-unitary circuits via generalized preparation and measurement nodes parameterized by Pauli strings. The result is the \PCOAST graph, which enables novel optimizations based on whether a user needs to preserve the quantum state after executing the circuit, or whether they only need to preserve the measurement outcomes. Finally, the framework adapts a highly tunable greedy synthesis algorithm to implement the \PCOAST graph with a given gate set.

\PCOAST is implemented as a set of compiler passes in the Intel\registered Quantum SDK. In this paper, we evaluate its compilation performance against two leading quantum compilers, Qiskit and \tket. We find that \PCOAST reduces total gate count by 32.53\% and 43.33\% on average, compared to the best performance achieved by Qiskit and \tket respectively, two-qubit gates by 29.22\% and 20.58\%, and circuit depth by 42.02\% and 51.27\%.

%updated numbers 4/28 2:30pm
\end{abstract}


\section{Introduction}
\label{sec:intro}
Quantum circuit optimizations address an important challenge in the design of efficient quantum computing systems by reducing the number of operations required to execute quantum algorithms~\citep{chong2017programming, prasad2006data, valiron2015programming}. Optimizations fall in two main classes: local, peephole-style optimizations~\citep{kliuchnikov2013optimization, abdessaied2014quantum, nam2018automated,Qiskit, pointing2021optimizing, xu2022quartz}, where local patterns of gates are replaced by other patterns; and global optimizations, where circuits are converted to an intermediate mathematical structure that highlights some semantic equivalence, simplified according to the rules of that structure, and synthesized back into a circuit that could be significantly different from the original.
%While peephole optimizations can be effective for recognizing simple patterns, they tend to miss richer opportunities for optimizations, both because they only act locally on neighboring gates, and because of the number and complexity of patterns needed to identify such opportunities.
This paper presents a novel global optimization framework called \PCAST, a Pauli-based Circuit Optimization, Analysis, and Synthesis Toolchain, which successfully reduces total gate count, two-qubit gate count, and depth compared to the best performance of state-of-the-art optimizing compilers Qiskit~\citep{Qiskit} and \tket~\citep{Sivarajah_2021_tket}. %On applications for quantum chemistry, these improvements rise to reductions in gate count by 79\%, two-qubit gates by 77\%, and depth by 85\%.

%Quantum circuit optimizations address an important challenge in the design of efficient quantum computing systems by reducing the number of operations required to execute quantum algorithms~\citep{chong2017programming, prasad2006data, valiron2015programming}. In this paper, we present a novel optimization framework called \PCAST, a Pauli-based Circuit Optimization, Analysis, and Synthesis Toolchain, which successfully reduces the number of gates in our benchmarks by an average of \todo{average} compared to top competitors.

Recently, a class of global optimizations based on Pauli rotations have proved successful in reducing gate count for unitary circuits beyond the reach of more traditional local peephole optimizations~\citep{Zhang2019,cowtan2019phase,Schmitz2021}. These optimizations take advantage of the fact that unitary circuits can be decomposed into Clifford gates (generated by the Hadamard gate $H$, the phase gate $S$ and the controlled-not gate $\CNOT$); and non-Clifford gates represented by Pauli rotations $\Rot{P}{\theta}=e^{-i \theta/2 P}$, where $P$ is one the Pauli matrices $X$, $Y$, or $Z$. Because Cliffords always map Paulis to Paulis by conjugation, it is possible to push Clifford unitaries $U$ past the non-Clifford Pauli rotations $\Rot{P}{\theta}$ to produce a new rotation $\Rot{U P U^\dagger}{\theta}$. Doing so can expose the fact that some rotations can be merged, reducing the number of gates required in the final circuit, even if the original gates did not appear directly next to each other in the original circuit~\citep{Zhang2019}. The implementation of Pauli gadgets (the equivalent of Pauli rotations) in the \tket compiler~\citep{Sivarajah_2021_tket} has been shown to reduce both the number of 2-qubit entangling gates, relevant for Noisy Intermediate Scale Quantum (NISQ) workflows, and single-qubit non-Clifford gates, relevant for fault-tolerant workflows~\citep{cowtan2019phase}.

These optimizations are quite powerful, but they fall short in two main ways. First, they focus on unitary circuits, so equivalences enabled by non-unitary operations like preparations and measurements are not taken into account. For example, if a $Z$ rotation occurs on a qubit that was prepared in a $Z$ eigenstate, that rotation can be eliminated. Incorporating non-unitary gates enables several powerful optimizations, such as the ability to drop all unitary gates after measurement when coherence on the quantum device doesn't need to be preserved after measurement. Second, while optimizations reduce the number of rotations and put an upper bound on the size of the resulting circuit, the number of two-qubit gates required can be expensive if the rotations are synthesized in a suboptimal order. \citet{Schmitz2021} and \citet{li2022paulihedral} both explore synthesis in the context of Hamiltonian simulation, with the goal of automatically synthesizing a Hamiltonian, described as a product of Pauli rotations, into a circuit with a designated gate set. In particular, \citeauthor{Schmitz2021} describe a greedy search algorithm  that reduces this search problem to a variation of the traveling salesman problem~\citep{Schmitz2021}.

\begin{figure}
  \centering
  \subfloat[Example circuit. The $c$ argument in $\MeasZ^c$ indicates the classical variable that the measurement outcome is written to.]{
    \begin{tikzpicture} \node[scale=0.6]{
    \begin{quantikz}[row sep=0.1cm]
    & \gate{\PrepZ}
        & \gate{\RX{\theta_1}} & \targ{}
        & \gate{\RX{\theta_2}} & \ctrl{1}
        & \gate{\RX{\theta_3}} & \gate{\MeasZ^{c_0}} & \qw \\
    & \gate{\PrepZ}
        & \gate{H}              & \ctrl{-1}
        & \qw                   & \targ{}
        & \qw                   & \gate{\MeasZ^{c_1}} & \qw
    \end{quantikz}
    }; \end{tikzpicture}
    \label{fig:demo-circuit}
  }
  \newline
  \subfloat[The \PCAST graph generated by the example circuit. The groupings indicate nodes to be merged together, where $\Rot{X_0}{\theta_1}$ and $\Rot{X_0}{\theta_2}$ combine to $\Rot{X_0}{\theta_1 + \theta_2}$, and $\Rot{Z_1}{\theta_3}$ is absorbed by $\Prep{Z_1}{X_1}$. Note that the measurement to variable $c_1$ has been transformed into a measurement of the first qubit, $Z_0$, due to the permutations of Clifford gates ($F$) past the measurement.]{
    \begin{tikzpicture} \node[scale=0.6]{
    \begin{tikzcd}[row sep=0.2cm]
    \poprgate{\Prep{Z_0}{X_0}} \arrow[end anchor={[xshift=-1.25ex]}]{r}
        & \startingpoprgate{\Rot{X_0}{\theta_1}} \poprgrouping{2}{1}
            \arrow[start anchor={[xshift=1ex]}]{r}
            \arrow[start anchor={[xshift=5ex, yshift=-4ex]}, end anchor={north}]{ddr}
        %& \poprgate{\Rot{X_0}{\theta_1 + \theta_2}}
            %\arrow[rd, dash]
            %\vqw{1}
        &  \startingpoprgate{\Meas[c_1]{Z_0}} \arrow[r]
        %\slice{}
        & \startingframegate{3}{
        F = \begin{pmatrix}
                Z_0 X_1 & Z_1 \\
                Z_0 & X_0 Z_1
            \end{pmatrix}
        }
    \\
    \poprgate{\Rot{Z_1}{\theta_3}} \poprgrouping{2}{1}
        & \startingpoprgate{\Rot{X_0}{\theta_2}}
        &
        &
    \\
    \startingpoprgate{\Prep{Z_1}{X_1}} 
            \arrow[start anchor={[xshift=1.5ex]}]{rr}
        & 
        & \startingpoprgate{\Meas[c_0]{Z_0 X_1}} \arrow[r]
        &
    \end{tikzcd}
    }; \end{tikzpicture}
    \label{fig:demo-graph}
  }
  \newline
    \subfloat[The optimized \PCAST graph obtained when specifying a release outcome. The optimization has reduced the support of the measurement node $\Meas{Z_0 X_1}$ to $\Meas{X_1}$ by recognizing that the $Z_0 X_1$ measurements can be reconstructed by measuring $X_1$ and combining the outcomes with the measurement of $Z_0$ classically. The measurement results of the optimized graph is guaranteed to produce the same probability distribution as \cref{fig:demo-graph}.]{
    \begin{tikzpicture} \node[scale=0.68]{
    \begin{tikzcd}[row sep=0.2cm]
    \poprgate{\Prep{Z_0}{X_0}} \arrow[r]
        &\poprgate{\Rot{X_0}{\theta_1 + \theta_2}} \arrow[r]
        & \poprgate{\Meas[c_1']{Z_0}}
            \arrow[rd]
        & \startingframegate{1}{F' = \begin{pmatrix}
                Z_0 & X_0 \\
                Y_1 & X_1
            \end{pmatrix}}
    \\
    \poprgate{\Prep{Z_1}{X_1}} \arrow[rr]
        & \qw
        & \poprgate{\Meas[c_0']{X_1}}
            \arrow[r]
        & \startingframegate{1}{ \begin{aligned}
            \mu =~&\cAssign{c_0}{c_0'}; \\
            &\cAssign{c_1}{c_0' + c_1'}
        \end{aligned} }
    \end{tikzcd}
    }; \end{tikzpicture}
    \label{fig:demo-optimized-graph}
  }
  \newline
  \subfloat[Optimized released circuit synthesized from \cref{fig:demo-optimized-graph}, along with assignments to classical variables to account for the release outcome optimization.]{
    \begin{tikzpicture} \node[scale=0.7]{
    \begin{tikzcd}[row sep=0.2cm]
    & \gate{\PrepZ} & \gate{\RX{\theta_1 + \theta_2}} & \gate{\MeasZ^{c_1'}} & \qw
    \\
    & \gate{\PrepZ} & \gate{\RY(-\tfrac{\pi}{2})} & \gate{\MeasZ^{c_0'}} & \qw &
    \end{tikzcd}
    }; \end{tikzpicture}
    \begin{tikzpicture} \node[scale=0.7]{
    \begin{tikzcd}
    \startingframegate{1}{\begin{aligned}
            \mu =~
            &\cAssign{c_0}{c_0'}; \\
            &\cAssign{c_1}{c_0' + c_1'}
    \end{aligned}}
    \end{tikzcd}
    }; \end{tikzpicture}
    \label{fig:demo-optimized-circuit}
  }
  \caption{\PCAST optimization on an example circuit.}
  \label{fig:demo}
\end{figure}

In this paper, we show that by addressing these deficits, Pauli strings can be used not only as a standalone optimization pass, but as a cohesive optimization framework, \PCAST. \cref{fig:demo} shows an example \PCOAST workflow. First, the circuit in \cref{fig:demo-circuit} is converted into the \PCAST graph in \cref{fig:demo-graph}. The graph highlights the commutativity of rotation nodes in relation to each other---if there is no dependency between two nodes, they commute with each other. In addition to the unitary Pauli rotations, our \PCAST graph contains non-unitary gates---preparation and measurement---that are parameterized by Pauli rotations and are subject to the same commutativity rules as rotations. To represent Cliffords, we use a compact representation as a Pauli frame $F$, otherwise known as a Pauli tableau~\citep{Aaronson2004}, which emphasizes the behavior of the Clifford on Pauli arguments (\cref{sec:frame}).
A summary of different types of nodes is shown in \cref{fig:nodesCategories}%
\ifExtended{ and an overview of the notations used is shown in \cref{fig:notations}}{}.

\begin{figure}
    \centering
    \resizebox{0.48\textwidth}{!}{%
    \begin{tikzpicture}
    
        \draw[fill=blue!5, very thick, rounded corners] (-1.25, -2.25) rectangle (10.5, 2.75);
        \node[align=center] (nodes) at (5.25, 2.5)
            {\PCAST nodes};

        \node[align=center, draw, fill=white, rectangle, rounded corners] (msf)
            at (9.0,0.65)
            {Measurement \\ space \\ functions \\
            $\mu : \mathcal{M}_1 \rightarrow \mathcal{M}_2$};

        \draw[fill=blue!10, very thick, rounded corners] (-1, -2.0) rectangle (7.5, 2.25);
        \node[align=center] (gates) at (3.75, 2.0)
            {Quantum gates};

        \node[align=center, draw, fill=white, rectangle, rounded corners] (prep)
            at (1.0,-1.25)
            {Pauli preparations \\
            $\Prep{P_Z}{P_X}$};
        \node[align=center, draw, fill=white, rectangle, rounded corners] (prep)
            at (5.0,-1.25)
            {Pauli measurements \\
            $\Meas[c]{P}$};

        \draw[fill=blue!15, very thick, rounded corners]
            (-0.75, -0.5) rectangle (7.2, 1.75);
        \node[align=center] (unitary) at (3.75, 1.5)
            {Unitary gates};

        \node[align=center, draw, fill=white, rectangle, rounded corners] (rotations) at (5.5, 0.5) { Pauli rotations \\ $\Rot{P}{\theta} = e^{-i \frac{\theta}{2} P}$ };

        \draw[fill=blue!20, very thick, rounded corners] (-0.5,-0.35) rectangle (3.75, 1.25);
        \node[align=center] (clifford) at (1.75, 1.0)
            {Clifford gates};

        \node[align=center, draw, fill=white, rectangle, rounded corners] (pauli) at (0.5, 0.25) {Paulis \\ $X_0 Z_1 Y_3$ };
        
        \node[align=center, draw, fill=white, rectangle, rounded corners] (frames) at (2.5, 0.25) {Pauli frames \\ $F : \Pauli \rightarrow \Pauli$ };

    \end{tikzpicture}
    }
    \caption{Types of \PCAST nodes in relation to each other. Note that Paulis themselves are not nodes, but are represented in \PCAST as Pauli frames.}
    \label{fig:nodesCategories}
\end{figure}

\begin{ExtendedVersion}
\begin{table}
    \centering
    \begin{tabular}{lll} \toprule
    Category & Notation & Definition \\
    \midrule % Constants
    Preliminaries
    %& $\theta \in \mathbb{R}$ & real rotation angle \\
    %& $\alpha \in \mathbb{C}$ & complex number \\
    & $g$ & Quantum gate \\
    & $C$ & Quantum circuit \\
    & $U$ & Unitary transformation \\
    %& $\ket{\phi}$ & quantum state \\
    & $\rho$ & Density matrix \\
    \midrule % Semantics
    Semantics
    & $\mathcal{M}$ & Measurement space \\
    & $c$ & Classical variable \\
    & $b \in \{0,1\}$ & Boolean value \\
    & $m \in \mathcal{M}$ & Classical state \\
    & $\gamma \in \CQSTATE[\mathcal{M}]$ & Classical-quantum state \\
    & $\mu : \mathcal{M}_1 \rightarrow \mathcal{M}_2$ 
        & Measurement space function \\
    %& \jennifer{?} & classical-quantum channel \\
    %& v & Pauli vector \\
    %& \jennifer{?} & Pauli map \\
    \midrule % Paulis
    \PCAST
    & $p \in \{I, X, Y, Z\}$ & Single-qubit Pauli \\
    & $P \in \Pauli{n}$ & $n$-qubit Pauli \\
    & $F$ & Pauli frame \\
    & $n$ & \PCAST node \\
    & $t$ & \PCAST term \\
    & $o \in \{\text{hold}, \text{release}\} $ & Measurement outcome \\
    & $G$ & \PCAST graph \\
    \bottomrule
\end{tabular}
    \caption{Notations used in this paper.}
    \label{fig:notations}
\end{table}

\end{ExtendedVersion}

The addition of these non-unitary nodes enables a host of additional internal optimizations on \PCAST graphs. Users can choose between two optimization outcomes: either a \emph{hold outcome}, where the optimizations preserve the semantics of the original circuit precisely; or a \emph{release outcome}, where more aggressive optimizations can be applied as long as they produce the same measurement results. A release outcome  will drop unitary gates that can be delayed until after measurement, with the guarantee that the measurement results will always be statistically equivalent. To achieve this, we introduce classical remappings of measurement variables via what we call measurement space functions. For example, if a \release outcome is specified for \cref{fig:demo}, it will be optimized to the \PCAST graph as shown in \cref{fig:demo-optimized-graph} using the measurement space function given in the bottom right.

Finally, we adapt \citet{Schmitz2021}'s synthesis algorithm to determine both how to order commuting nodes, and how to decompose multi-qubit nodes using sequences of two-qubit entangling gates. The result is shown in \cref{fig:demo-optimized-circuit}. We customize the synthesis algorithm with a number of heuristics to minimize cost according to a given cost model, map into a target gate set, and reduce the number of measurements required for a stabilizer search. Currently, synthesis primarily aims to minimize algorithm-level resource requirements like circuit depth, although the design allows for customization to prioritize other search criteria.\footnote{The full implications of such customization, including hardware-aware layout, routing, and scheduling, are beyond the scope of this paper.}

This work makes the following contributions:
\begin{itemize}
\item We develop a semantics in which to describe the behavior of \PCAST nodes, terms, and graphs that incorporates both classical and quantum states.
\item We introduce the key \PCAST data structures, including Pauli frames and the \PCAST graph. 
\item We present three major components of the \PCAST optimization: compiling a circuit to a \PCAST graph, optimizing the graph, and synthesizing a circuit back out. 
\item We implement \PCAST in C++ as a sequence of compiler passes in the Intel$^{\text{\textregistered}}\xspace$ Quantum Software Development Kit (SDK)\footnote{\url{https://developer.intel.com/quantumsdk}}~\citep{Khalate2022}, and evaluate its compilation performance against two state-of-the-art optimizing quantum compilers, Qiskit~\citep{Qiskit} and \tket~\citep{Sivarajah_2021_tket}. Our experimental results show that \PCOAST reduces total gate count by 32.53\% and 33.33\% on average, compared to the best performance achieved by Qiskit and \tket respectively, two-qubit gates by 29.22\% and 20.58\%, and circuit depth by 42.02\% and 51.27\%.
%\item We compare \PCAST to related work (\cref{sec:related}) and conclude with a discussion of future work (\cref{sec:conclusion}).
\end{itemize}

%\footnotetext{Full details of the internal optimizations and proofs of correctness can be found in a companion paper~\companion.}

\begin{ShortVersion}
An extended version of this paper gives full proofs for all lemmas and theorems~\citep{PaykinSchmitz2023PCOAST}.
%Full proofs of the lemmas and theorems in this paper can be found in the extended version~\citep{PaykinSchmitz2023PCOAST}.
\end{ShortVersion}

%This work makes the following contributions. \cref{sec:background} introduces necessary background about quantum circuits and Pauli matrices, while \cref{sec:semantics} develops a semantics in which to describe the behavior of \PCAST nodes, terms, and graphs that incorporates both classical and quantum states. \cref{sec:structures} introduces the main structures of the \PCAST optimization, including Pauli frames and non-unitary gates for preparation and measurement parameterized by arbitrary Paulis. Key theorems throughout these sections are given proofs in an extended version of this paper~\extendedref.

%\cref{sec:optimization} describes the three major components of the \PCAST optimization: compiling a circuit to a \PCAST graph, optimizing the graph, and synthesizing a circuit back out. Further details of the internal optimizations and proofs of correctness are given in a companion paper \companion.

%We have implemented \PCAST in C++ as a sequence of compiler passes in the Intel Quantum Software Development Kit (SDK)\footnote{\url{https://developer.intel.com/quantumsdk}}~\citep{Khalate2022}, and we evaluate its behavior in \cref{sec:results}. We compare the behavior of the \PCAST optimizations to other implementations implemented in Qiskit and \tket, and find that \todo{findings}.

%Finally, in \cref{sec:related} we discuss related work, and in \cref{sec:conclusion} we conclude with a discussion of future work, including plans to adapt the synthesis algorithm to include hardware-aware optimizations such as topology and scheduling.

\section{Background}
\label{sec:background}
\begin{comment}
\subsection{Principles of Quantum Computation} \label{sec:quantumbackground}
% qubit -> quantum gate -> pauli group

An $n$-qubit quantum state is a linear combination of $2^n$ basis states $\ket{b_0 \ldots b_{n-1}}=(b_0 ~\cdots ~b_{n-1})^T$ with norm $1$.
%
A quantum circuit is a series of quantum gates---operations applied sequentially to transform one quantum state into another. Unitary gates correspond to unitary matrices---$2^n \times 2^n$ complex matrices $U$ whose conjugate transpose is its own inverse---that transform a state $\ket{\varphi}$ into $U \ket{\varphi}$.  %The set of single-qubit gates and two-qubit controlled gates is universal, in that they can be combined to construct any quantum gate~\cite{nielsen2002quantum}.

Non-unitary operations, namely preparation and measurement, act probabilistically on quantum states.
\end{comment}
Quantum states are represented as density matrices $\rho$: positive semi-definite, Hermitian complex matrices with trace $1$. Unitary transformations act on them via conjugation: $U \rho U^\dagger$. A density matrix is called a pure state if it can be written as the outer product of two state vectors $\ket{\phi}\bra{\phi}$.
If not, it is a mixed state and can be written as the weighted sum of pure states, representing a probability distribution over pure states.
The behavior of a quantum circuit can be described as a function over density matrices known as a quantum channel~\cite{nielsen2002quantum}.
\subsection{The Pauli group}

A single-qubit Pauli $p$ is one of $X$, $Y$, $Z$, or $I$, where%, as shown in \cref{fig:PauliMultiplication}. 
%The figure shows that $p_1 \cdot p_2$ is always equal to another Pauli, scaled by a constant factor $\alpha \in \{1,-1,i,-i\}$.
\begin{align}
    X &= \begin{pmatrix}
        0 & 1 \\
        1 & 0
    \end{pmatrix} \qquad
    Y = \begin{pmatrix}
        0 & -i \\
        i & 0
    \end{pmatrix} \qquad
    Z = \begin{pmatrix}
        1 & 0 \\
        0 & -1
    \end{pmatrix}
\end{align}
$X$, $Y$, and $Z$ are Hermitian, meaning that $p p = I$, and satisfy $X Y Z = i$.
A consequence is that any two single-qubit Paulis either commute (written
$\commute{p_1}{p_2}$), meaning $p_1 \cdot p_2 = p_2 \cdot p_1$; or
anticommute ($\anticommute{p_1}{p_2}$), meaning %
$p_1 \cdot p_2 = -p_2 \cdot p_1$. We write
\begin{align}
    \lambda(p_1,p_2) &= \begin{cases}
        0 & \commute{p_1}{p_2} \\
        1 & \anticommute{p_1}{p_2}
    \end{cases}
\end{align}

\begin{comment}
\begin{figure}
    \centering
\begin{align*}
    \begin{aligned}
    X &= \begin{pmatrix}
        0 & 1 \\
        1 & 0
    \end{pmatrix} \\
    Y &= \begin{pmatrix}
        0 & -i \\
        i & 0
    \end{pmatrix}
    \end{aligned}
    \qquad
    \begin{aligned}
    Z &= \begin{pmatrix}
        1 & 0 \\
        0 & -1
    \end{pmatrix} \\
    I &= \begin{pmatrix}
        1 & 0 \\
        0 & 1
    \end{pmatrix}
    \end{aligned}
    \qquad
    \begin{aligned}
    %I \cdot p &= p = p \cdot I \\
    p \cdot p &= I \\
    X \cdot Y &= iZ = -Y \cdot X \\
    Z \cdot X &= i Y = -X \cdot Z \\
    Y \cdot Z &= i X = -Z \cdot Y
    \end{aligned}
\end{align*}    
    \caption{Multiplicative identities of the single-qubit Pauli group}
    \label{fig:PauliMultiplication}
\end{figure}
\end{comment}

An $n$-qubit Pauli $P =\alpha(p_0,\ldots,p_{n-1}) \in \Pauli_n$
is the tensor product of $n$ single-qubit Paulis scaled by 
$\alpha \in \{1,-1,i,-i\}$.
Its support, $\support(P)$, is the set of indices for which $p_i \neq
I$. We write $X_i$, $Y_i$, and $Z_i$ for Paulis with support $\{i\}$,
and thus the Pauli string $X_0 Z_2$ refers to $(X, I, Z)$.
%We also overload the notation $I$ for the Pauli with no support: $(I,\cdots,I)$.

Multiplication can be lifted to $n$-qubit Paulis as follows:
\begin{align}
    P_1 \cdot P_2 &= \alpha_1 \alpha_2 \sigma_0 \cdots \sigma_{n-1} (q_0,\ldots,q_{n-1})
\end{align}
where $P_i = \alpha_i (p_0^i,\ldots,p_{n-1}^i)$ and $p_i^1 \cdot p_i^2 = \sigma_i q_i$.
As a result, $n$-qubit Paulis form a group with identity $I=(I,\ldots,I)$.

Commutativity can be lifted to $n$-qubit Paulis via a binary function
$\commutativity{P}{P'} \in \{0,1\}$ such that
$P \cdot P' = (-1)^{\commutativity{P}{P'}} P' \cdot P$:
{\small \begin{align}
    \commutativity{\alpha(p_0,\ldots,p_{n-1})}{\alpha'(p_0',\ldots,p_{n-1}')}
        &= \sum_{i=0}^{n-1} \commutativity{p_i}{p_i'}
            \mod{2}
\end{align} }%
We write $\commute{P}{P'}$ if $\lambda(P,P')=0$
and $\anticommute{P}{P'}$ if $\lambda(P,P')=1$.
\begin{ExtendedVersion}
    Note that $\commutativity{P}{Q}=\commutativity{Q}{P}$.
\end{ExtendedVersion}

We are most interested in Hermitian Paulis where $P P = I$. Since single-qubit Paulis are all Hermitian, an $n$-qubit Pauli is Hermitian
if and only if its coefficient is $\pm1$.
\begin{ExtendedVersion}
\begin{lemma} \label{lem:HermitianPauliFacts}
Let $P_1$ and $P_2$ both be Hermitian. Then
\begin{itemize}
    \item $-P_1$ is Hermitian.
    \item If $\commute{P_1}{P_2}$, then $P_1 \cdot P_2$ is Hermitian.
    \item If $\anticommute{P_1}{P_2}$ then $i P_1 \cdot P_2$ is Hermitian.
\end{itemize}
\end{lemma}
\begin{proof}
    The first bullet is trivial.
    The second and third follow because 
    \begin{align}
        P_1 P_2 P_1 P_2 = (-1)^{\commutativity{P_1}{P_2}} P_1 P_1 P_2 P_2 = (-1)^{\commutativity{P_1}{P_2}} I.
    \end{align}
    If $\commute{P_1}{P_2}$ this is $I$, and if $\anticommute{P_1}{P_2}$ it is $-I$, in which case $(i P_1 P_2)\cdot(i P_1 P_2) = -i^2 I = I$.
\end{proof}
\end{ExtendedVersion}
The \emph{Hermitian product} of Hermitian Paulis is $P_1 \hmult P_2 = (-i)^{\commutativity{P_1}{P_2}}P_1 \cdot P_2$: if $P_1$ and $P_2$ are both Hermitian, then so is $P_1 \hmult P_2$.

\section{Semantics} \label{sec:semantics}
As quantum channels, \PCAST nodes could be seen as transformations on quantum states. However, because \PCAST deals with mixed unitary and non-unitary circuits, it also must account for transformations on \emph{classical states}, for example when writing measurement outcomes a classical registers. A classical state $m$ is a finite sequence of assignments of classical variables $c$ to boolean values $b \in \{0,1\}$, written
$
    \cAssign{c_0}{b_0}; \cdots; \cAssign{c_{n-1}}{b_{n-1}}.
$
The boolean value associated with a variable is written $\cLookup{m}{c}$, and 
a finite set of classical states is referred to as a \emph{measurement space} $\mathcal{M}$.
\begin{ExtendedVersion}
    The domain of a classical state $\dom(m)=\{c_0,\ldots,c_{n-1}\}$ is the set of variables it affects.
\end{ExtendedVersion}

\subsection{Classical-quantum states}

Instead of working with density matrices directly, we will operate over mixed \emph{classical-quantum states}~ \citep{feng2021quantum,selinger2004towards}. A cq-state $\gamma \in \CQSTATE[\mathcal{M}] = \mathcal{M} \rightarrow \mathcal{C}^n \times \mathcal{C}^n$, sometimes written $m \mapsto \gamma_m$, is a function from a classical state $m \in \mathcal{M}$ to the quantum state of the system after the measurement outcome $m$ is observed. The quantum state $\gamma_m$ is represented as a \emph{partial} density matrix $\gamma_m$, whose trace $0 \le \tr(\gamma_m) \le 1$ corresponds to the probability of observing $m$. 
The sum of all the partial density matrices in the image of a cq-state is a full density matrix with trace $1$.
%$ \sum_{m \in \mathcal{M}} \tr(\gamma_m) = 1$.

%, and will write $\bigoplus_{m \in M} \cqstate{m}{\rho_m}$ for the function that maps each $m$ to $\rho_m$. 
As an example, the cq-state obtained from executing the circuit $\PrepZ(0); H(0); \MeasZ^c(0)$ is $m \mapsto \tfrac{1}{2}\ket{\cLookup{m}{c}}\bra{\cLookup{m}{c}}$.
%\begin{align}
%    \cqstate{c \mapsto 0}{\frac{1}{2}\ket{0}\bra{0}}
%    \oplus
%    \cqstate{c \mapsto 1}{\frac{1}{2}\ket{1}\bra{1}}.
%\end{align}

When it is clear from context, we write $\rho$ for the constant cq-state $\_ \mapsto \rho$.
Scaling and summation of cq-states over the same measurment space is defined pointwise.

\PCAST utilizes two different equivalence relations on cq-states. The \emph{hold} relation completely preserves the quantum state corresponding to every classical state, while the \emph{release} relation only requires that the probability of being in the same state, $\tr(\gamma^i_m)$, is the same for each classical state $m$.%\footnote{If an outcome is not specified, as in $\gamma^1 \equiv \gamma^2$, a hold outcome is assumed.}
\begin{align}
    \gamma^1 \equiv^{\text{hold}} \gamma^2 &\iff \forall m \in \mathcal{M},~ \gamma^1_m=\gamma^2_m \\
    \gamma^1 \equiv^{\text{release}} \gamma^2 &\iff \forall m \in \mathcal{M},~ \tr(\gamma^1_m) = \tr(\gamma^2_m)
\end{align}

In this context, the semantics of a quantum circuit $C$ can be described as a classical-quantum channel---a linear function $\interpM{C} : \CQSTATE[\mathcal{M}_1] \rightarrow \CQSTATE[\mathcal{M}_2]$ between cq-states, where $\mathcal{M}_1$ is the set of states exectuion may be in before $C$, and $\mathcal{M}_2$ contains the states the program may be in after executing $C$.
\begin{ExtendedVersion}
    The set of variables that $C$ writes to is referred to as $\fv(C)$, and it will always be the case that
    \begin{align}
        \mathcal{M}_2 \subseteq \{m_1; \{\cAssign{c_i}{b_i}\} \mid m_1 \in \mathcal{M}_1 \wedge \fv(C)=\{c_i\} \}.
    \end{align}
\end{ExtendedVersion}

\subsection{Sum-of-Pauli semantics}

Classical-quantum channels will be used to describe the behavior of circuits and \PCAST nodes on classical variables. However, the majority of nodes do not affect classical states at all.
%Even measurement, which outputs results to a classical variable, does not read from the classical state.
In that case, their behavior can be naturally described as a \emph{Pauli map}, a function from $n$-qubit Paulis to cq-states. 
Intuitively, the quantum component of the input cq-state will be decomposed into a sum of Pauli operators scaled by arbitrary complex values, which we call \emph{Pauli vectors}. 

\begin{lemma} \label{lem:PauliVectorDecomposition}
Every $2^n \times 2^n$ complex matrix $A$ can be decomposed into a Pauli vector
$A = \sum_{i} \alpha_i P_i$.
\end{lemma}
\begin{ExtendedVersion}
\begin{proof}
    Notice that we can break down every any complex matrix $A$
    into a sum of its individual elements:
    \begin{align}
        A = \sum_{i,j} A_{i,j} \times \ket{i}\bra{j}
    \end{align}
    where $\ket{i}\bra{j}$ is $1$ at index $(i,j)$ and $0$ elsewhere.
    ~
    It suffices to
    show that each $\ket{i}\bra{j}$ can be represented as a sum of Paulis.
    For the 1-qubit case ($0 \le i,j \le 1$) this is easy to check:
    \begin{align}
        \ket{0}\bra{0} &= \begin{pmatrix} 1 & 0 \\ 0 & 0 \end{pmatrix}
            = \frac{1}{2}(I + Z) \\
        \ket{0}\bra{1} &= \begin{pmatrix} 0 & 1 \\ 0 & 0 \end{pmatrix}
            = \frac{1}{2}(X + i Y) \\
        \ket{1}\bra{0} &= \begin{pmatrix} 0 & 0 \\ 1 & 0 \end{pmatrix}
            = \frac{1}{2}(X - i Y) \\
        \ket{1}\bra{1} &= \begin{pmatrix} 0 & 0 \\ 0 & 1 \end{pmatrix}
            = \frac{1}{2}(I - Z)
    \end{align}

    Now we can prove the property by induction: suppose the property holds for matrices of size $n$, which we write $\ket{i} \cdot^n\bra{j}$, and consider $\ket{i} \cdot^{n+1}\bra{j}$. If $i,j$ is in the first 
    quadrant of the matrix e.g. $0 \le i,j < 2^{n}$,
    then 
    \begin{align}
        \ket{i} \cdot^{n+1} \bra{j}
        &= \begin{pmatrix} 1 & 0 \\ 0 & 0 \end{pmatrix} \otimes \ket{i} \cdot^{n} \bra{j}
        = \frac{1}{2} (I + Z) \otimes \ket{i} \cdot^{n} \bra{j}.
    \end{align}
    This can itself be written as a sum of scaled Paulis by the distributivity of
    the tensor product over addition. Similarly, if $0 \le i < 2^n$ and $2^n \le j < 2^{n+1}$
    then 
    \begin{align}
        \ket{i} \cdot^{n+1} \bra{j} &= \begin{pmatrix} 0 & 1 \\ 0 & 0 \end{pmatrix} \otimes \ket{i} \cdot^{n} \bra{j}
        = \frac{1}{2} (X + iY) \otimes \ket{i} \cdot^{n} \bra{j}
    \end{align}
    and similarly for the rest of the quadrants.
\end{proof}
\end{ExtendedVersion}

\begin{ExtendedVersion}
A complex matrix is one basis for describing a linear operator, and a Pauli vector another. In general, one might use one basis over another because it is more sparse (has fewer no-zero terms) or because an action on a vector is easier to understand in that basis. So it is for a Pauli vector: because Paulis are natural for both unitaries and measurables, it tends to be a sparser representation, and many gate actions (Cliffords, for example) are naturally understood in this basis.
\end{ExtendedVersion}

Multiplication of Pauli vectors, written $v_1 \cdot v_2$, and conjugate
transpose, $v^\dagger$, are defined in the expected way.
\begin{ExtendedVersion}
\begin{align}
    \left( \sum_i \alpha_i P_i \right) \cdot \left( \sum_j \alpha_j' P_j' \right)
        &= \sum_{i,j} \alpha_i \alpha_j' P_i P_j' \\
    \left( \sum_i \alpha_i P_i \right)^\dagger
        &= \sum_i \alpha_i^\ast P_i^\dagger
\end{align}
A Pauli vector $v$ has an inverse $v^{-1}$ when $v^{-1} \cdot v = v \cdot v^{-1}
= 1$.
    As expected, a Pauli vector $v$ is called \emph{unitary} when
$v^{-1} = v^\dagger$, and \emph{pure} when $v v = v$.
\end{ExtendedVersion}
%
%The decomposition can be used as an alternative to cq-channels.

\begin{definition} \label{defn:PauliMap}
A \emph{Pauli map} is a function $f : \Pauli_n \rightarrow \CQSTATE[\mathcal{M}]$ from $n$-qubit Paulis to cq-states. It can be lifted to a cq-channel
$\interp{f} : \CQSTATE[\mathcal{M}_0] \rightarrow \CQSTATE[\mathcal{M}_0 + \mathcal{M}]$ as\footnote{$\mathcal{M}_1 + \mathcal{M}_2 = \{m_1; m_2 \mid m_i \in \mathcal{M}_i\}.$}
    \begin{align}
        \interp{f}(\gamma) = m_0;m \mapsto \sum_{i} \alpha_i \gamma_i(m)
    \end{align}
    where
    $\gamma(m_0)=\sum_i \alpha_i P_i$ and
    $\gamma_i = f(P_i) \in \CQSTATE[M]$.
\end{definition}

A Pauli vector $v$ can be lifted to a Pauli map
$\conjugate{v}(P) = v P v^\dagger$, called
the conjugation action of $v$. Scaling ($\alpha m$) and addition ($v_1 + v_2$) of Pauli maps are defined pointwise, and we say a Pauli map
\emph{acts on} a Pauli vector by mapping it over every Pauli in the vector:
$
    m(\sum_i \alpha_i P_i) = \sum_i \alpha_i m(P_i)
$.
%We can then define composition as $m_2 \circ m_1(v) = m_2(m_1(v))$.

\begin{ExtendedVersion}
Note that conjugation distributes over composition:
    \begin{align}\label{lem:ConjugateDist}
        \conjugate{v}(\conjugate{v'}(P))
        &= \conjugate*{v' v}(P).
    \end{align}

\begin{definition}
A \emph{monoidal} Pauli map is a map $f$ that respects the Pauli multipliciative
identities:
\begin{align}
    f(I) &= I \\
    f(P_1 \cdot P_2) &= f(P_1) \cdot f(P_2)
\end{align}
\end{definition}

\begin{lemma} \label{lem:MonoidalMapsNatural}
    Monoidal Pauli maps respect multiplication of Pauli vectors:
    $
        f(v_1 \cdot v_2) = f(v_1) \cdot f(v_2)
    $.
\end{lemma}
\begin{proof}
    {\small \begin{align}
        &f\left(
            \left( \sum_i \alpha_i P_i \right)
            \cdot
            \left( \sum_j \alpha_j' P_j' \right)
        \right)
        = f\left( \sum_{i,j} \alpha_i \alpha_j' P_i P_j' \right) \\
        &= \left( \sum_{i,j} \alpha_i \alpha_j' f\left(P_i P_j'\right) \right)
        = \left( \sum_{i,j} \alpha_i \alpha_j' f\left(P_i\right) m\left(P_j'\right) \right) \\
        &= \left( \sum_i \alpha_i f(P_i) \right)
          \cdot
          \left( \sum_j \alpha_j' f(P_j') \right)
    \end{align} }%
\end{proof}

\begin{lemma} \label{lem:conjugateMonoidal}
    If $v$ is unitary, then $\conjugate{v}$ is a monoidal Pauli map.
\end{lemma}
\begin{proof}
\begin{align}
    \conjugate{v}(I) &= v I v^\dagger = I \\
    \conjugate{v}(P_1 P_2) &= v P_1 P_2 v^\dagger
        = v P_1 v^\dagger v P_2 v^\dagger = \conjugate{v}(P_1) \cdot \conjugate{v}(P_2)
\end{align}
\end{proof}
\end{ExtendedVersion}

\section{\PCAST data types}
\label{sec:structures}
  \PCAST uses five types of nodes for different sorts of gates:
\begin{itemize}[left=5pt]
    \item Pauli rotations $\Rot{P}{\theta}$ for non-Clifford unitaries $e^{-i \theta/2 P}$
    \item Pauli preparations $\Prep{P_1}{P_2}$
    \item Pauli measurements $\Meas[c]{P}$
    \item Measurement space functions $\mu$
    \item Pauli frames $F$ for Clifford unitaries
\end{itemize}
The semantics of \PCAST nodes are defined as Pauli maps $\interpM{n} : \Pauli_k \rightarrow \CQSTATE$, except for measurement space functions, which are defined as cq-channels directly.
%The semantics of a \PCAST node is given as a cq-channel $\interpM{n} : \CQSTATE \rightarrow \CQSTATE$,
%though rotations, measurements, and preparations can be more directly described in terms of Pauli maps. 

\begin{definition} \label{defn:rot}
A Pauli rotation $\Rot{P}{\theta}$ consists of a Hermitian Pauli $P$ on $k$
qubits and a real number $\theta \in \mathbb{R}$. The semantics of a Pauli
rotation $\interpM{\Rot{P}{\theta}}$ is given by the conjugation action of its corresponding unitary, $\interp{\Rot{P}{\theta}}$:\footnote{Recall, if the classical state is not specified, it is assumed to be the constant cq-state $\cqstate{\_}{\rho}$.}
  \begin{align}
    \interpM{\Rot{P}{\theta}}(Q) &= \interp{\Rot{P}{\theta}} Q \interp{\Rot{P}{-\theta}} \\
    \interp{\Rot{P}{\theta}} &= e^{-i \theta / 2 P}
    = \cosBy2{\theta} I + i \sinBy2{\theta} P
  \end{align}
\end{definition}

% \citet{Zhang2019} define an optimization technique that commutes Cliffords past rotations of the form $\Rot{P}{\tfrac{\pi}{2}}$, since they were primarily concerned with the Clifford + T gateset. In this work we are focused on NISQ architectures, many of which support arbitrary single-qubit rotations. \todo{Move this discussion to introduction?} In addition, \PCAST optimizations support not just unitaries, but also preparations and measurements. Generalizing these components by their Pauli axis as well, we can use the same technique to describe how Cliffords commute past them.

\begin{definition} \label{defn:prep}
    A Pauli preparation
    $\Prep{P_Z}{P_X}$ consists of a pair of non-commutative Hermitian Paulis. 
    \begin{align}        
        &\interpM{\Prep{P_Z}{P_X}}(Q) 
        \nonumber \\
        &=
        \frac{1}{4} \left( (I + P_Z)Q(I + P_Z) + P_X(I - P_Z)Q(I-P_Z)P_X \right)
        \nonumber \\
        &= \begin{cases}
            Q (I + P_Z) & \commutativity{Q}{P_Z} = \commutativity{Q}{P_X} = 0 \\
            0           & \text{otherwise}
        \end{cases}
    \end{align}
    %where $\conjugate{v}(Q) = v Q v^\dagger$ corresponds to conjugation by
    %a unitary $v$. \albert{I don't think you want to define this for only unitaries}
\end{definition}

Intuitively, $\Prep{P_Z}{P_X}$ collapses the state in the eigen-subspaces of $P_Z$ and, if -1 is obtained, applies $P_X$.
As an example, $\Prep{Z_i}{X_i}$ prepares the $i$th qubit in the $Z$ basis.

\begin{definition} \label{defn:meas}
    A Pauli measurement $\Meas[c]{P}$ 
    consists of a Hermitian Pauli acting as a measurement operator.
    \begin{align}
        &\interpM{\Meas[c]{P}}(Q)
        = m \mapsto \begin{cases}
        \tfrac{1}{4} (I + P) Q (I + P)
            & \cLookup{m}{c}=0 \\
        \tfrac{1}{4} (I - P) Q (I - P)
            & \cLookup{m}{c}=1 
        \end{cases} \nonumber \\
        &\qquad= \cAssign{c}{b} \mapsto \begin{cases}
        \tfrac{1}{2} Q (I + (-1)^b P)
            & \commutativity{P}{Q}=0 \\
        % \tfrac{1}{2} Q (I + P)
        %     & \commutativity{P}{Q}=0 \wedge \cLookup{m}{c}=0 \\
        % \tfrac{1}{2} Q (I - P)
        %     & \commutativity{P}{Q}=0 \wedge \cLookup{m}{c}=1 \\
        0
            & \commutativity{P}{Q}=1
        \end{cases}
    \end{align}
\end{definition}

As an example, $\Meas[c]{Z_i}$ measures qubit $i$ in the $Z$ basis.
$c$ is the classical register in which measurement is
recorded.
%\footnote{\jennifer{Instead of having $c$ as a list of registers, can we just use measurement space functions?}}
%We write $c \plusplus c'$ for the list concatenation of $c$ with $c'$.

\begin{definition}
    A measurement space function $\mu : \mathcal{M}_1 \rightarrow \mathcal{M}_2$ is a function between two measurement spaces. It is interpreted as a cq-channel $\interpM{\mu} : \CQSTATE[\mathcal{M}_1] \rightarrow \CQSTATE[\mathcal{M}_2]$ as follows:
    \begin{align}
        \interpM{\mu}(\gamma) &= m_2 \mapsto \sum_{m_1 \in \mu^{-1}(m_2)} \gamma_{m_1}.
    \end{align}
\end{definition}

As an example, in the circuit $\Prep{Z_i}{X_i}; \Meas[c]{Z_i}$, the measurement can be optimized away, yielding $\Prep{Z_i}{X_i}; \mu$
% \begin{align}
%     \Prep{P_Z}{P_X}; \Meas[c]{P_Z} &\equiv \Prep{P_Z}{P_X}; \mu
% \end{align}
where $\mu(m)=m; \cAssign{c}{0}$.%
\begin{ExtendedVersion}
\footnote{For conciseness, we sometimes write $\mu=\cAssign{c_0}{b_0};\cdots;\cAssign{c_{n-1}}{b_{n-1}}$ instead of $\mu(m)=m; \cAssign{c_0}{b_0};\cdots;\cAssign{c_{n-1}}{b_{n-1}}$, for example in \cref{fig:demo}.}
\end{ExtendedVersion}

  \subsection{Pauli Frames} \label{sec:frame}

 Clifford unitaries satisfy the property that when acting on a Pauli $P$ by conjugation, the result $U P U^\dagger$ is still a Pauli.
A \emph{Pauli frame}\footnotemark\xspace is a compact representation of a Clifford unitary defined by that conjugation---or more precisely, its inverse conjugation $U^\dagger P U$---on every base Pauli string $Z_j$, $X_j$, and $Y_j$.
%The conjugation action of the Clifford respects multiplication, in that $ U^\dagger (P_1 P_2) U = (U^\dagger P_1 U) (U^\dagger P_2 U) $, and identity, in that $U^\dagger I U = I$.
%This means that to specify how it acts on a Pauli $\alpha (p_0,\ldots,p_{n-1})$, it suffices to show how it acts on each single-qubit Pauli at index $j$. 
For example, the inverse conjugation of $U=\CNOT_{0,1}$ is
\begin{align} \begin{array}{c | c c c}
    j & U^\dagger Z_j U & U^\dagger X_j U & U^\dagger Y_j U  \\
    \hline
    0 & Z_0 & X_0 X_1 & Y_0 X_1 \\
    1 & Z_0 Z_1 & X_1 & Z_0 Y_1
\end{array} \end{align}
It suffices to store only the first two columns of this table: we can derive $U^\dagger Y_j U=-i(U^\dagger Z_j U) (U^\dagger X_j U)$ since $Y=-i Z X$.
%Because $Y = -i Z X$, we can derive $U^\dagger Y_j U$ as $-i (U^\dagger Z_j U) (U^\dagger X_j U)$; so it suffices to store only the first two columns of this table.
Thus, an $n$-qubit Clifford is represented by an $n\times2$ Pauli frame, where the first column stores $U^\dagger Z_j U$ and the second column stores $U^\dagger X_j U$:
\begin{align}
    \begin{pmatrix}
        Z_0 & X_0 X_1 \\
        Z_0 Z_1 & X_1
    \end{pmatrix}
\end{align}

\footnotetext{Pauli tableaus~\cite{Aaronson2004} were first introduced as a way to simulate stabilizer states generated entirely from Clifford gates and single-qubit measurements. Since then, Pauli tableaus have been used to represent Clifford circuits in general, not just for the purposes of stabilizer simulation. Following \citep{Schmitz2021,Schmitz2023PCOAST}, in this work we refer to Pauli tableaus as \emph{Pauli frames} to emphasize their linear algebraic structure with regards to the Pauli group.}

Note that whether we store $U^\dagger P U$ or $U P U^\dagger$ in the entries of the Pauli frame is a matter of style---the inverse of a Pauli frame can always be calculated to obtain one from the other~\citep{PaykinSchmitz2023PCOAST}.
However, the choice affects the efficiency of the lookup operation $\fwdAction{F}$ (\cref{defn:frameaction} below). For efficiency in \PCAST, the lookup operation should implement the opposite of the semantic interpretation $\interpM{F}(Q)=U Q U^\dagger$ (\cref{defn:framesemantics}). 
%Thus these two interpretations will correspond to opposite conjugations:
%\begin{align}
%    \interpM{F}(Q) &= U Q U^\dagger 
%    \qquad\text{and}\qquad
%    \fwdAction{F}(Q) = U^\dagger Q U
%\end{align}

\begin{definition}
    A \emph{Pauli frame} $F$ on $k$ qubits is a $k \times 2$ array of Hermitian $k$-qubit Paulis:
    \begin{align}
        F = \begin{pmatrix}
            \effZ_0 & \effX_0 \\
            %\effZ_1 & \effX_1 \\
            \vdots & \vdots \\
            \effZ_{n-1} & \effX_{n-1}
        \end{pmatrix}
    \end{align}
    The arguments in the first column are called the \emph{effective $Z_j$} Paulis,
    and the arguments in the second column the \emph{effective $X_j$} Paulis,
    and they must respect all the same commutativity relations as the corresponding $Z_j$ and $X_j$ Paulis:
    %$\effZ_j$ anticommutes with $\effX_j$, but commutes with all other Paulis in the frame, and vice versa.
    \begin{align}
        \commutativity{\effZ_i}{\effZ_j} = \commutativity{\effX_i}{\effX_j}
        &= 0 \\
        \commutativity{\effZ_i}{\effX_j} = \commutativity{\effX_j}{\effZ_i}
        &= \delta_{i,j}\footnotemark 
    \end{align}
    \end{definition}

    \footnotetext{$\delta_{i,j}$ is $0$ if $i=j$ and $1$ otherwise.}

    \begin{definition} \label{defn:frameaction}
    The \emph{lookup action} of $F$ on $P=\alpha(p_0,\ldots,p_{k-1})$, written $\fwdAction{F}(P)$, is the product of the effective entry of each $p_j$:
    \begin{align}
        \fwdAction{F}(\alpha (p_0,\ldots,p_{n-1})) = \alpha \prod_{j} \eff{p_j}
    \end{align}
    where $\eff{I_j} = I$ and $\eff{Y_j} = \effZ_j \hmult \effX_j$.
    \end{definition}

%We can show that every Pauli frame does in fact correspond to a Clifford unitary.
\begin{lemma} \label{lem:FrameUnitary}
    For every Pauli frame $F$ there is a Clifford unitary $\unitaryof{F}$, unique up to
    overall phase, satisfying $\fwdAction{F}(P) = (\unitaryof{F})^\dagger P \unitaryof{F}$ for
    any Pauli $P$.
\end{lemma}
\begin{ExtendedVersion}
\begin{proof}
    To show that there exists such a unitary, it suffices to show that
    $\fwdAction{F}(\ket{i}\bra{i})$ is a pure state $\ket{\phi}\bra{\phi}$, in which case we can define $U(\ket{i})=\ket{\phi}$. To prove it is a pure state, it suffices to show that $\fwdAction{F}(\ket{i}\bra{i})\fwdAction{F}(\ket{i}\bra{i})=I$.
    From the proof of \cref{lem:PauliVectorDecomposition}, we can observe that
    \begin{align}
        \ket{i}\bra{i} = (I \pm Z_0) \cdots (I \pm Z_{n-1})
    \end{align}
    and thus
    \begin{align}
        \fwdAction{F}(\ket{i}\bra{i}) =  (I \pm \effZ^F_0) \cdots (I \pm \effZ^F_{n-1})
    \end{align}
    where $i \neq j$ implies $\commute{\effZ^F_i}{\effZ^F_j}$. Thus
    \begin{align}
        &\fwdAction{F}(\ket{i}\bra{i})\fwdAction{F}(\ket{i}\bra{i}) \\
        %&= \left( (I \pm \effZ^F_0) \cdots (I \pm \effZ^F_{n-1}) \right)
        %    \left( (I \pm \effZ^F_0) \cdots (I \pm \effZ^F_{n-1}) \right) \\
        &= (I \pm \effZ^F_0) (I \pm \effZ^F_0) \cdots (I \pm \effZ^F_{n-1}) (I \pm \effZ^F_{n-1}) \\
        &= I 
                 \qedhere
    \end{align}
\end{proof}
\end{ExtendedVersion}

\begin{ExtendedVersion}
\begin{lemma} \label{conj:FrameUnitaryCompleteness}
    If $\unitaryof{F_1} = \unitaryof{F_2}$ then $F_1 = F_2$.
\end{lemma}
\begin{proof}
    For each $j$, we have
    \begin{align}
        \eff{p}_j^{F_1} &= (\unitaryof{F_1})^\dagger p_j (\unitaryof{F_1})
        =  (\unitaryof{F_2})^\dagger p_j (\unitaryof{F_2}) = \eff{p}_j^{F_2}
    \end{align}
\end{proof}
\end{ExtendedVersion}

\begin{definition} \label{defn:framesemantics}
    The semantics of a Pauli frame is given by the Pauli map
    $\interpM{F}(Q) = \unitaryof{F} Q (\unitaryof{F})^\dagger$.
\end{definition}

% \begin{ExtendedVersion}
% \begin{lemma}
%     For all Paulis $P$, $F^0(P)=P$. As a corollary, $\unitaryof{F^0} = I$.
% \end{lemma}
% \begin{proof}
% Let $P=(p_0,\ldots,p_{n-1})$ and let $F^0(P)=P_0' \cdot \ldots \cdot P_{n-1}'$. If
% $p_i=X$ then $p_i' = \effX^F_i = X_i$, and in general, $p_i' = p_i$, that is,
% the Pauli with the single support $i$. Therefore $P_0' \cdots P_{n-1}'$ is just
% $P$.
% \end{proof}
% \end{ExtendedVersion}

\begin{ExtendedVersion}
\subsubsection{Cliffords}
\end{ExtendedVersion}

%\cref{fig:CliffordFrame} shows the Pauli frames corresponding to the unitaries $H$, $\CNOT$, and $S$; in general, 
The frame $\frameof{U}$ associated with a Clifford unitary $U$ is
\begin{align}
        %\effZ^{\frameof{U}}_i &= U^\dagger Z_i U
    %&
        %\effX^{\frameof{U}}_i &= U^\dagger X_i U.
    \eff{p_j}^{\frameof{U}} = U^\dagger p_j U,
\end{align}
where $\eff{p_j}^F$ is entry of $F$ corresponding to $p_j$.
We refer to $\frameof{I}$ as the \emph{origin frame}.

\begin{ExtendedVersion}
\begin{lemma}
    $\frameof{U}$ is a well-formed Pauli frame on $n$ qubits
    and $\unitaryof{\frameof{U}} = U$ up to a constant phase.
\end{lemma}
\begin{proof}
    To show $\frameof{U}$ is well-formed, it suffices to show that, for all entries $\eff{p}_i$ and $\eff{q}_j$ in the frame, $\commutativity{\eff{p}_i}{\eff{q}_j}=\commutativity{p_i}{q_j}$:
    \begin{align}
        \eff{p}_i \eff{q}_j &= U^\dagger p_i U U^\dagger q_j U = U^\dagger p_i q_j U \\
        &= U^\dagger (-1)^{\commutativity{p_i}{q_j}} q_j p_i U
        = (-1)^{\commutativity{p_i}{q_j}} \eff{q}_j \eff{p}_i
    \end{align}
    Since $\unitaryof{\frameof{U}}$ satisfies $\fwdAction{F}(P)=U^\dagger P U$ by \cref{lem:FrameUnitary}, it is equivalent to $U$up to an overall phase.
\end{proof}
\end{ExtendedVersion}

\begin{ExtendedVersion}
\subsubsection{Commutativity}
\end{ExtendedVersion}

A consequence of the commutativity rules for frames is that the lookup action preserves commutativity and composition.
\begin{ShortVersion}
\begin{align}
    \commutativity{\fwdAction{F}(P)}{\fwdAction{F}(Q)} &= \commutativity{P}{Q}
        \label{lem:FrameCommutativityPreservation} \\
    \fwdAction{F}(P Q) &= \fwdAction{F}(P) \fwdAction{F}(Q)
        \label{lem:FrameCompositionNatural}
\end{align}
\end{ShortVersion}

\begin{ExtendedVersion}
\begin{lemma} \label{lem:FrameCompositionNatural}
    $\fwdAction{F}(P_2 \cdot P_1) = \fwdAction{F}(P_2) \cdot \fwdAction{F}(P_1)$
\end{lemma}
\begin{proof}
    Let $U=\unitaryof{F}$. By \cref{lem:FrameUnitary},
    \begin{align} 
        \fwdAction{F}(P_2 P_1)
        &= U^\dagger P_2 P_1 U \\
        &= U^\dagger P_2 U U^\dagger P_1 U
        = \fwdAction{F}(P_2) \fwdAction{F}(P_1). 
    \end{align}
    % For $i=0,1$, let $P_i = \alpha_i (p_0^i,\ldots,p_{n-1}^i)$, meaning that
    % $P_2 \cdot P_1 = \alpha_2 \alpha_1 \sigma_0 \cdots \sigma_{n-1} (q_0,\ldots,q_{n-1})$
    % where $p_j^2 p_j^1 = \sigma_j q_j$.
    % Let us write out $\fwdAction{F}(P_i) = \alpha_i Q_0^i \cdots Q_{n-1}^i$ where $Q^i_j = \eff{p^i_j}$,
    % and similarly $\fwdAction{F}(P_2 \cdot P_1) = \alpha_2 \alpha_1 \sigma_0 \cdots \sigma_{n-1} R_0 \cdots R_{n-1}$ where $R_j = \eff{q}_j$.
    % By the commutativity rules of Pauli frames, we know that for $j \neq k$,
    % we have $\commute{Q_j^i}{Q_k^i}$, in which case we can rewrite
    % \begin{align}
    %     \fwdAction{F}(P_2) \cdot \fwdAction{F}(P_1) &= \alpha_2 \alpha_1 Q_0^2 Q_0^1 \cdots Q_{n-1}^2 Q_{n-1}^1
    % \end{align}
    % Now it suffices to check that for each $j$ that $Q_j^2 Q_j^1 = \sigma_j R_j$,
    % which can be done by case analysis on $p_j^2$ and $p_j^1$.
\end{proof}
\end{ExtendedVersion}

\begin{ExtendedVersion}
\begin{lemma} \label{lem:FrameCommutativityPreservation}
Let $F$ be a Pauli frame over $n$ qubits and $P^1, P^2$ be $n$-qubit
Paulis. Then $\commutativity{P^1}{P^2} = \commutativity{\fwdAction{F}(P^1)}{\fwdAction{F}(P^2)}$.
\end{lemma}
\begin{proof}
    %Let $P_1=\alpha_1 (p_0,\ldots,p_{n-1})$ and $P_2 = \alpha_2 (q_0,\ldots,q_{n-1})$.
    Notice that for $i \neq j$, $\commute{\eff{p}_i}{\eff{q}_j}$.
    Then
    \begin{align}
        \fwdAction{F}(P_1) \fwdAction{F}(P_2)
        &= U^\dagger P_1 U U^\dagger P_2 U
        = U^\dagger P_1 P_2 U \\
        &= (-1)^{\commutativity{P_1}{P_2}} U^\dagger P_2 P_1 U \\
        &= (-1)^{\commutativity{P_1}{P_2}} \fwdAction{F}(P_2) \fwdAction{F}(P_1)
        % &\fwdAction{F}(P_1) \fwdAction{F}(P_2) \notag \\
        % &= \left( \alpha_1 \eff{p}_0 \hmult \cdots \hmult \eff{p}_{n-1}\right)
        %    \left( \alpha_2 \eff{q}_0 \hmult \cdots \hmult \eff{q}_{n-1}\right) \\
        % &= \alpha_1 \alpha_2 \prod_{j=0}^{n-1} \eff{p}_j \hmult \eff{q}_j \\
        % &= \alpha_1 \alpha_2 \prod_{j=0}^{n-1} (-1)^{\commutativity{p_j}{q_j}} \eff{q}_j \hmult \eff{p}_j \\
        % &= \left( \prod_{j=0}^{n-1} (-1)^{\commutativity{p_j}{q_j}}\right) \fwdAction{F}(P_2) \fwdAction{F}(P_1) \\
        % &= (-1)^{\sum_{j=0}^{n-1} \commutativity{p_j}{q_j}} \fwdAction{F}(P_2) \fwdAction{F}(P_1)
        %     \\
        % &= (-1)^{\lambda(P_1,P_2)} \fwdAction{F}(P_2) \fwdAction{F}(P_1)
    \end{align}
\end{proof}
\end{ExtendedVersion}

\begin{ExtendedVersion}
Note that Pauli frames will always map Hermitian Paulis to Hermitian Paulis.
    \begin{align} \label{lem:FrameHermitian}
        \fwdAction{F}(P) \fwdAction{F}(P)
        &= \fwdAction{F}(P P) = \fwdAction{F}(I) = I.
    \end{align}
    % We know $\fwdAction{F}(P) = \alpha \eff{p}_0 \hmult \cdots \hmult \eff{p}_{n-1}'$ where $\commute{\eff{p}_i'}{\eff{p}_j'}$
    % for $i \neq j$. Therefore by \cref{lem:HermitianPauliFacts} 
    % it suffices to show that each $\eff{p}_i'$ is itself Hermitian, which is
    % true by construction.
\end{ExtendedVersion}

\begin{ExtendedVersion}
    \subsubsection{Composition of Pauli frames}

    Pauli frames can be composed with each other by mapping the action of the first over each element of the second.
\end{ExtendedVersion}

\begin{definition}
    Composition of Pauli frames $F_2 \circ F_1$ is defined as
    $
        \eff{p}^{F_2 \circ F_1}_i = \fwdAction{F_1}(\eff{p}^{F_2}_i)
        %\effX^{F_2 \circ F_1}_i &= \fwdAction{F_1}(\effX^{F_2}_i) 
    $.
    \ifExtended{}{It satisfies $\unitaryof{F_2 \circ F_1} = \unitaryof{F_2}\unitaryof{F_1}$.}
\end{definition}

\begin{ExtendedVersion}
\begin{lemma} \label{lem:FrameCompositionExpanded}
    $\unitaryof{F_2 \circ F_1}
    = \unitaryof{F_2} \unitaryof{F_1}$.
\end{lemma}
\begin{proof}
It suffices to show that, for every Pauli $P$, we have 
$\fwdAction{F_2 \circ F_1}(P)=\fwdAction{F_1}(\fwdAction{F_2}(P))$.

Let $P=\alpha(p_0,\ldots,p_{n-1})$ and $\fwdAction{F_2}(P) = \alpha \eff{p}^{F_2}_0 \hmult \cdots \hmult \eff{p}^{F_2}_{n-1}$. By \cref{lem:FrameCompositionNatural}, we know
that $\fwdAction{F_1}(\fwdAction{F_2}(P)) = \alpha \fwdAction{F_1}(\eff{p}^{F_2}_0) \hmult \cdots \hmult \fwdAction{F_1}(\eff{p}^{F_2}_{n-1})$.

Similarly, let $\fwdAction{(F_2 \circ F_1)}(P) = \alpha \eff{p}^{F_2 \circ F_1}_0 \hmult \cdots \hmult \eff{p}^{F_2 \circ F_1}_{n-1}$; the result follows from the fact that $\eff{p}^{F_2 \circ F_1}_j = \fwdAction{F_1}(\eff{p}^{F_2}_j)$.
\end{proof}
\end{ExtendedVersion}

  \subsection{\PCAST Terms} \label{sec:PoPRMonoid}
    A \emph{\PCAST term} $t$ is a sequence of \PCAST nodes $n_0;\ldots;n_{k-1}$ with $1$ indicating the empty sequence.
The semantics of nodes can be lifted to terms $\interpM{t} : \CQSTATE \rightarrow \CQSTATE$.
\begin{ExtendedVersion}
\begin{align}
    \interpM{1}(\gamma) &= \gamma \\
    \interpM{t_1; t_2}(\gamma) &= \interpM{t_2}(\interpM{t_1}(\gamma))
\end{align}
\end{ExtendedVersion}

\begin{ExtendedVersion}
Because the action of the \PCAST nodes' semantics as Pauli maps
is linear on their Paulis, they can be described by how they operate on
arbitrary Pauli vectors.\footnote{This applies to all \PCAST nodes except Pauli frames and measurement space functions.}
\begin{lemma} \label{lem:PauliNodeSubstitution}
    \begin{align}
        &\interpM{\Rot{P}{\theta}}(\gamma) 
            = m \mapsto \interp{\Rot{P}{\theta}} \gamma(m) \interp{\Rot{P}{-\theta}} 
            \label{eqn:PauliNodeSubstitutionRot}\\
        &\interpM{\Prep{P_Z}{P_X}}(\gamma) = m \mapsto \\
            &(I + P_Z) \gamma(m) (I + P_Z) + P_X (I - P_Z) \gamma(m) (I - P_Z) P_X 
            \nonumber \\
        &\interpM{\Meas[c]{P}}(\gamma) = (m, \cAssign{c}{b}) \mapsto 
            \label{eqn:PauliNodeSubstituionMeas} \\
            &\begin{cases}
                \frac{1}{4} (I + P) \gamma(m) (I + P) & b=0 \\
                \frac{1}{4} (I - P) \gamma(m) (I - P) & b=1
            \end{cases} \nonumber
    \end{align}
\end{lemma}
\begin{proof}
    The semantics of these three nodes are given as Pauli vectors, and are lifted to cq-channels as in \cref{defn:PauliMap}. For example, in the case of rotations,
    \begin{align}
        \interpM{\Rot{P}{\theta}}(\gamma) &= m \mapsto 
            \sum_i \alpha_i \interpM{\Rot{P}{\theta}}(P_i)(m) \\
        &=  m \mapsto 
            \sum_i \alpha_i \interp{\Rot{P}{\theta}}(P_i)\interp{\Rot{P}{-\theta}}
    \end{align}
    where $\gamma(m) = \sum_i \alpha_i P_i$. But by distributivity of matrix multiplication and addition, this can be written as
    \begin{align}
        m \mapsto \interp{\Rot{P}{\theta}} \left( \sum_i \alpha_i P_i\right) \interp{\Rot{P}{-\theta}}
    \end{align}
    which is exactly \cref{eqn:PauliNodeSubstitutionRot}.

    Similarly, for measurement the definition as a cq-channel depends on its definition as a Pauli map from \cref{defn:prep}:
    \begin{align}
        \interpM{\Meas[c]{P}}(\gamma) &= (m, \cAssign{c}{b}) \mapsto
            \sum_i \alpha_i \interpM{\Meas[c]{P}}(P_i)(\cAssign{c}{b}) \\
        &= (m, \cAssign{c}{b}) \mapsto \sum_i \alpha_i (I \pm P) P_i (I \pm P)
    \end{align}
    where $\gamma(m) = \sum_i \alpha_i P_i$. Note that the sign of $(I \pm P)$ depends only on the value of $b$, so it doesn't change across the terms of the summation. Therefore it can be rewritten as
    \begin{align}
        (m,\cAssign{c}{b}) \mapsto (I \pm P) \left( \sum_i \alpha_i P_i \right) (I \pm P)
    \end{align}
    as in \cref{eqn:PauliNodeSubstituionMeas}.

    The same reasoning holds for $\Prep{P_Z}{P_X}$.
\end{proof}
\end{ExtendedVersion}

\begin{ExtendedVersion}
\subsubsection{Equivalence relation}
\end{ExtendedVersion}

We define an equivalence relation on \PCAST terms parameterized by a \emph{hold} or \emph{release} outcome $o$, corresponding to the two equivalence relations on cq-states.
\begin{align}
    t_1 \equiv^o t_2 
    &\iff
    \forall \gamma,~ \interpM{t_1}(\gamma) \equiv^o \interpM{t_2}(\gamma)
\end{align}
If not specified, we assume $o$ is the stronger \hold outcome.
%\jennifer{TODO: note that $\equiv^r$ only respects composition on the right, not left, but this is okay because of how the release outcome is used.}

\begin{ExtendedVersion}
It is trivial to see that sequential composition $t_1; t_2$ forms a monoid with
$1$ as the unit.
\begin{lemma}
    \begin{align}
        1; t &\equiv t;1 \equiv t \\
        (t_1 ; t_2); t_3 &\equiv t_1 ; (t_2 ; t_3)
    \end{align}
\end{lemma}
\end{ExtendedVersion}

\begin{ExtendedVersion}
Release equivalence only holds on the right-hand-side of composition, which is appropriate because the release outcome specifies the behavior at the end of a circuit.
    \begin{lemma}
        If $t_1 \equiv^\hold t_2$ and $t_1' \equiv^o t_2'$, then $t_1;t_1' \equiv^o t_2;t_2'$.
    \end{lemma}
    \begin{proof}
        We want to show that for all $\gamma$, 
        \begin{align}
            \interpM{t_1'}(\interpM{t_1}(\gamma)) \equiv^o \interpM{t_2'}(\interpM{t_2}(\gamma)).
        \end{align}
        By the induction hypothesis for $t_1' \equiv^o t_2'$, it suffices to show that
        $\interpM{t_1}(\gamma) = \interpM{t_2}(\gamma)$. But for two functions to be equal, it suffices to show they are the same on every input, which is the precise definition of $t_1 \equiv^\hold t_2$.
    \end{proof}
\end{ExtendedVersion}

\begin{ExtendedVersion}
    \subsubsection{Commutativity}
\end{ExtendedVersion}

Intuitively, two \PCAST nodes commute exactly when their underlying
Paulis commute. Formally, \cref{fig:CommuteTermTerm} defines a
commutativity relation $\commute{n}{n'}$ between nodes, aided by a helper relation
$\commute{Q}{n}$ (\cref{fig:CommutePauliNode}) that indicates when a Pauli commutes with a \PCAST node.

\begin{figure}
{\small \begin{align*} \begin{array}{c}
    \inferrule*%[right=Rot-\commuteTTRule]
        {\commute{P}{n}}
        {\commute{\Rot{P}{\theta}}{n}}
    \quad
    \inferrule*%[right=Prep-\commuteTTRule]
        {\commute{P_1}{n} \\ \commute{P_2}{n}}
        {\commute{\Prep{P_1}{P_2}}{n}}
    \quad
    \inferrule*%[right=Meas-\commuteTTRule]
        {\commute{P}{n}}
        {\commute{\Meas[c]{P}}{n}}
    \\ \\
    \inferrule*
        {F \circ F' = F' \circ F}
        {\commute{F}{F'}}
    \quad
    \inferrule*%[right=F-\commuteTTRule]
        {n~\text{not a Pauli frame} \\ \commute{n}{F}}
        {\commute{F}{n}}
    \\
    \inferrule*
        {~}
        {\commute{\mu}{\Rot{P}{\theta}}}
    \quad
    \inferrule*
        {~}
        {\commute{\mu}{F}}
    \quad
    \inferrule*
        {~}
        {\commute{\mu}{\Prep{P_1}{P_2}}}
    % \\ \\
    % \text{where}~\fv(n) = \begin{cases}
    %     \{c\} & n = \Meas[c]{P} \\
    %     \{c \mid \cLookup{m}{c} \neq \cLookup{\mu(m)}{c} \}
    %         & n = \mu \\
    %     \emptyset &\text{otherwise}
    % \end{cases}
\end{array} \end{align*}}%
\caption{Definition of $\commute{n_1}{n_2}$ indicating when two \PCAST nodes commute.
These definitions, given as inference rules, are understood as follows: the conclusion below the line holds if and only if all the hypotheses above the line hold.
\ifExtended{The bottom three rules are vacuously true as they have no hypotheses.}{}
%In the final rule, $\fv(n)$ is the set of classical variables are affected by $n$.
}
\label{fig:CommuteTermTerm}
\end{figure}

\begin{figure}
{\small \begin{align*} \begin{array}{c}
    \inferrule*%[right=Rot-\commuteTPRule]
        {\commute{Q}{P}}
        {\commute{Q}{\Rot{P}{\theta}}}
    \qquad
    \inferrule*%[right=F-\commuteTPRule]
        {\fwdAction{F}(Q)=Q}
        {\commute{Q}{F}}
    \\ \\
    \inferrule*%[right=Prep-\commuteTPRule]
        {\commute{Q}{P_1} \\ \commute{Q}{P_2}}
        {\commute{Q}{\Prep{P_1}{P_2}}}
    \quad
    \inferrule*%[right=Meas-\commuteTPRule]
        {\commute{Q}{P}}
        {\commute{Q}{\Meas[c]{P}}}
    \quad
    \inferrule*
        {~}
        {\commute{Q}{\mu}}
\end{array} \end{align*} }%
\caption{Definition of $\commute{Q}{n}$, when a Pauli $Q$ commutes with a \PCAST node $n$. Recall that for two Paulis, we write $\commute{P_1}{P_2}$ if and only if
$\commutativity{P_1}{P_2} = 0$.}
\label{fig:CommutePauliNode}
\end{figure}

\begin{ExtendedVersion}
A unitary node $n$ is either a rotation or Pauli frame, and unitary nodes
distribute over multiplication of Pauli vectors. Note that unitary nodes are all monoidal (\cref{lem:conjugateMonoidal}).

\begin{lemma} \label{lem:unitaryNodeCommutePauli}
    If $n$ is unitary and $\commute{Q}{n}$, then $\interpM{n}(Q)=Q$.
\end{lemma}
\begin{proof}
    By \cref{fig:CommutePauliNode} for frames and \cref{defn:rot} for rotations.
\end{proof}
\end{ExtendedVersion}

\begin{theorem} \label{thm:PoPRNodeCommuteCorrect}
    If $\commute{n_1}{n_2}$ then $n_1; n_2 \equiv n_2; n_1$.\footnotemark
\end{theorem}
\footnotetext{Note that this property does not hold in the other direction; the rules for $\commute{n_1}{n_2}$ are strictly more
restrictive. In particular, the relation is not reflexive on $\Prep{P_Z}{P_X}$ since $\anticommute{P_Z}{P_X}$.
}
\begin{ExtendedVersion}
\begin{proof}
    By induction on the derivation of $\commute{n_1}{n_2}$.

    \proofpart[$n_1=\Prep{P_Z^1}{P_X^1}$ and $n_2=\Prep{P_Z^2}{P_X^2}$] \label{case:PrepPrep}
    Let $n_1=\Prep{P_Z^1}{P_X^1}$ such that
    $\commute{P_Z^1}{n_2}$ and $\commute{P_X^1}{n_2}$, and
    let $n_2$ be $\Prep{P_Z^2}{P_X^2}$, in which it
    will be the case that
    $P_Z^1$ and $P_X^1$ each commute with both $P_Z^2$ and $P_X^2$.
    Observe that $\interpM{n_1; n_2}(Q)$ can be written compactly as $\tfrac{1}{16} \sum_{b_1,b_2 \in \{0,1\}} \conjugate*{v_1 v_2}(Q)$, where
    \begin{align} \label{eqn:PrepTermCompact}
        v_j = (P_X^j)^{b_j} (I \pm^{b_j} P_Z^j)
    \end{align}
    and where $v \pm^b v' = v + (-1)^b v'$.
    But because of the commutativity constraints, $v_1 v_2 = v_2 v_1$, and so
    $n_1;n_2 \equiv n_2;n_1$.

    \proofpart[$n_1=\Prep{P_Z}{P_X}$ and $n_2=\Meas{P}$] \label{case:PrepMeas}

    Next assume $n_1=\Prep{P_Z}{P_X}$ and $n_2=\Meas[c]{P}$, in which case
    $\commute{P_Z}{P}$ and $\commute{P_X}{P}$. Then 
    \begin{align} \label{eqn:InternalPrepMeas}
        \interpM{n_1}\left( \interpM{n_2}(Q) \right)
        &= \cAssign{c}{b_2} \mapsto \sum_{b_1 \in \{0,1\}} v_1 (I \pm^{b_2} P) Q (I \pm^{b_2} P) v_1
    \end{align}
    where $v_1$ is defined in \cref{eqn:PrepTermCompact}. But again, $v_1$ commutes with $(I \pm^{b_2} P)$,
    which means that \cref{eqn:InternalPrepMeas} is equal to
    \begin{align}
        \cAssign{c}{b_2} \mapsto \sum_{b_1 \in \{0,1\}} (I \pm^{b_2} P) v_1 Q v_1 (I \pm^{b_2} P)
        &= \interpM{n_2}\left( \interpM{n_1}(Q) \right).
    \end{align}
    
    % \begin{align}
    %     &\interpM{n_1}\left(
    %         \interpM{n_2}(Q)
    %     \right) \\
    %     &= 
    %     \frac{1}{4} (I + P_Z) \interpM{n_2}(Q) (I + P_Z) \nonumber \\
    %     &+ \frac{1}{4} P_X (I - P_Z) \interpM{n_2}(Q) (I - P_Z) P_X \\
    %     &= \frac{1}{16} (I + P_Z) (I \pm P) Q (I \pm P) (I + P_Z) \nonumber \\
    %     &+ \frac{1}{16} P_X (I - P_Z) (I \pm P) Q (I \pm P) (I - P_Z) P_X
    % \end{align}
    % Now, because $\commute{P_Z}{P}$ and $\commute{P_X}{P}$, we can commute the $(I \pm P)$ factors to the outside of each term to obtain
    % \begin{align}
    %     &\frac{1}{16}  (I \pm P) (I + P_Z) Q (I + P_Z) (I \pm P) \nonumber \\
    %     &+ \frac{1}{16}  (I \pm P)  P_X (I - P_Z) Q (I - P_Z) P_X  (I \pm P) \\
    %     &= \interpM{n_2}\left(
    %         \interpM{n_1}(Q)
    %     \right)
    % \end{align}

    \proofpart[$n_1=\Prep{P_Z}{P_X}$ and $n_2$ is unitary] \label{case:PrepUnitary}

    \begin{align}
        \interpM{n_2}(\interpM{n_1}(Q))
        = \frac14 \interpM{n_2} (
            &(I + P_Z) Q (I + P_Z) \\
            &+
            P_X (I - P_Z) Q (I - P_Z) P_X
        )
    \end{align}
    Since $n_2$ is a unitary, we know that 
    $\interpM{n_2}(P_X) = P_X$ and 
    $\interpM{n_2}(I \pm P_Z) = I \pm P_Z$, and since it is monoidal (\cref{lem:MonoidalMapsNatural}), the equation above becomes
    \begin{align}
        \frac14 (
            &(I + P_Z) \interpM{n_2}(Q) (I + P_Z) \nonumber  \\
            &+
            P_X (I - P_Z) \interpM{n_2}(Q) (I - P_Z) P_X
        ) \\
        &= \interpM{n_1}(\interpM{n_2}(Q))
    \end{align}

    \proofpart[$n_1=\Meas{P}$ and $n_2=\Meas{P'}$]
    Similar to Cases~\ref{case:PrepPrep} and \ref{case:PrepMeas}.

    \proofpart[$n_1=\Meas{P}$ and $n_2$ unitary]
    Similar to \cref{case:PrepUnitary}.

    \proofpart[$n_1=\Rot{P}{\theta}$ and $n_2$ unitary]

    Assuming $\commute{n_2}{P}$, it must be the case that
    \begin{align}
        \interpM{n_2}(\interp{\Rot{P}{\theta}})
        &= \interpM{n_2}\left(
            \cosBy2{\theta} I + i \sinBy2{\theta} P
        \right) \\
        &= \cosBy2{\theta} I + i \sinBy2{\theta} P = \interp{\Rot{P}{\theta}} \notag
    \end{align}
    Thus
    \begin{align}
        &\interpM{n_2}(\interpM{\Rot{P}{\theta}}(Q)) \\
        &= \interpM{n_2}(\interp{\Rot{P}{-\theta}} Q \interp{\Rot{P}{\theta}}) \\
        &= \interpM{n_2}(\interp{\Rot{P}{-\theta}}) (\interpM{n_2}(Q)) \interpM{n_2}(\interp{\Rot{P}{\theta}}) \\
        &= \interp{\Rot{P}{-\theta}} (\interpM{n_2}(Q)) \interp{\Rot{P}{\theta}} \\
        &= \interpM{\Rot{P}{\theta}}(\interpM{n_2}(Q)).
    \end{align}

\end{proof}
\end{ExtendedVersion}

  \subsection{The \PCAST Graph}
    
\begin{definition}
    A \PCAST graph $G=(V,E)$ is a directed acyclic graph whose vertices $V$ are
    \PCAST nodes.
    % and whose edges $E$ are antireflexive and
    % asymmetric, and satisfy
    % { \small \begin{align}
    %     \forall n_1, n_2 \in V,~ \commute*{n_1}{n_2}
    %         \iff (n_1,n_2) \in E
    %         \xor
    %         (n_2,n_1) \in E.
    % \end{align} }%
For any nodes $n_1$ and $n_2$ that do not commute, there
is either an edge from $n_1$ to $n_2$ or vice versa (but not both). There are no
edges between commuting vertices.
\end{definition}

As an example, in \cref{fig:demo-optimized-graph} there are edges from both measurement nodes to the the measurement space function node $\mu$, but there are no edges to the Pauli frame $F'$ because both $Z_0$ and $X_1$ commute with $F'$.

% Note that for any topological sort $v_0,\ldots,v_{m-1}$, there is
% a unique way to add a new vertex $v$ at index $i$; any non-commutative nodes
% $v_0,\ldots,v_{i-1}$ will be given an edge into $v$, and any non-commutative 
% nodes $v_{i},\ldots,v_{m-1}$ will be given an edge out of $v$.

%\subsection{Semantics of a \PCAST graph}
%\label{sec:PoPR-graph-semantics}

Every topological ordering of a \PCAST graph $G$ corresponds to a \PCAST term $\termof{G}$.
Even though topological orderings of a graph are not unique, they are equivalent to each other due to
\cref{thm:PoPRNodeCommuteCorrect}: if two vertices do not commute, there is an
edge between them in one direction or the other, and that edge will be
preserved by the topological ordering.

\begin{ExtendedVersion}
\begin{lemma} \label{lem:topologicalOrderEquiv}
    For any two topological orderings $(n_0,\ldots,n_{k-1})$ and
    $(n_0',\ldots,n_{k-1}')$ of the same graph $G$,
    \begin{align}
        n_0; \cdots; n_{k-1}
        \equiv
        n_0'; \cdots; n_{k-1}'.
    \end{align}
\end{lemma}
\begin{proof}
    By induction on $k$. If $k=0$ the property is trivial.

    Otherwise, suppose the property holds for $k-1$. Then there exists some $j$ with $n_j'=n_0$ such that there is no edge from any $n_i'$ to $n_j'$ with $i < j$, and thus $\commute{n_i'}{n_j'}$. Therefore
    \begin{align}
        n_0'; \cdots; n_{j-1}'; n_j'; n_{j+1}';\cdots; n_{k-1}'
        &\equiv n_j'; n_0'; \cdots; n_{j-1}'; n_{j+1}';\cdots; n_{k-1}'
    \end{align}
    But then $n_1;\cdots;n_{k-1}$ and $ n_0'; \cdots; n_{j-1}'; n_{j+1}';\cdots; n_{k-1}'$ are both topological orderings of the graph $G' \setminus \{n_0\}$, so they too are equivalent, completing the proof.
\end{proof}
\end{ExtendedVersion}

\begin{ExtendedVersion}
\begin{definition}
    Let $G$ be a \PCAST graph and let $(n_0,\ldots,n_{k-1})$ be any topological
    ordering of its vertices. The semantics of $G$ as a cq-channel is given by
    \begin{align}
        \interpM{G} &= \interpM{n_0;\cdots;n_{k-1}}
            = \interpM{n_{k-1}} \circ \cdots \circ \interpM{n_0}.
    \end{align}
\end{definition}
\end{ExtendedVersion}

\begin{ExtendedVersion}
Similarly, any \PCAST term $t$ can be converted into a \PCAST graph $G^t$ with vertices $\{n_0,\ldots,n_{k-1}\}$ and edges
$(n_i,n_j)$ if and only if both $\commute*{n_i}{n_j}$ and $i < j$.

\begin{lemma}
    For any \PCAST term $t$ and Pauli $Q$,
    \begin{align}
        \interpM{G^t}(Q) &\equiv \interpM{t}(Q).
    \end{align}
\end{lemma}
\begin{proof}
    Follows from \cref{lem:topologicalOrderEquiv} and the fact that $t$ is a topological ordering of $G^t$.
\end{proof}
\end{ExtendedVersion}

\begin{comment}
Under this semantics, a \PCAST graph essentially describes an equivalence class of \PCAST terms.
Transformations on \PCAST graphs can then be thought of as transformations on
\PCAST terms: (1) convert a \PCAST graph into a linear \PCAST term corresponding to a
topological ordering of its vertices; (2) apply the transformation on the \PCAST
term; and (3) convert the \PCAST term back into a \PCAST graph.

Going into further detail, to apply a particular optimization like $n_1;n_2 \rightarrow t$
where $n_1$ and $n_2$ correspond to vertices in the \PCAST graph, we first need to
find a topological ordering whereby $n_1$ and $n_2$ are next to each other. This is possible
exactly when either (1) $n_1$ and $n_2$ commute and have exact same set of edges; or
(2) $n_1$ and $n_2$ do not commute and have distance $1$.
\end{comment}

\PCAST graphs satisfy three key invariants: frame-terminal, measurement-space terminal, and fully merged.
\begin{comment}
\begin{enumerate}
    \item \emph{Frame-terminal:} A well-formed \PCAST graph contains a single Pauli frame node, which occurs at the very end of the circuit;
    \item \emph{Measurement-space terminal:} A well-formed \PCAST graph contains at most one measurement space function node, which occurs at the very end of the circuit next to the terminating Pauli frame; and
    \item \emph{Fully merged:} A well-formed \PCAST graph contains no mergeable nodes.
\end{enumerate}
\end{comment}

\subsubsection{Frame-terminal graphs}

A \PCAST graph is called \emph{frame-terminal} when it contains a single Pauli frame node, and that frame has no outgoing edges. It is always possible to construct a frame-terminal graph because a Pauli frame $F$ can be commuted past any other node $n$ by transforming $n$ into a new node $\fwdAction{F}(n)$ that satisfies $F; n \equiv \fwdAction{F}(n); F$.
\begin{align} \label{lem:frameTerminal}
    \fwdAction{F}(n) &= \begin{cases}
        F^{-1} \circ F' \circ F & n=F' \\
        \Rot{\fwdAction{F}(P)}{\theta} & n = \Rot{P}{\theta} \\
        \Prep{\fwdAction{F}(P_Z)}{\fwdAction{F}(P_X)} & n=\Prep{P_Z}{P_X} \\
        \Meas[c]{\fwdAction{F}(P)} & n=\Meas[c]{P} \\
        \mu & n = \mu
    \end{cases}
\end{align}

\begin{ExtendedVersion}
\begin{lemma} \label{lem:extendedFrameTerminal}
    $F; n \equiv \fwdAction{F}(n); F$.
\end{lemma}
\begin{proof}
    By case analysis on $n$.

    \proofpart[$n=F'$] A corollary of \cref{lem:FrameCompositionExpanded} is that $F; F' \equiv F' \circ F$. Therefore, $F^{-1}\circ F' \circ F; F \equiv F \circ F^{-1}\circ F' \circ F \equiv F' \circ F$, as required.

    \proofpart[$n=\Rot{P}{\theta}$] \label{proof:frameTerminalRot} Notice that
    \begin{align}
        &\interpM{F}(\interp{\Rot{\fwdAction{F}(P)}{\theta}}) \nonumber \\
        &= \unitaryof{F} \left(\cosBy2{\theta}I + i \sinBy2{\theta}\fwdAction{F}(P)\right) (\unitaryof{F})^\dagger \\
        &= \unitaryof{F} \left(\cosBy2{\theta}I + i \sinBy2{\theta} (\unitaryof{F})^\dagger P \unitaryof{F} \right) (\unitaryof{F})^\dagger \\
        &= \cosBy2{\theta}I + i \sinBy2{\theta} P = \interp{\Rot{P}{\theta}}.
    \end{align}
    Because $\interpM{F}$ is monoidal (\cref{lem:conjugateMonoidal}), we have that
    \begin{align}
        &\interpM{F}\left( \interpM{\Rot{\fwdAction{F}(P)}{\theta}}(Q)\right) \nonumber \\
        %&= \interpM{F}\left( \interp{\Rot{\fwdAction{F}(P)}{\theta}} Q \interp{\Rot{\fwdAction{F}(P)}{-\theta}}\right) \\
        &= \left[\interpM{F}\left( \interp{\Rot{\fwdAction{F}(P)}{\theta}} \right)\right]
            \interpM{F}(Q)
            \left[\interpM{F}\left(\interp{\Rot{\fwdAction{F}(P)}{-\theta}}\right) \right] \\
        &= \interp{\Rot{P}{\theta}} \interpM{F}(Q) \interp{\Rot{P}{-\theta}} \\
        &= \interpM{\Rot{P}{\theta}}\left(\interpM{F}(Q)\right).
    \end{align}

    \proofpart[$n=\Prep{P_Z}{P_X}$] Similar to \cref{proof:frameTerminalRot}, it suffices to show that
    \begin{align}
        \interpM{F}\left(I \pm \fwdAction{F}(P_Z)\right) = I \pm P_Z 
        \quad\text{and}\quad
        \interpM{F} \left( \fwdAction{F}(P_X) \right) = P_X.
    \end{align}

    \proofpart[$n=\Meas{P}$] Similarly follows from 
    \begin{align}
        \interpM{F}\left(I \pm \fwdAction{F}(P)\right) &= I \pm P.
    \end{align}

    \proofpart[$n=\mu$] Trivial because $\mu$ commutes with $F$.
\end{proof}
\end{ExtendedVersion}

\subsubsection{Measurement-space terminal graph}

Similarly, a graph is called \emph{measurement-space terminal} when it contains at most one measurement space function node, and that node has no outgoing edges. Since frames and measurement space functions always commute, this does not conflict with the graph being frame-terminal. To construct a measurement-space-terminal graph, it suffices to push all measurement space nodes past measurement nodes\ifExtended{.}{ via the equivalence
\begin{align}
    \mu; \Meas[c]{P} &\equiv \Meas[c']{P}; \mu'
\end{align}
where $c'$ is fresh and $\mu'(m)=m; \cAssign{c}{m(c')}$.}

\begin{ExtendedVersion}
\begin{lemma}
    \begin{align}
    \mu; \Meas[c]{P} &\equiv \Meas[c']{P}; \mu[c \mapsto c']
    \end{align}
    where $c'$ is fresh and $\mu[c \mapsto c'](m)=m; \cAssign{c}{m(c')}$.
\end{lemma}
\begin{proof}
    Let $\mu'=\mu[c \mapsto c']$. Then
    \begin{align}
        &\interpM{\Meas[c']{P}; \mu'}(Q)
        = \interpM{\mu'}\left( \interpM{\Meas[c']{P}}(Q) \right) \\
        &= \interpM{\mu'}\left( m, \cAssign{c'}{b} \mapsto \tfrac{1}{4} (I \pm P)Q(I \pm P) \right) \\
        &= m, \cAssign{c'}{b}, \cAssign{c}{b} \mapsto \tfrac{1}{4} (I \pm P)Q(I \pm P)
    \end{align}
    This is exactly the result of $\interpM{\mu;\Meas[c]{P}}(Q)$, as desired.
\end{proof}
\end{ExtendedVersion}

\subsubsection{Mergeable nodes}

A pair of nodes $(n_1,n_2)$ is called \emph{mergeable} if there is no path between them in the graph and there is a step $n_1;n_2 \mergestep t$ to a new Pauli term, as defined in \cref{fig:MergeVerticesStep}. For such a merge to occur, $n_1$ and $n_2$ must have the same underlying Pauli arguments modulo $\pm 1$, in which case they can be combined into a single node, or in the case of measurement, two new but still simpler nodes. 

\begin{ShortVersion}
\begin{figure}
    \centering
{\small \begin{align*} \begin{array}{c}
    \Rot{P}{\theta_1}; \Rot{\pm P}{\theta_2} \mergestep \begin{cases}
        \frameof{\Rot{P}{\theta}} & \theta=\theta_1 \pm \theta_2 ~\text{is a} \\
                            & \text{multiple of}~\tfrac{\pi}{2} \\
        \Rot{P}{\theta} &\text{otherwise}
    \end{cases} \\
    \begin{aligned}
    \Prep{P_Z}{P_X}; \Prep{\pm P_Z}{P_X'}
            &\mergestep
            \Prep{\pm P_Z}{P_X'} \\
    \Prep{P_Z}{P_X}; \Rot{\pm P_Z}{\theta}
            &\mergestep
            \Prep{P_Z}{P_X}
        \\
    \Meas[c_1]{P}; \Meas[c_2]{(-1)^b P}
            &\mergestep
            \Meas[c_1]{P}; \cAssign{c_2}{c_1 + b} \\
    \Prep{P_Z}{P_X}; \Meas[c]{(-1)^b P_Z}
            &\mergestep
            \Prep{P}{P_X}; \cAssign{c}{b} \\
    \Rot{\pm P}{\theta}; \Meas[c]{P}
            &\mergestep
            \Meas[c]{P}
    \end{aligned} \\
    F_1; F_2 \mergestep F_2 \circ F_1
    \qquad\qquad
    \mu_1; \mu_2 \mergestep \mu_2 \circ \mu_1
\end{array} \end{align*} }%
    \caption{Rules for merging vertices. When $\theta$ is a multiple of $\frac{\pi}{2}$, $\Rot{P}{\theta}$ is a Clifford unitary.
    }
    \label{fig:MergeVerticesStep}
\end{figure}
\end{ShortVersion}
\begin{ExtendedVersion}
    \begin{figure*}
    \centering
{\small \begin{align}
    \Rot{P}{\theta_1}; \Rot{\pm P}{\theta_2} &\mergestep \begin{cases}
        \frameof{\Rot{P}{\theta}} & \theta=\theta_1 \pm \theta_2 ~\text{is a} \\
                            & \text{multiple of}~\tfrac{\pi}{2} \\
        \Rot{P}{\theta} &\text{otherwise}
    \end{cases}
        \label{eqn:MergeRotRot} \\
    \Prep{P_Z}{P_X}; \Prep{\pm P_Z}{P_X'}
            &\mergestep
            \Prep{\pm P_Z}{P_X'}
        \label{eqn:MergePrepPrep}\\
    \Prep{P_Z}{P_X}; \Rot{\pm P_Z}{\theta}
            &\mergestep
            \Prep{P_Z}{P_X}
        \label{eqn:MergePrepRot} \\
    \Meas[c_1]{P}; \Meas[c_2]{(-1)^b P}
            &\mergestep
            \Meas[c_1]{P}; \cAssign{c_2}{c_1 + b}
        \label{eqn:MergeMeasMeas} \\
    \Prep{P_Z}{P_X}; \Meas[c]{(-1)^b P_Z}
            &\mergestep
            \Prep{P}{P_X}; \cAssign{c}{b}
        \label{eqn:MergePrepMeas} \\
    \Rot{\pm P}{\theta}; \Meas[c]{P}
            &\mergestep
            \Meas[c]{P}
        \label{eqn:MergeRotMeas} \\
    F_1; F_2 &\mergestep F_2 \circ F_1
        \label{eqn:MergeFrameFrame} \\
    \mu_1; \mu_2 &\mergestep \mu_2 \circ \mu_1
        \label{eqn:MergeMuMu}
\end{align} }%
    \caption{Rules for merging vertices. When $\theta$ is a multiple of $\frac{\pi}{2}$, $\Rot{P}{\theta}$ is a Clifford unitary.
    } 
    \label{fig:MergeVerticesStep}
\end{figure*}
\end{ExtendedVersion}

\begin{ExtendedVersion}
\begin{figure}
    \centering
  \begin{align}
    \Rot{P}{\theta_1}; \Rot{P}{\theta_2} &\equiv \Rot{P}{\theta_1 + \theta_2}
        \label{eqn:RotRotEquivRot} \\
    \Rot{-P}{\theta} &\equiv \Rot{P}{-\theta}
        \label{eqn:RotNegEquivRot} \\
    % \Rot{P}{0} &\equiv 1
    %     \label{eqn:RotEquivOne} \\
    \Rot{P}{k\tfrac{\pi}{2}} &\equiv \frameof{\interp{\Rot{P}{k\tfrac{\pi}{2}}}}
        \label{eqn:RotEquivClifford} \\
    \Prep{P_Z}{P_X}; \Rot{P_Z}{\theta}
            &\equiv
            \Prep{P_Z}{P_X}
            \label{eqn:PrepRotEquivPrep}
            \\
    \Rot{P}{\theta}; \Meas[c]{P}
            &\equiv
            \Meas[c]{P}
            \label{eqn:RotMeasEquivMeas}
            \\
    % \Rot{P}{\pi} &\equiv \pauliToFrame{P} \\
    % \Rot{P}{\frac{\pi}{2}} &\equiv \pauliToFrame2{P}
    \Meas[c_1]{P}; \Meas[c_2]{P}
            &\equiv
            \Meas[c_1]{P}; c_2 \gets c_1
            \label{eqn:MeasMeasEquivMeas} \\
    \Meas[c]{-P}
            &\equiv
            \Meas[c]{P}; c \gets \neg c
            \label{eqn:MeasNegEquivMeas} \\
    \Prep{P_Z}{P_X}; \Prep{P_Z}{P_X'}
            &\equiv
            \Prep{P_Z}{P_X}
            \label{eqn:PrepPrepEquivPrep} \\
    \Prep{-P_Z}{P_X} &\equiv \Prep{P_Z}{P_X}; \frameof{P_X}
            \label{eqn:PrepNegEquivPrep} \\
        % \Prep{P_Z}{P_X}
        %     &\equiv
        %     \Prep{P_Z}{-i P_Z P_X}
        %     \label{eqn:PrepEquivPrep1} \\
        % \Prep{P_Z}{P_X}
        %     &\equiv
        %     \Prep{P_Z}{-P_X}
        %     \label{eqn:PrepEquivPrep2} \\
        % \Rot{P_Z}{\theta}; \Prep{P_Z}{P_X}
        %     &\equiv
        %     \Prep{P_Z}{P_X}
        %     \label{eqn:RotPrepEquivPrep} \\
    \Prep{P_Z}{P_X}; \Meas[c]{P_Z}
            &\equiv
            \Prep{P_Z}{P_X}; \cAssign{c}{0}
            \label{eqn:PrepMeasEquivPrep}
    \end{align}
    \caption{Equivalence properties of \PCOAST nodes. In \cref{eqn:RotEquivClifford}, $k \in \mathbb{Z}$ is an integer.
    }
    \label{fig:NodeEquivProperties}
\end{figure}

To show $t_1 \mergestep t_2$ is sound, it is simpler to derive them from the equivalences in \cref{fig:NodeEquivProperties}.

\begin{lemma}
    The equivalence rules in \cref{fig:NodeEquivProperties} are sound.
\end{lemma}
\begin{proof}
    \proofpart[\cref{eqn:RotRotEquivRot,eqn:RotNegEquivRot}] Follows from the unitary definition of $\interp{\Rot{P}{\theta}}$. \proofpart[\cref{eqn:RotEquivClifford}] It suffices to check that $\interp{\Rot{P}{k\tfrac{\pi}{2}}}$ is indeed a Clifford for $k \in \mathbb{Z}$.

    \proofpart[\cref{eqn:PrepRotEquivPrep}]
        To prove $\Prep{P_Z}{P_X}; \Rot{P_Z}{\theta} \equiv \Prep{P_Z}{P_X}$, it suffices to see that
        $\interpM{\Rot{P_Z}{\theta}}$  is the identity on $(Q (I + P_Z))$ when $\commutativity{Q}{P_Z}=0$.
        
    \proofpart[\cref{eqn:RotMeasEquivMeas}] To see that $\Rot{P}{\theta}; \Meas[c]{P} \equiv \Meas[c]{P}$, since $\interpM{\Rot{P}{\theta}}$ is $0$ when $\commutativity{Q}{P}=1$ and $Q$ otherwise; it suffices to observe that $\interpM{\Meas[c]{P}}(Q)$ is 0 when $\commutativity{Q}{P}=1$.
    
    \proofpart[\cref{eqn:MeasMeasEquivMeas}] 
    { \begin{align}
        &\interpM{\Meas[c_2]{P}}\left( \interpM{\Meas[c_1]{P}} (Q) \right) \\
        &\equiv \cAssign{c_1}{b_1} \mapsto
            \interpM{\Meas[c_2]{P}}\left(\frac{1}{4} (I \pm^{b_1} P) Q (I \pm^{b_1} P) \right)
            \nonumber \\
        &\equiv 
            \begin{aligned}[t]
            &\cAssign{c_1}{b_1}, \cAssign{c_2}{b_2} \mapsto \\
            &\frac{1}{16} (I \pm^{b_2} P) (I \pm^{b_1} P) Q (I \pm^{b_1} P) (I \pm^{b_2} P)
            \end{aligned}
    \end{align} }%
    by \cref{lem:PauliNodeSubstitution}. When $b_1 \neq b_2$, this term reduces to $0$, so we can ignore those measurement outcomes. On the other hand, when $b_1 = b_2$ the term becomes
    \begin{align}
        \cAssign{c_1}{b}, \cAssign{c_2}{b} \mapsto
        \frac{1}{4} (I \pm^{b} P) Q (I \pm^{b} P)
    \end{align}
    which is exactly the semantics of $\Meas[c_1]{P}; \cAssign{c_2}{c_1}$.
    
    \proofpart[\cref{eqn:MeasNegEquivMeas}]
    \begin{align}
        \interpM{\Meas[c]{-P}}(Q)
        &= \cAssign{c}{b} \mapsto \frac{1}{4} (I \mp^{b} P) Q (I \mp^{b} P) \\
        &= \cAssign{c}{b} \mapsto \frac{1}{4} (I \pm^{\neg b} P) Q (I \pm^{\neg b} P) 
    \end{align}
    
    \proofpart[\cref{eqn:PrepPrepEquivPrep}]
        Expanding out \cref{lem:PauliNodeSubstitution},
        \begin{align} \label{eqn:internalPrepPrepEquivPrep}
            &\interpM{\Prep{P_Z}{P_X'}}\left( \interpM{\Prep{P_Z}{P_X}} (Q) \right) \nonumber \\
            &= \frac{1}{16} \sum_{b_1,b_2 \in \{0,1\}} \conjugate*{v_{b_1,b_2}}(Q)
        \end{align}
        where
        \begin{align} 
            v_{b_1,b_2} &= (P_X')^{b_2} (I \pm^{b_2} P_Z) (P_X)^{b_1} (I \pm^{b_1} P_Z).
        \end{align}
        Expanding out the definition of $v_{b_1,b_2}$, we see that
        \begin{align}
            v_{0,0} &= (I + P_Z) (I + P_Z) = 2(I + P_Z) \\
            v_{0,1} &= P_X' (I - P_Z) (I + P_Z) = 0 \\
            v_{1,0} &= (I + P_Z) P_X (I - P_Z) = 2 P_X (I - P_Z) \\
            v_{1,1} &= (P_X') (I - P_Z) P_X (I - P_Z) = 0
        \end{align}
        Therefore, \cref{eqn:internalPrepPrepEquivPrep} can be rewritten as
        \begin{align}
            &\frac{1}{4} (I + P_Z) Q (I + P_Z) + \frac{1}{4} P_X (I - P_Z) Q (I - P_Z) P_X  \\
            &= \interpM{\Prep{P_Z}{P_X}}(Q).
        \end{align}

        % \begin{align}
        %     &=
        %         \frac{1}{4} (I + P_Z) \left(  \interpM{\Prep{P_Z}{P_X}} (Q) \right) (I + P_Z) \nonumber \\
        %         &+ \frac{1}{4} P_X' (I + P_Z) \left( \interpM{\Prep{P_Z}{P_X}} (Q) \right) (I - P_Z) P_X'
        %         \\
        %     &= \frac{1}{16} (I + P_Z) (I + P_Z) Q (I + P_Z) (I + P_Z) \nonumber \\
        %     &+ \frac{1}{16} (I + P_Z) P_X (I - P_Z) Q (I - P_Z) P_X (I + P_Z) \nonumber \\
        %     &+ \frac{1}{16} P_X' (I - P_Z) (I + P_Z) Q (I + P_Z) (I - P_Z) P_X' \nonumber \\
        %     &+ \frac{1}{16} P_X' (I - P_Z) P_X (I - P_Z) Q (I - P_Z) P_X (I - P_Z) P_X'
        %         \label{eqn:internalPrepPrepEquivPrep}
        % \end{align}
        % Now, the third and forth terms in that sum reduce to $0$, and furthermore $(I \pm P_Z)(I \pm P_Z) = 2 (I \pm P_Z)$, so overall \cref{eqn:internalPrepPrepEquivPrep} is equal to
        % \begin{align}
        %     \frac{4}{16} (I + P_Z) Q (I + P_Z)
        %     + \frac{4}{16} P_X (I - P_Z) Q (I - P_Z) P_X
        % \end{align}
        % which is exactly $\interpM{\Prep{P_Z}{P_X}}(Q)$.
        
    \proofpart[\cref{eqn:PrepNegEquivPrep}]
        Expanding out definitions,
        {\small \begin{align}
            &\interpM{\Prep{-P_Z}{P_X}}(Q) \nonumber \\
            &= \frac{1}{4} (I - P_Z) Q (I - P_Z) + \frac{1}{4} P_X (I + P_Z) Q (I + P_Z) P_X \\
            &= P_X \left( \frac{1}{4} P_X (I - P_Z) Q (I - P_Z) P_X + \frac{1}{4} (I + P_Z) Q (I + P_Z)\right) P_X \\
            &= \interpM{\frameof{P_X}}\left( \interpM{\Prep{P_Z}{P_X}}(Q) \right)
        \end{align} }%
    
    \proofpart[\cref{eqn:PrepMeasEquivPrep}]
    Observe that
    \begin{align}
        &\interpM{\Meas[c]{P_Z}}\left( \interpM{\Prep{P_Z}{P_X}}(Q) \right) \nonumber \\
        &= \begin{cases}
            \interpM{\Meas[c]{P_Z}}((I + P_Z) Q)
                & \commutativity{P_Z}{Q}=\commutativity{P_X}{Q} = 0 \\
            0   & \text{otherwise}
        \end{cases} 
    \end{align}
    Therefore it suffices to show that 
    \begin{align}
        &\interpM{\Meas[c]{P_Z}}((I + P_Z) Q) \nonumber \\
        &= \cAssign{c}{b} \mapsto
            \tfrac{1}{4} (I \pm^b P_Z) (I + P_Z) Q (I + P_Z) (I \pm^b P_Z) \\
        &= \cAssign{c}{b} \mapsto 
        \begin{cases}
            (I + P_Z) Q & b=0 \\
            0 & b=1
        \end{cases} \\
        &= \interpM{\cAssign{c}{0}; \Prep{P_Z}{P_X}}(Q) \qedhere
    \end{align}
\end{proof}
\end{ExtendedVersion}

\begin{lemma} \label{lem:mergestep}
    If $t_1 \mergestep t_2$
    then $t_1 \equiv t_2$.
\end{lemma}
\begin{ExtendedVersion}
    \begin{proof}
        % \cref{eqn:MergeRotRot} follows from \cref{eqn:RotRotEquivRot,eqn:RotNegEquivRot,eqn:RotEquivClifford}.

        % \cref{eqn:MergePrepPrep} follows from \cref{eqn:PrepPrepEquivPrep,eqn:PrepNegEquivPrep}.

        % \cref{eqn:MergePrepRot} follows from \cref{eqn:PrepRotEquivPrep,eqn:RotNegEquivRot}.

        % \cref{eqn:MergeMeasMeas} follows from \cref{eqn:MeasMeasEquivMeas,eqn:MeasNegEquivMeas}.

        % \cref{eqn:MergePrepMeas} follows from \cref{eqn:PrepMeasEquivPrep,eqn:MeasNegEquivMeas}.

        % \cref{eqn:MergeRotMeas} follows from \cref{eqn:RotNegEquivRot,eqn:RotMeasEquivMeas}.

        Most of the rules follow directly from \cref{fig:NodeEquivProperties}, but \cref{eqn:MergeFrameFrame} follows from \cref{lem:FrameCompositionExpanded}, and \cref{eqn:MergeMuMu} follows from definition.
    \end{proof}
\end{ExtendedVersion}

\begin{ExtendedVersion}
  \subsection{Two-axis Rotations} \label{sec:twoaxis}
    
In the implementation we extend \PCAST nodes with two-axis rotations: $\Rot2{P_1}{P_2}[\theta][\phi]$, a generalization of a single-qubit gate used in NISQ machines.

\begin{definition} \label{defn:twoaxis}
    For two anticommuting Hermitian Paulis $P_1$ and $P_2$, a two-axis Pauli rotation
    $\Rot2{P_1}{P_2}[\theta][\phi]$ is a rotation of a Pauli of the form
    $\cos{\phi}P_1 + \sin{\phi} P_2$:
    {\small \begin{align}
        &\interpM{\Rot2{P_1}{P_2}[\theta][\phi]}(Q) \\
        &= \interp{\Rot2{P_1}{P_2}[\theta][\phi]} Q \interp{\Rot2{P_1}{P_2}[-\theta][\phi]} \nonumber \\
        &\interp{\Rot2{P_1}{P_2}[\theta][\phi]}
        = \cos(\tfrac{\theta}{2}) I + i \sin(\tfrac{\theta}{2}) (\cos{\phi} P_1 + \sin{\phi} P_2)
    \end{align} }%
\end{definition}

We extend the commutativity relation, $\fwdAction{F}(n)$, and the merge relation to account for two-axis rotations in \cref{fig:Rot2EquivProperties}.

\begin{figure*}
    \centering
\begin{align*}
    \inferrule*
        {\commute{P_1}{n} \\ \commute{P_2}{n}}
        {\commute{\Rot2{P_1}{P_2}[\theta][\phi]}{n}}
    \qquad
    \inferrule*
        {\commute{Q}{P_1} \\ \commute{Q}{P_2}}
        {\commute{Q}{\Rot2{P_1}{P_2}[\theta][\phi]}}
\end{align*}
\begin{align}
    \fwdAction{F}(\Rot2{P_1}{P_2}[\theta][\phi])
    &= \Rot2{\fwdAction{F}(P_1)}{\fwdAction{F}(P_2)}[\theta][\phi] \\
    \Prep{P_Z}{P_X}; \Rot2{\pm P_Z}{P'}[\theta][\phi]
        &\mergestep \Prep{P_Z}{P_X} \\
    \Prep{P_Z}{P_X}; \Rot2{P'}{\pm P_Z}[\theta][\phi]
        &\mergestep \Prep{P_Z}{P_X} \\
    % \Rot2{P_Z}{P_X}[\theta_1][\phi]; \Rot2{P_Z}{P_X}[\theta_2][\phi]
    %     &\mergestep \begin{cases}
    %         1 & \theta_1 + \theta_2 = 0 \\
    %         \Rot2{P_Z}{P_X}[\theta_1+\theta_2][\phi] & \text{otherwise}
    %     \end{cases} \\
    \Rot2{P_Z}{P_X}[\theta_1][\phi]; \Rot2{P_Z}{\pm P_X}[\theta_2][\pm \theta]
        &\mergestep \Rot2{P_Z}{P_X}[\theta_1 + \theta_2][\phi] \\
    \Rot2{P_Z}{P_X}[\theta_1][\phi]; \Rot2{P_Z}{\pm P_X}[\theta_2][\pm(\pi - \phi)]
        &\mergestep \Rot2{P_Z}{P_X}[\theta_1 + \theta_2][\phi] 
\end{align}
    \caption{Properties of two-axis rotations}
    \label{fig:Rot2EquivProperties}
\end{figure*}

\begin{comment}
\begin{align}
    F ; \Rot2{P_1}{P_2}[\theta][\phi]
        &\equiv
        \Rot2{F^{-1}(P_1)}{F^{-1}(P_2)}[\theta][\phi]; F
    \\
    \Rot2{P}{P_2}[\theta][\phi]; \Meas[c]{P}
            &\equiv
            \Meas[c]{P}
            \label{eqn:Rot2MeasEquiv}
            % Albert: no general rule I know of...
    \\
    F ; \Rot2{P_1}{P_2}[\theta][\phi]
        &\framestep
        \Rot2{F^{-1}(P_1)}{F^{-1}(P_2)}[\theta][\phi]; F
    \\
\end{align}

\begin{align} \begin{array}{c}
    \inferrule*
        {
            P = \pm P_1
            \albert{\text{No general rule that I know of}}
        }
        {
            \Rot2{P_1}{P_2}[\theta][\phi]; \Meas[c]{P}
            \mergestep
            \Meas[c]{P}
        }
    \\ \\
    \inferrule*
        { 
            P_1' = \alpha_1 P_1 \\
            P_2' = \alpha_2 P_2 \\
            \alpha_1,\alpha_2 \in \{1,-1\} \\
            \phi = \alpha_2 \tfrac{\pi}{2} (1-\alpha_1 + \alpha_1) \phi'
        }
        { \Rot2{P_1}{P_2}[\theta][\phi] ; \Rot2{P_1'}{P_2'}[\theta'][\phi']
            \mergestep
            \Rot2{P_1}{P_2}[\theta + \theta'][\phi]
        }
    \\ \\
\end{array} \end{align}
\end{comment}

\begin{lemma}
    \cref{thm:PoPRNodeCommuteCorrect}, \cref{lem:extendedFrameTerminal}, and \cref{lem:mergestep} still hold with the addition of rules from \cref{fig:Rot2EquivProperties}.
\end{lemma}
\begin{proof}

    \proofpart[\cref{thm:PoPRNodeCommuteCorrect}] To show the commutativity rules in \cref{fig:Rot2EquivProperties} are sound, we must show that if $\commute{\Rot2{P_1}{P_2}[\theta][\phi]}{n}$ then $\Rot2{P_1}{P_2}[\theta][\phi]; n \equiv n; \Rot2{P_1}{P_2}[\theta][\phi]$. We start by observing that two-axis rotations satisfy \cref{lem:unitaryNodeCommutePauli}: namely, that if $\commute{Q}{\Rot2{P_1}{P_2}[\theta][\phi]}$ then $\interpM{\Rot2{P_1}{P_2}[\theta][\phi]}(Q)=Q$. Then, the proof of \cref{thm:PoPRNodeCommuteCorrect} automatically extends to two-axis rotations as instances of unitary nodes. It suffices to show the property holds when $n$ is unitary, in which case
    {\small \begin{align} \label{eqn:InternalTwoPoPRNodeCommuteCorrect}
        \interpM{n} \left( \interpM{\Rot2{P_1}{P_2}[\theta][\phi]}(Q) \right)
        = &\interpM{n} \left( \interp{\Rot2{P_1}{P_2}[\theta][\phi]} \right) \nonumber \\
        \cdot &\interpM{n} \left( Q \right) \\
        \cdot &\interpM{n} \left( \interp{\Rot2{P_1}{P_2}[-\theta][\phi]} \right) \nonumber 
    \end{align} }%
    We know that $\commute{n}{P_1}$ and $\commute{n}{P_2}$, and so by \cref{lem:unitaryNodeCommutePauli},
    \begin{align}
        \interpM{n} \left( \interp{\Rot2{P_1}{P_2}[\theta][\phi]} \right) = \interp{\Rot2{P_1}{P_2}[\theta][\phi]}.
    \end{align}
    Thus \cref{eqn:InternalTwoPoPRNodeCommuteCorrect} is equal to
    $
        \interpM{\Rot2{P_1}{P_2}[-\theta][\phi]}\left( \interpM{n}(Q) \right)
    $
    as desired.

    \proofpart[\cref{lem:extendedFrameTerminal}] Our goal is to show that
    {\small \begin{align}
        F; \Rot2{P_1}{P_2}[-\theta][\phi] &\equiv \Rot2{\fwdAction{F}(P_1)}{\fwdAction{F}(P_2)}[-\theta][\phi]; F.
    \end{align} }%
    Expanding out definitions, it suffices to observe that
    \begin{align}
        \interpM{F}\left( \interp{\Rot2{\fwdAction{F}(P_1)}{\fwdAction{F}(P_2)}[-\theta][\phi]} \right)
        = \interp{\Rot2{P_1}{P_2}[-\theta][\phi]}.
    \end{align}

    \proofpart[\cref{lem:mergestep}] Finally, we need to prove that if $t_1 \mergestep t_2$ then $t_1 \equiv t_2$ for the rules in \cref{fig:Rot2EquivProperties}. These can be derived from the set of following properties, all of which can be easily checked:
    \begin{align}
    \Rot2{P_1}{P_2}[\theta][\phi] \equiv~
        &\Rot{-i P_2 P_1}{\phi};
        \Rot{P_1}{\theta} ; \notag \\
    & \Rot{-i P_2 P_1}{-\phi} \\
    %(\Rot2{P_1}{P_2}[\theta][\phi])^\dagger
    %    &\equiv
    %    \Rot2{P_1}{P_2}[-\theta][\phi]
    %    \\
    \Rot2{P_1}{P_2}[\theta][\phi] \equiv~ &\Rot2{P_2}{P_1}[\theta][\tfrac{\pi}{2}-\phi] \\
    \Rot2{P_1}{-P_2}[\theta][\phi]
        \equiv~
        &\Rot2{P_1}{P_2}[\theta][-\phi] \\
    \Rot2{P_1}{P_2}[0][\phi] &\equiv 1 \\
    \Rot2{P_1}{P_2}[\theta][0] &\equiv \Rot{P_1}{\theta} \\
    \Rot2{P_1}{-P_2}[\theta][\phi]
        &\equiv
        \Rot2{P_1}{P_2}[\theta][-\phi]
        \\
    \Rot2{-P_1}{P_2}[\theta][\phi]
        &\equiv
        \Rot2{P_1}{P_2}[\theta][\pi-\phi]
\end{align}
\end{proof}

\end{ExtendedVersion}

\section{The \PCAST Optimization}
\label{sec:optimization}
The \PCAST optimization, illustrated in \cref{fig:optimizationflow}, 
starts with a circuit, compiles it to a \PCAST graph, optimizes that graph, and then synthesizes a circuit again.
To maintain the frame- and measurement-space-function-terminal invariants, we separate out the terminal frame $F$ and measurement space function $\mu$ from the rest of the graph $G$.
%The graph $G^{\text{orig}}$ constructed directly from the original circuit will not have an accompanying measurement space function, since measurement space functions are only introduced via a release outcome in the optimization stage.
%If introduced, the synthesized output circuit will include the measurement space function to enable classical rewriting of measurement variables.

\begin{figure}
    \centering
    \begin{adjustbox}{width=\columnwidth}% there is also 'max width' to only scale it down if it is larger
    \begin{tikzpicture} 
      \node (Corig) at (0,0) {
        \begin{tikzcd}[row sep=0.1cm]
            \startingpoprgate{C^{\text{orig}}}
        \end{tikzcd}
      };

      \node (Gorig) at (3,0) {
        \begin{tikzcd}[row sep=0.1cm]
          \pcastgraph{3}{2} & \startingpoprgate{F^{\text{orig}}} %3
          \\
          \startingpoprgate{G^{\text{orig}}} \arrow[ru] \arrow[rd] & %3
          \\
          & \startingpoprgate{\mu^{\text{orig}}}
        \end{tikzcd}
      };

      \node (Gopt) at (7,0) {
        \begin{tikzcd}[row sep=0.1cm]
          \pcastgraph{3}{2} & \startingpoprgate{F^{\text{opt}}} %2
          \\
          \startingpoprgate{G^{\text{opt}}} \arrow[ru] \arrow[rd] & %2
          \\
          & \startingpoprgate{\mu^{\text{opt}}}
      \end{tikzcd}
      };

      \node (Copt) at (10,0) {
        \begin{tikzcd}[row sep=0.1cm]
          \startingpoprgate{C^{\text{opt}}}
                \pcastgraph{2}{1}
          \\
          \startingpoprgate{\mu^{\text{opt}}}
      \end{tikzcd}
      };
      
        \path[->] (Corig) edge node[above] {A} (Gorig);
        \path[->] (Gorig) edge node[above] {B} (Gopt);
        \path[->] (Gopt)  edge node[above] {C} (Copt);
    \end{tikzpicture}
    \end{adjustbox}
    \caption{Flow for constructing a \PCAST graph from a circuit, starting with the circuit on the left $C^{\text{orig}}$ and ending with the optimized circuit on the right $C^{\text{opt}}$. The steps correspond to Sections~\ref{sec:toGraph}, \ref{sec:opts}, and \ref{sec:synthesis}).}
    \label{fig:optimizationflow}
\end{figure}

\subsection{Compiling circuits to \PCAST graphs}
\label{sec:toGraph}

The procedure $\circuitToPoPR{C}$ (\cref{alg:circToPoPR}) produces a \PCAST graph $G$, a terminating Pauli frame $F$, and a terminating measurement space function $\mu$ from a circuit $C$ by moving gates into $G$ or $F$ respectively using a helper function $\addNode{G}{n}$ (\cref{fig:addNode}).
The algorithm maintains the loop invariant that $\termof{C} \equiv g; F; \mu; \termof{C'}$.

\begin{comment}
\begin{figure}
\begin{algorithmic}[1]
    \Procedure {$\circuitToPoPR{C}$}{}
        \State $G \gets \emptyset$; $F \gets \frameof{I}$; $C' \gets C$
        \While{$C' = g; C''$}
            \If {$g$ is a Clifford gate}
                \State $F \gets \frameof{g} \circ F$
            \Else
                \State $(G, F') \gets \addTerm{G}{\fwdAction{F}(\termof{g})}$
                \State $F \gets F \circ F'$
            \EndIf
            \State $C' \gets C''$
        \EndWhile
    \State \Return $(G, F)$
    \EndProcedure
\end{algorithmic}
    \caption{The function $\circuitToPoPR{C}$ takes as input a circuit $C$ and returns a \PCAST graph $G$ and a Pauli frame $F$ such that $\termof{G}; F \equiv \termof{C}$.  We assume that every gate in $C$ can be written in a straightfoward way as a \PCAST term $\termof{g}$, or, if it is a Clifford gate, as a Pauli frame $\frameof{g}$.}
    \label{alg:circToPoPR}
\end{figure}
\end{comment}

\begin{figure}
\begin{algorithmic}[1]
    \Procedure {$\circuitToPoPR{C}$}{}
        \State $G \gets \emptyset$; $F \gets \frameof{I}$; $\mu \gets \mu^0$
        \For {each $g$ in $C$}
            \For {each $n$ in $\termof{g}$}
                \If {$n=\Meas[c]{P}$}
                        \State $n \gets \Meas[c']{P}$; $\mu \gets \mu[c \mapsto c']$
                            \Comment{$c'$ fresh}
                \EndIf
                \If {$n=\mu'$} $\mu \gets \mu' \circ \mu$
                \ElsIf {$n=F'$} $F \gets F' \circ F$
                \Else
                    \State $(G', F', \mu') \gets \addNode{G}{\fwdAction{F}(n)}$
                    \State $G \gets G'; F \gets F \circ F'; \mu \gets \mu \circ \mu'$
                \EndIf
            \EndFor
        \EndFor
    \State \Return $(G, F, \mu)$
    \EndProcedure
\end{algorithmic}
    \caption{The function $\circuitToPoPR{C}$ takes as input a circuit $C$ and returns a \PCAST graph $G$, a Pauli frame $F$, and a measurement space function $\mu$ such that $G; F; \mu \equiv \termof{C}$.  We assume that every gate in $C$ can be written in a straightforward way as a \PCAST term $\termof{g}$.}
    \label{alg:circToPoPR}
\end{figure}

\begin{figure}
\begin{algorithmic}[1]
\Procedure {$\ADDNODE$}{$G$, $n$}
    \If {$n=F$} \Return $(G,F,\mu^0)$
    \ElsIf {$n=\mu$} \Return $(G,\frameof{I},\mu)$
    \Else
        \For {each node $n_0$ in $G$ with outdegree $0$}
            \If {$n_0; n \mergestep n_0'$ OR $n_0';\mu_0$}
                \State $G' \gets$ \Call{\REMOVEVERTEX}{$G, n_0$}
                \State $(G'', F, \mu) \gets \addNode{G'}{n_0'}$
                \State \Return $(G'', F, \mu_0 \circ \mu)$
            \EndIf
        \EndFor
        
        \State $G' \gets$ \Call{\ADDVERTEX}{$G, n'$}
        \For {each node $n_0$ in $G'$} 
            \If {$\commute*{n_0}{n'}$}
                $G' \gets $\Call{\ADDEDGE}{$G', n_0, n'$}
            \EndIf
        \EndFor
        \Return $(G', \frameof{I}, \mu^0)$
    \EndIf
\EndProcedure
\end{algorithmic}
    \caption{The procedure $\addNode{G}{n}$, where $G$ is a \PCAST graph and $n$ is a \PCAST node,
returns an updated graph $G'$, a Pauli frame $F$, and a measurement space function $\mu$ such that $G; n \equiv G'; F; \mu$. 
The functions \ADDVERTEX, \REMOVEVERTEX, and \ADDEDGE implement the corresponding simple graph operations.
Note that if $n_0; n \mergestep t$, it is either the case that $t$ is a single node $n_0'$, or is equal to $n_0'; \mu_0$.
}
    \label{fig:addNode}
\end{figure}

\begin{comment}
\begin{figure}
\begin{algorithmic}[1]
\Procedure {$\ADDTERM$}{$G$, $t$}% {$\addTerm{G}{t}$}{}
%    \State \Comment{Returns $(G',F)$ such that $\termof{G'}; F \equiv \termof{G}; t$.}
\State $G' \gets G; t' \gets t; F' \gets \frameof{I}$
\While {$t'=n;t''$}
\State $t' \gets t''$
\If {$n$ is a Pauli frame $F$}
    $F' \gets F \circ F'$
\Else~
    \State $n' \gets \fwdAction{F'}(n)$ \Comment{commute $n$ past $F'$}
    \For {each node $n_0$ in $G'$ with outdegree $0$}
        \If {$n_0; n' \mergestep t_0$}
            \State $G' \gets$ \Call{\REMOVEVERTEX}{$G', n_0$} 
            \State $(G', F'') \gets$ \Call{\ADDTERM}{$G',t_0$}
            \State $F' \gets F' \circ F''$
            \State $n' \gets 1$  \Comment{break out of the outer loop}
        \EndIf
    \EndFor
    \State $G' \gets$ \Call{\ADDVERTEX}{$G', n'$}
    \For {each node $n_0$ in $G'$} 
      \If {$\commute*{n_0}{n'}$}
        $G' \gets $\Call{\ADDEDGE}{$G', n_0, n'$}
      \EndIf
    \EndFor
\EndIf
\EndWhile
\Return $(G', F')$

\EndProcedure
\end{algorithmic}
\caption{The procedure $\addTerm{G}{t}$, where $G$ is a \PCAST graph and $t$ is a \PCAST term,
returns an updated graph $G'$ and Pauli frame $F$ such that $\termof{G'}; F \equiv \termof{G}; t$. 
The functions \ADDVERTEX, \REMOVEVERTEX, and \ADDEDGE implement the corresponding simple graph operations.}
\label{fig:addTerm}
\end{figure}
\end{comment}

\begin{ExtendedVersion}

\begin{lemma} \label{lem:addNodeCorrect}
    $\addNode{G}{n}$ always terminates, and returns $(G',F, \mu)$ such that $G; n \equiv G'; F;\mu$.
\end{lemma}
\begin{proof}
    $\ADDNODE$ always terminates because the recursive call on line 8 is always invoked with a $G'$ containing one fewer node than $G$.

    Regarding correctness, if none of the branches of the for loop on line 5 are taken, then the semantics of $G'$ (line 10) will be equivalent to $G;n$, satisfying the invariant. However, if one of the branches is taken, i.e. there is an $n_0$ with $n_0; n \mergestep n_0'; \mu_0$, then we know we can write $G \equiv G'; n_0$ by the fact that $n_0$ has outdegree 0 in $G$.
    Then, the recursive call tells us that $G'; n_0' \equiv G''; F; \mu$, and so
    \begin{align}
        G; n &\equiv G'; n_0;n  \equiv G'; n_0'; \mu_0 \\
        &\equiv G''; F; \mu; \mu_0 \equiv G''; F; \mu_0 \circ \mu
        \qedhere
    \end{align}
\end{proof}

\begin{theorem} \label{lem:circuitToPoPRCorrect}
    $\circuitToPoPR{C}$ returns $(G, F, \mu)$ such that $C \equiv G; F; \mu$.
\end{theorem}
\begin{proof}
    The loop on line 3 preserves the invariant that $C \equiv G; F; \mu; C'$ where $C$ is the original circuit and $C'$ is the remainder after each loop iteration. In the body of the for loops, we have $C \equiv G; F; \mu; n; C'$. After the if statements in lines 5 and 7, we know that we can commute $n$ past $\mu$ to obtain $C \equiv G; F; n; \mu; C'$.
    From \cref{lem:extendedFrameTerminal} we know that this is equivalent to $G; \fwdAction{F}(n); F; \mu; C'$. From \cref{lem:addNodeCorrect} this is equivalent to $G'; F'; \mu'; F; \mu; C'$ where $(G',F',\mu'$ is the result of $\addNode{G}{\fwdAction{F}{n}}$. Then, since $\mu'$ and $F$ commute, this is equivalent to $G'; F\circ F'; \mu \circ \mu'$, as obtained in line 11, which completes the proof.
\end{proof}
\end{ExtendedVersion}

\subsection{Internal optimizations on \PCAST graphs} \label{sec:opts}
We give a brief overview of the internal optimizations on the \PCAST graphs here, where a detailed discussion is the contents of \companion. The optimizations are primarily concerned with the interfaces between unitary and non-unitary elements of the graph, and how these optimizations are leveraged depends on the desired outcome of the quantum program, release or hold.
\begin{comment}
\begin{enumerate}
    \item A \emph{hold outcome} indicates that the quantum state is in part or in whole a desired output of the program, so the state of the system is maintained after the program is complete.
    \item A \emph{release outcome} indicates that the only desired outcomes are the measurements, in which case once all measurement results have been recorded, the quantum state can be ``released''. Formally we define it as:
\end{enumerate}
\begin{definition} %\albert{alternate to the alternative definition in terms of cq-states}
\label{defn:releaseoutcome}
    A release outcome is an equivalence class of Pauli maps defined relative to a classical measurement space $\mathcal M$, whereby $C_1 \simeq C_2$ if and only if for every Pauli $P$, if $C_1(P) = \bigoplus_{m' \in \mathcal M_1} \cqstate{\sigma^1_{m'}}{\rho^1_{m'}}$ and $C_2(P)= \bigoplus_{m'' \in \mathcal M_2} \cqstate{\sigma^2_{m''}}{\rho^2_{m''}}$, then there exists surjective functions $ f_1: \mathcal M_1 \to \mathcal M$ and  $ f_2: \mathcal M_2 \to \mathcal M$ such that for all $m \in \mathcal M$ $\sum_{m' \in f_1^{-1}(\{m\})} \tr(\rho^1_{m'}) = \sum_{m'' \in f_2^{-1}(\{m\})} \tr(\rho^2_{m''})$.
\end{definition}
\end{comment}
We give a brief outline of the optimizations here, and illustrate them in \cref{fig:pcastopt}.
\begin{enumerate}[leftmargin=*]
\item In both cases, reduce the size of the support for nodes which are either weakly or not at all dependent (as defined in \companion) on a preparation (\cref{fig:pcastopta}).
\item For a \emph{hold} outcome, reduce the cost of the terminating frame in the presence of preparations with outdegree 0 in the graph (\cref{fig:pcastoptb}).
\item For a \emph{release} outcome (\cref{fig:pcastoptc}):
\begin{enumerate}
\item Trivialize the terminating frame and remove nodes that can be commuted past all measurements.
\item If there are preparation nodes with outdegree 0,
\begin{enumerate}
\item Replace all measurement nodes with outdegree 0 with an equivalent and more efficient set of measurements and a terminal measurement map via a \emph{stabilizer search}~\companion.
\item reduce the cost of the the frame generated by 3(b.i) in the presence of preparations with outdegree 0.
\item  convert the graph back to frame-terminating.
\end{enumerate}
\end{enumerate}
\end{enumerate}

\begin{figure*} 

  \centering

\begin{minipage}{0.335\textwidth}
  \subfloat[Reducing node support in the presence of a preparation node (top) is the equivalent of removing a control-line when the control is initialized via a preparation (bottom).]{  \label{fig:pcastopta}
    \begin{tikzpicture} 
      \node[scale=0.6] at (0,0) {
        \begin{tikzcd}[row sep=0.2cm]
          \startingpoprgate{\Prep{Z_0}{X_0}} \arrow[r] 
            &\poprgate{\Rot{Z_0 Z_1}{\theta/2}}
          \\
          &\startingpoprgate{\Rot{Z_1}{\theta/2}}
      \end{tikzcd}
      };
      
      \node at (2, 0) {\Large $\equiv$};
      
      \node[scale=0.6] at (3,0) {
        \begin{tikzcd}[row sep=0.2cm]
          \poprgate{\Prep{Z_0}{X_0}}
          \\
          \startingpoprgate{\Rot{Z_1}{\theta}}
        \end{tikzcd}
      };

      \node[scale=0.6] at (0, -1.25) {
        \begin{tikzcd}[row sep=0.2cm]
          & \gate{\PrepZ} & \octrl{1} & \qw \\
          & \qw & \gate{\RZ(\theta)} & \qw
        \end{tikzcd}
      };
      \node at (2, -1.25) {\Large $\equiv$};
      \node[scale=0.6] at (3, -1.25) {
        \begin{tikzcd}[row sep=0.2cm]
          & \gate{\PrepZ} & \qw \\ 
          & \gate{\RZ(\theta)} & \qw
        \end{tikzcd}
      };
    \end{tikzpicture}
  }
  \subfloat[Reducing the cost of the Pauli Frame in the presence of a preparation (top) is the equivalent of removing a CZ when the control line is initialized via a preparation (bottom).]{ \label{fig:pcastoptb}
    \begin{tikzpicture}
      \node[scale=0.6] at (0,0) {
        \begin{tikzcd}[row sep=0.2cm]
          \startingpoprgate{\Prep{Z_0}{X_0}} \arrow[r]
          & \startingframegate{1}{\begin{pmatrix}
              Z_0 & X_0 Z_1 \\
              Z_1 & Z_0 X_1
            \end{pmatrix}}
        \end{tikzcd}
      };

      \node at (2, 0) {\Large $\equiv$};
      
      \node[scale=0.6] at (3,0) {
        \begin{tikzcd}[row sep=0.2cm]
          \poprgate{\Prep{Z_0}{X_0}} \\
           \startingframegate{1}{\begin{pmatrix}
                Z_0 & X_0 \\
                Z_1 & X_1
          \end{pmatrix}}
        \end{tikzcd}
      }; 

      \node[scale=0.6] at (0.5, -1.25) {
        \begin{tikzcd}[row sep=0.2cm]
          & \gate{\PrepZ} & \ctrl{1} & \qw \\
          & \qw & \ctrl{} & \qw &
        \end{tikzcd}
      };
      \node at (2, -1.25) {\Large $\equiv$};
      \node[scale=0.6] at (3, -1.25) {
        \begin{tikzcd}[row sep=0.3cm]
          & \gate{\PrepZ} & \qw \\ 
          & \qw & \qw
        \end{tikzcd}
      };
    \end{tikzpicture}
  }
\end{minipage}
\begin{minipage}{0.65\textwidth}
\vfill
  \subfloat[Sequence of transformations applied in the case a release outcome.]{ \label{fig:pcastoptc}
   \begin{tikzpicture}
     \node[scale=0.55, draw, thick, rounded corners, fill=blue!20] (g0) at (-0.25,2) {
       \begin{tikzcd}[row sep=0.2cm, solid]
           &\startingpoprgate{\Meas[c_0]{Y_0 X_1 Z_2}} \arrow[rd]
         \\
         \startingpoprgate{\Prep{X_0}{Y_0}} \arrow[r] \arrow[rd] \arrow[ru]
         &\poprgate{\Rot{Z_0 Z_1}{\theta_1}} \arrow[r]
         & \startingframegate{1}{\begin{pmatrix}
                Z_0 X_1 & Z_1 \\
                Z_0 & X_0 Z_1 \\
                Z_2 & X_2
           \end{pmatrix}}
         \\
         &\startingpoprgate{\Meas[c_1]{X_0 Y_1 Z_2}} \arrow[ru]
         %\\
         %&\startingpoprgate{\Meas[c_1]{Z_2}}
       \end{tikzcd}
     };

     \node[scale=0.55, draw, thick, rounded corners, fill=blue!20] (g1) at (5.75, 1.75) {
       \begin{tikzcd}[row sep=0.2cm]
           &\startingpoprgate{\Meas[c_0]{Y_0 X_1 Z_2}} \arrow[rd]
         \\
         \poprgate{\Prep{X_0}{Y_0}} \arrow[rd] \arrow[ru]
         &
           & \startingframegate{1}{\begin{pmatrix}
                Z_0 & X_0 \\
                Z_1 & X_1 \\
                Z_2 & X_2
           \end{pmatrix}}
         \\
         &\startingpoprgate{\Meas[c_1]{X_0 Y_1 Z_2}} \arrow[ru]
       \end{tikzcd}
     };

     \node[scale=0.55, draw, thick, rounded corners, fill=blue!20] (g2) at (0, -1) {
       \begin{tikzcd}[row sep=0.2cm]
           & &\startingpoprgate{\Meas[c_0']{Y_0}} \arrow[rd]
         \\
         \poprgate{\Prep{X_0}{Y_0}} \arrow[r]
         & \startingframegate{1}{\begin{pmatrix}
                Z_0 X_1 & X_0 \\
                Z_1 X_0 & X_1 \\
                Z_2 & X_2
           \end{pmatrix}} \arrow[ru] \arrow[rd]
         &\startingpoprgate{\Meas[c']{Z_2}} \arrow[r]
         & \startingframegate{1}{\begin{aligned}
            &\cAssign{c_0}{c_0' + c'}; \\
            &\cAssign{c_1}{c_1' + c'}
        \end{aligned} }
         \\
         & &\startingpoprgate{\Meas[c_1']{Y_1}} \arrow[ru]
       \end{tikzcd}
     };

%     \node[scale=0.7] at (0, -10.5) {
%       \begin{tikzcd}[row sep=0.2cm]
%           & &\startingpoprgate{\Meas[c_0]{Y_0}}
%         \\
%        \poprgate{\Prep{X_0}{Y_0}} \arrow[r]
%         & \startingframegate{1}{\begin{pmatrix}
%                Y_0 & X_0 \\
%                Z_1 & X_1 \\
%                Z_2 & X_2
%           \end{pmatrix}} \arrow[ru]
%         &\startingpoprgate{\Meas[c_0 + c_1 + c_2]{Z_2}}
%         \\
%         & &\startingpoprgate{\Meas[c_2]{Y_1}}
%       \end{tikzcd}
%     };

%     \node at (0, -12.25) {\Large $\simeq (\mod c_0 \otimes c_1$) };
% prev -14
     \node[scale=0.55, draw, thick, rounded corners, fill=blue!20] (g3) at (6, -3) {
       \begin{tikzcd}[row sep=0.2cm]
           &\startingpoprgate{\Meas[c_0']{-Z_0}} \arrow[rd]
           & \startingframegate{1}{\begin{pmatrix}
                Z_0 & X_1 \\
                Z_1 & X_1 \\
                Z_2 & X_2
           \end{pmatrix}}
         \\
         \poprgate{\Prep{X_0}{Y_0}} \arrow[ru]
         &\startingpoprgate{\Meas[c']{Z_2}} \arrow[r]
         & \startingframegate{1}{\begin{aligned}
            &\cAssign{c_0}{c_0' + c'}; \\
            &\cAssign{c_1}{c_1' + c'}
        \end{aligned} }
         \\
         &\startingpoprgate{\Meas[c_2']{Y_1}} \arrow[ru]
       \end{tikzcd}
     };

     \path[->, line width = 0.3mm] (g0) edge node [above] {$\equiv^\hold$} (g1);
     \path[->, line width=0.3mm, pos=0] (g1) edge node [below] {~~$\equiv^\release$} (g2);
     \path[->, line width=0.3mm, pos=0] (g2) edge node [below] {$\equiv^\hold$} (g3);
     
    \end{tikzpicture}
  }
  \end{minipage}

\caption{Examples of the \PCAST graph optimizations.}
  \label{fig:pcastopt}
\end{figure*}

\begin{theorem}[\companion]
    If $G$ is optimized to $G'$ with outcome $o$, then $\termof{G} \equiv^o \termof{G'}$.
\end{theorem}

\subsection{Synthesizing circuits from \PCAST graphs}
\label{sec:synthesis}

After optimizing the \PCAST graph, the next step is to synthesize an equivalent circuit. In this section we summarize the method outlined in Ref. \cite{Schmitz2021} along with the modifications needed to adapt it to \PCAST.

We will always start with a frame-terminal, fully merged \PCAST graph $(G, F)$.\ifExtended{ The process is performed in two parts:
\begin{enumerate}
\item Synthesize the non-Clifford elements of $G$ up to a residual Clifford unitary. In the \release case, we also defer synthesis of all measurements topologically at the end of the DAG.
\item Finalize the synthesis:
\begin{enumerate}
    \item In the \hold case, transform the original terminating frame by the residual frame from step 1 and synthesize an equivalent Clifford circuit.
    \item In the \release case, transform the end measurements by the residual frame from step 1 and perform a \emph{stabilizer search}~\companion to simultaneously produce a set of \release equivalent measurements, a Clifford circuit to realize them, and a measurement space function to complete the \release equivalence.
\end{enumerate}
\end{enumerate}
}{} 
The process starts by prepending the graph with an empty circuit and iteratively synthesizing gates from dependency-free nodes. If a node has a direct equivalent as a gate, such as $\Prep{Z_i}{X_i}$, that gate is added directly. If not, we select a Clifford gate $g$ to add to the circuit, and transform the rest of the graph by $F'=\frameof{g^{-1}}$ to obtain
\begin{align} \label{eq:gatepush}
    C; g; F'; G; F
    \equiv
    C; g; \fwdAction{F'}(G); F \circ F'
\end{align}
where $\fwdAction{F'}(G)$ applies $\fwdAction{F'}(n)$ to each node in $G$.

Thus synthesis is primarily a method for selecting the sequence of Clifford gates, particularly two-qubit entangling (TQE) gates, that minimizes the gate cost, depth or other cost metric of the resulting circuit.\footnote{For a given qubit pair there are $3 \times 3 = 9$ TQE gates, generalizing CNOT and CZ, such that we have one Pauli operator basis for each qubit~\cite{Schmitz2021}.} To do so, we adapt the search algorithm of \cite{Schmitz2021} with what we call \emph{search functions}.
\begin{enumerate}
    \item $\nodeCost(n)$ returns the cost to implement $n$, which is zero if and only if a gate $\gateof{n}$ implements it.
    \item $\reduceNode(n)$ returns a set of TQE gates $g$ that minimize the $\nodeCost$ of $\fwdAction{\frameof{g}}(n)$. \ifExtended{It promises \emph{monotonicity} with respects to $\nodeCost$, i.e. the repeated process of applying any sequence of gates returned by the function will monotonically reduce $\nodeCost$ to zero.}{Repeated application will monotonically reduce $\nodeCost$ to 0.}
    %It promises \emph{monotonicity} with respects to $\nodeCost$, i.e. the repeated process of apply any sequence of gates returned by the function will monotonically reduce $\nodeCost$ to zero.
    \item $\gateCost(C,G, F, g)$ returns the cost of the action of Eq.~\eqref{eq:gatepush}.
    \item $\addGate(C, G, F, g)$ performs the action of Eq.~\eqref{eq:gatepush}.
\end{enumerate}

We start with synthesizing non-Clifford nodes (\cref{algo:greedy}). \ifExtended{The algorithm depicted in \cref{algo:greedy} has two important properties which we prove in the following lemma:
\begin{lemma}
  The search algorithm of \cref{algo:greedy} terminates and satisfies $\termof{C}; \termof{G'}; \termof{F'} \equiv \termof{G}; \termof{F}$ for its inputs and outputs.
\end{lemma}

\begin{proof}
To prove the algorithm of \cref{algo:greedy} terminates, consider lines~\ref{line:getmins}-\ref{line:addgate} for any iteration of the main loop. Regardless of what gate is returned in line~\ref{line:getgate}, its application in line~\ref{line:addgate} reduces $\nodeCost$ for at least one element $n \in \text{Min}$ by definition of the search functions. At the top of the next iteration, either $\nodeCost(n) = 0$ and the node is removed from $G'$, or it is necessarily found in Min again and thus possibly reduced by the next TQE gate added to the circuit. Thus by the monotonicity property of the $\reduceNode$ function, $n$ is eventually reduced to $\nodeCost(n) = 0$ and removed from $G'$. Thus a finite number of iterations of the main loop always removes at least one element which implies a finite number of iterations of the main loop removes all relevant nodes related to the exit condition. Therefore, the algorithm terminates.

As for the second property, each call to $\addGate$ satisfies \cref{eq:gatepush} or is called to add a gate whose equivalent is removed from $G'$ and therefore every modification to $C, G'$ and $F'$ satisfies $\termof{C}; \termof{G'}; \termof{F'} \equiv \termof{G}; \termof{F}$ including when returned.
\end{proof}
}{The search always terminates due to the search functions' guarantees~\cite{Schmitz2021}. It should be clear that $\termof{C}; \termof{G'}; \termof{F'} \equiv \termof{G}$, as $\addGate$ preforms a logical identity between the graph and the circuit.
%In the \hold case (\release case), we delay synthesizing the terminating frame (end measurements) as these will be synthesized separately.
}

\begin{figure}
    \begin{algorithmic}[1]
    \Procedure{SearchNonClifford}{$G, F$, \hold /\release}
        \State $C \gets \emptyset$, $G' \gets G$, $F' \gets F$
        % \While {(\hold) $G'$ contains non-Clifford nodes OR
        
        %     \qquad (\release) $G'$ contains non-measurement nodes}
        \While {$G'$ contains (\release) non-measurement /
        
                \qquad \qquad \qquad~~~ (\hold) non-Clifford nodes}
            \For{$n \in \BEGIN(G')$}
                \If{$\nodeCost(n)=0$}
                  \If{(\hold) OR $n \notin \END(G')$ OR 
                     \State \quad $n$ not a measurement node
                     }
                    \State $\Call{addGate}{C, G', F', \gateof{n}}$
                    \State $G' \gets$ \Call{\REMOVEVERTEX}{$G', n$}
                  \EndIf
                \EndIf
            \EndFor
            \State Min $\gets \argmin_{n \in \BEGIN(G')} \nodeCost(n)$ \label{line:getmins}
            \State MinGate $\gets \reduceNode(\text{Min})$ \label{line:getgates}
            \State $g_{min} \gets \argmin_{g \in \text{MinGate}} \gateCost(C, G', F', g)$ \label{line:getgate}
            \State $\addGate(C, G', F', g_{min})$ \label{line:addgate}
        \EndWhile
        \State \Return $(C, G', F')$
    \EndProcedure
    \end{algorithmic}
    \caption{Ultra Greedy Search Algorithm for Non-Clifford nodes of a \PCAST graph $G$. Returns a circuit $C$ and terminating frame $F'$ such that $\termof{C}; \termof{G'}; \termof{F'} \equiv \termof{G}$. $\BEGIN(G)$ is the set of dependency-free non-Clifford nodes in $G$ and $\END(G)$ the set of nodes with outdegree 0. }
    \label{algo:greedy}
\end{figure}

\subsubsection{Search function implementations}
Search functions can be implemented differently based on complexity and cost.
%Here we discuss those currently implemented in the Intel Quantum SDK, with future work discussed in Section~\ref{sec:conc}.
The basic implementation in the Intel Quantum SDK primarily considers the minimum number of TQE gates required to reduce the node cost to zero. This cost depends on whether the node is a \emph{singlet} or a \emph{factor} node. 

Singlet nodes $n(P)$ are those defined by a single Pauli operator such as $\Rot{P}{\theta}$ or $\Meas[c]{P}$.\ifExtended{Ref.~\cite{Schmitz2021} argues that the minimum number of TQE gates to reduce a singlet node to single-qubit support is}{}
\begin{align}
\nodeCost(n(P))  = \support(P) - 1.
\end{align}
%as argued in Ref.~\cite{Schmitz2021}.

Factor nodes $n(P, Q)$ are defined by two anti-commuting Pauli operators. We can understand factor nodes as \emph{effective qubits} 
%consisting of
\ifExtended{(partial dimension-2 factorization of the Hilbert space)}{} consisting of
an effective $Z$ and effective $X$, much like the rows of the Pauli frame. 
\ifExtended{We have introduced two factor nodes, $\Prep{P_Z}{P_X}$ and $\Rot2{P_Z}{P_X}[\theta][\phi]$, and}{Thus far we have only introduced one type of factor node, $\Prep{P_Z}{P_X}$, but} to synthesize Clifford circuits in \cref{sec:cliffordSynthesis}, we break an $n$-qubit Pauli frame into $n$ factor nodes, one for each row, called \emph{local frames}.\ifExtended{}{\footnote{In the implementation, we also have 2-axis rotations $\Rot2{P_1}{P_2}[\theta][\phi]$, the generalization of the Intel Quantum SDK's $\RXY$ gate; see \cite{PaykinSchmitz2023PCOAST}.}}

%To understand the minimum number of TQE gates to reduce factor nodes, i.e. reduce the effective qubit to an actual one, we introduce the \emph{local support matrix} for any anti-commuting pair $(P, Q)$ and qubit $i$ as
The node cost of a factor node is the number of TQE gates needed to reduce the effective qubit to an actual one.\ifExtended{ For this purpose we define the}{The} \emph{local support matrix} of an anti-commuting pair $P, Q$ and qubit $i$ \ifExtended{as}{is}
\begin{align}
\support_i(P, Q) =
\begin{pmatrix}
    \lambda(P, X_i) & \lambda(P, Z_i)  \\
     \lambda(Q, X_i) & \lambda(Q, Z_i) 
\end{pmatrix},
\end{align}
The determinant (mod 2) of the local support matrix is called the \emph{local determinant}. 
We say a node has \emph{weak support} on qubit $i$ if its local support matrix is nonzero, but has zero determinant, and \emph{strong support} if its local determinant is~1.
%For any factor node whose local determinant is 1 on qubit $i$, we say that node has \emph{strong support} on qubit $i$ and likewise if the local support matrix is not zero, but has zero determinant, we say that node has \emph{weak support} on qubit $i$. 
\ifExtended{The following lemma}{Ref.~\cite{PaykinSchmitz2023PCOAST}} argues that for any factor node, the sum over all local determinants is $1 \mod 2$.
\begin{ExtendedVersion}
\begin{lemma} \label{lemma:sumDet}
For a factor noded defined by anti-commuting pair $(P,Q)$, 
\begin{align}
\sum_i \det(\support_i(P,Q)) \mod 2 = 1.
\end{align}
\end{lemma}

\begin{proof}
  By definition, anti-commutativity implies
  \begin{align}
    1 &= \lambda(P,Q) \nonumber \\
      &= \sum_i \left( \lambda(P, X_i)  \lambda(Q, Z_i) + \lambda(P, Z_i) \lambda(Q, X_i) \right) \mod 2 \nonumber \\
      &= \sum_i \det(\support_i(P,Q))\mod 2
  \end{align}
\end{proof}

  To use the local support, we think of it as a dimension 2 binary column vector of dimension 2 binary row vectors (and thus all addition below is assumed to be binary),
  \begin{align}
      \support_i(P,Q) = \begin{pmatrix} \label{eq:lsm_vec}
    v^Z_i  &
    v^X_i 
\end{pmatrix}_i,
\end{align}
where we also define $v^Y_i = v^Z_i + v^X_i$ and the constant dimension 2 vectors $\underbar x = (0, 1)^\top, \underbar y = (1,1)^\top$ and $\underbar z = (1, 0)^\top$. We then consider the action of a TQE gate on the local support matrix when the node $n(P, Q)$ is conjugated by it. Each TQE gate on qubits $i$ and $j$ is defined by its Pauli operator basis for $i$ and $j$. For example $\CX_{ij}$ is $Z$ basis on $i$ and $X$ on $j$ and $\CZ_{ij}$ is $Z$ basis on both. We use the notation for TQE gates $g_{ij} = (\sigma_1,\sigma_2)_{i,j}$ where $\sigma_{k} \in \{X, Y, Z\}$, to represent these bases. As such, it is easy to show that under the action of the TQE gate, the local support matrices are transformed as follows:
\begin{lemma}\label{lemma:lsm_trans}
For a factor node defined by anti-commuting pair $(P,Q)$, and any TQE gate defined by the pair $(\sigma_1, \sigma_2)_{i,j}$ for qubits $i, j$, the action of the TQE gate on the local frames is given by
  \begin{align}\label{eq:lsm_trans}
\support(P,Q)_i \xrightarrow{(\sigma_1, \sigma_2)_{ij}}
\support(P,Q)_i + & 
\lambda(\sigma_1, X)\begin{pmatrix}
    v^{\sigma_2}_j  &
    0 
\end{pmatrix} \nonumber \\
+& \lambda(\sigma_1, Z)\begin{pmatrix}
    0 &
    v^{\sigma_2}_j
\end{pmatrix},
\end{align}

 where $v_i^\sigma$ is as defined in Eq. \eqref{eq:lsm_vec}. The analogous formula also applies for for the $j^{th}$ support matrix, and all others are left the same.
\end{lemma}

\begin{proof}
    Note conjugation of a single-qubit Pauli operator $\sigma_i$ by TQE gate $g_{ij} = (\sigma_1, \sigma_2)_{ij}$ is given by
    \begin{align}
      \fwdAction F^{g_{ij}}(\sigma_i) = \sigma_i * (\sigma_2)_j^{\lambda(\sigma_1, \sigma_i)},
    \end{align}
    and similarly for $\sigma_j$. Noting that $g_{ij}$ is self-inverse and using the properties of $\lambda$~\companion, we find that
    \begin{align}
      \lambda(\fwdAction F^{g_{ij}}(P), X_i) =& \lambda(P, \fwdAction F^{g_{ij}}(X_i)) \nonumber \\
      =& \lambda(P, X_i) + \lambda(\sigma_1, X) \lambda(P, (\sigma_2)_j),
    \end{align}
    and likewise for $Q$, and $Z_i$, Therefore when put altogether we find Eq. \eqref{eq:lsm_trans}. 
\end{proof}
Note that Eq. \eqref{eq:lsm_trans} is linear and the $2 \times 2$ matrix we add to the local support matrix as a result of the transformation is of weak support, i.e. zero determinate. It should also be clear that for every determinate 1 local support matrix, i.e. such as when a qubit $j$ has strong support, each constant vector $\underbar x, \underbar y$ and $\underbar z$ is represented as one of the column vectors $v_j^X, v_j^Y$ and $v_j^Z$. Thus a corollary of Lemma~\ref{lemma:lsm_trans} is one can add, via the action of a TQE gate, any weak support matrix to any local support matrix on qubit $i$ using another qubit $j$ which has strong support, which by virtue of Lemma~\ref{lemma:sumDet}, there is alway at lease one such qubit. By the same lemma, if qubits $i$ and $j$ have strong support, reducing one to weak support must also reduce the other to weak support, but we can not fully remove support from either because the matrix we add to the $i (j)$ support matrix is always determinate 0. By simply counting, we find that there are six possible determinate 0 matrices which can reduce a given determinate 1 matrix to determinate 0 under addition, and thus six TQE gates which can reduce a pair of strong support qubits to weak support qubits. Likewise, there is trivially only one determinate 0 matrix which can reduce itself to zero under addition, and thus exactly one TQE gate which can reduce a weak support qubit to no support using a strong support qubit (which must remain strong by virtue of Lemma~\ref{lemma:sumDet}).\footnote{There do exists cases where a TQE gate can reduce two qubits with weak support to one with weak support, but it is not guaranteed.}  

\end{ExtendedVersion}
\begin{ShortVersion}
%We also argue for any factor node and pair of qubits, there exists six TQE gates to reduce two strong supports to two weak supports and exactly one to reduce one strong and one weak support to no support on the weak support qubit.\footnote{We note here that there does exists cases where a TQE gate can reduce two qubits with weak support to one qubit with weak support, but it is not guaranteed in all cases, so we ignore this possibility without violating the promises made by the search functions.} 
In addition, there exist six TQE gates to reduce any two qubits with strong support to two with weak support, and exactly one to reduce one strong and one weak support to no support on the weak support qubit.\footnote{There do exists cases where a TQE gate can reduce two qubits with weak support to one with weak support, but it is not guaranteed.}
\end{ShortVersion}
%Reducing a Pauli factor node is then the process of reducing all strong support to a single qubit (which will constitute the final qubit on which the node is reduced to) and reducing all weak support (included those which were formerly strong support) to no support. This results in the definition of $\nodeCost$,
Reducing a Pauli factor node is then the process of reducing all strong support to a single qubit (the final qubit to which the node is reduced) and eliminating all weak support. Thus
\begin{align}
\nodeCost(n(P,Q)) =& \left(\sum_i \det(\support_i(P,Q)) - 1\right)/2
\nonumber\\ 
&+ \left( \sum_i \left[\support_i(P,Q) \neq 0\right] - 1 \right)
\end{align}
\ifExtended{
The first term represents the number of TQE gates to reduce strong support to one and the second term is the number of TQE gates to reduce all weak support to no support.
%From this discussion, the $\reduceNode$ function follows immediately\cite{PaykinSchmitz2023PCOAST}.
The implementation of $\reduceNode$ follows as an application of Eq. \eqref{eq:lsm_trans} to all strong-strong and strong-weak qubit pair such that the gate action results in weak-weak and stong-no support, respectively.
}{
The implementation of $\reduceNode$ follows from the discussion in \cite{PaykinSchmitz2023PCOAST} around the selection of gates which reduce strong-strong and strong-weak qubit pairs.
}

%The final two search functions are straightforward. 
$\gateCost$ is the average change in $\nodeCost$ over all remaining nodes.%
%\footnote{More nuanced methods are the contents of future work as discussed Section~\ref{sec:conc}. However one simple modification is to give greater weight to those nodes which are free of dependencies, as studied in Section~\ref{sec:results}. The weighting used there is to weight free nodes by a factor proportional to the total number of nodes in the graph divided by the total number of free nodes. This allows the weighting to have diminishing impact as the synthesis proceeds.\label{foot:rflag}} 
\footnote{One simple modification to $\gateCost$ is to give greater weight to dependency-free nodes, as discussed in Section~\ref{sec:results}. There, we weight free nodes proportionally to the total number of nodes in the graph divided by the number of free nodes, meaning that weighting has diminishing impact as synthesis proceeds. \label{foot:rflag}}
We use average change as opposed to average absolute cost as it is faster to compute. $\addGate$ produces the equivalent change to the \PCAST graph.

%All other versions of the search functions are variations on the basic version. In the case of the Intel Quantum SDK, other considerations include:
The Intel Quantum SDK implements a variation of the basic cost functions, incorporating the following:
\begin{itemize}
    \item A ``schedule'' search maintains an approximate ASAP scheduling of the circuit. $\gateCost$ includes a parallelization credit with a tuneable parameter~\cite{Schmitz2021}.\footnote{The value of this is a fine-tuning modification studied in \cref{sec:results}. \label{foot:Cflag}}
    
    \item A ``native gate'' search targeting a particular gate set, which informs circuit cost. In the case of the the Intel Quantum SDK, the set is $\{\CZ, \RXY, \MeasZ, \PrepZ\}$. 

\begin{ExtendedVersion}
    As such, we want to avoid the generation of $\RZ$ gates, and use the known decomposition of non-native gates, especially TQE gates, at the point of synthesis. This requires a modification primarily to the search function $\addGate$. $\addGate$ replaces a TQE gate with an asymmetric equivalent\footnotemark sequence of gates, namely at most one single-qubit Clifford on each qubit, followed by a $\CZ$. \footnotetext{Equivalence is with respects to the Pauli Frame equivalence classes of equally-entangled frames; see Ref. \cite{Schmitz2021}.} The asymmetric equivalent provides the same reduction for $\nodeCost$, but also guarantees that for the sake of reducing multi-qubit rotation nodes, the node will not be translated to the circuit as an $\RZ$ as it commutes with the $\CZ$. Though this reduces the chances that synthesis results in an $\RZ$ gate, it doesn't remove the possibility entirely. So when an $\RZ$ is passed to $\addGate$, the following decomposition is used:
    \begin{align}
    \RZ_q(\theta) = \X_q \RXY_q(\pi, -\theta/2)
    \end{align}
    where the $\RXY$ is added to the circuit and the $X$ is added to graph to be eventually absorbed by the terminating frame. Similarly, when a gate-equivalent Meas or Prep is passed to $\addGate$ which is not in the $Z$ basis, only half of the appropriate Clifford conjugation along with the $\MeasZ$ or $\PrepZ$ gate is added to the circuit while the second half is added to the graph and terminating frame.
\end{ExtendedVersion}

    \item $\gateCost$ includes a cost penalty for the action of the gate on the terminating frame. 
\begin{ExtendedVersion}
    An average change in cost over the rows of the Pauli frame--interpreted as local frame nodes--is added to penalize the generation of costly Pauli frames to synthesize (see \cref{sec:cliffordSynthesis}
\end{ExtendedVersion}

    %\item All search function described here can be used with the stabilizer search template to implement the \emph{Stabilizer Search Algorithm} of \companion as described therein.
    \item The search functions described here are also used to implement the stabilizer search algorithm of \companion%
    \ifExtended{ via the \emph{stabilizer search template}}{}.

\end{itemize}

%\subsubsection{Adaptation of Ultra-Greedy Search for Finalizing the Circuit}
\subsubsection{Adapting Ultra-Greedy Search to finalize the circuit}
\label{sec:cliffordSynthesis}

%After the synthesis of the non-Clifford nodes, the remaining logic $(G', F')$ as returned by \cref{algo:greedy} is either such that $G'$ is empty and $F'$ is non-trivial in the ``hold'' case or $G'$ contains only mutually commuting measurements and $F'$ can be discarded in the ``release'' case.
After synthesizing the non-Clifford nodes, \cref{algo:greedy} leaves us with $(G', F')$, where either (\hold) $G'$ is empty and $F'$ is non-trivial, or (\release) $G'$ contains only mutually commuting measurements and $F'$ can be discarded.

%In the \hold case, methods exist for synthesizing a circuit for a Pauli frame/tableau\cite{Aaronson2004, Maslov2018}. However, we leverage \cref{algo:greedy} for this task because doing so will automatically take into account all the considerations of the search functions.
In the \hold case, while methods exist for synthesizing a circuit for a Pauli frame/tableau~\cite{Aaronson2004, Maslov2018}, we leverage the ultra-greedy algorithm to automatically take into account all search function considerations.
%As such, Pauli frame synthesis is described exactly as that of \cref{algo:greedy} with the rows of the Pauli frame re-interpreted as a set of mutually independent local-frame factor nodes, and a zero-cost node being translated to the single-qubit Clifford it represents.
As such, Pauli frame synthesis is implemented by \cref{algo:greedy} with the rows of the Pauli frame re-interpreted as a set of mutually independent local-frame factor nodes, and a zero-cost node replacing the Clifford it represents.
With a few additional promises made by the search functions (see \companion for details), it is guaranteed that when one local frame is reduced, it is independent of all other local frames. %(i.e. single-qubit operators for the real qubit the local frame is reduced to, cannot appear in any other local frame), as conjugation by Clifford gates cannot change the commutativity relations between local frames. 
%We then consider if the local frame is reduced to the same or a different qubit than its row index. When different, this represents a swapping of the qubit indices. In some cases
%this swapping is considered free as it can be implemented virtually. However, there are circumstances where this qubit swapping must be performed by gates, so future versions may include the swap cost.
If the local frame is reduced to a different qubit than its row index, it is implemented by a swap. In some cases
this swap can be implemented virtually and thus considered free, but future versions may include the swap cost when qubit swapping must be performed by gates.

In the \release case, we can again adapt \cref{algo:greedy} for the measurements in $G'$ by replacing the search functions with the stabilizer-search-templated versions and adding any measurement space function generated by the search. As measurement space functions are (binary) linear, they add no more than a polynomial-time cost (in number of measurements) to the overall quantum-classical computation \companion.\footnote{Measurement space functions are handled automatically in the Intel Quantum SDK: classical instructions are generated in the LLVM IR to appropriately map fresh measurement outcomes to classical variables.}

\begin{ExtendedVersion}
\subsubsection{Time Complexity of the Ultra-Greedy Search Algorithm}
\label{sec:timeComplexity}

As argued and demonstrated in Ref. \cite{Schmitz2021}, the time-complexity of the ultra-greedy search algorithm is  $\mathcal{O}(N^3|G|^2)$, where $N$ is the number of qubits and $|G|$ is the number of nodes in the graph; we expect roughly the same scaling or better here. This scaling can be dramatically improved (or made worse) based on the complexity of implementing the search functions where the primary driver of that cost is often $\gateCost$. If $\gateCost$ function complexity is reduced to $\mathcal O(1)$, we find a floor for ultra greedy search of $\Omega(N|G|)$) without affecting the two most important properties, algorithm termination and outcome correctness, albeit at the cost of circuit performance.
\end{ExtendedVersion}

\section{Evaluation}
\label{sec:results}
%\input{src/tables/counts}
%\input{src/tables/counts_v2}
%\input{src/tables/counts_v3}
%\input{src/tables/counts_v4}
% \usepackage{booktabs}
% \usepackage{graphicx}
\begin{table*}[]
\caption{Results for total gate count, two-qubit gates, and depth compared with Qiskit and \tket. The coloring indicates the range from best (dark green, bold) to worst (white).}
\label{tab:gate_counts_combined}
\resizebox{\textwidth}{!}{%
\begin{tabular}{@{}ccccccccclcccccclcccccc@{}}
\toprule
\multicolumn{1}{l}{} &
  \multicolumn{1}{l}{} &
  \multicolumn{1}{l}{} &
  \multicolumn{6}{c}{\textbf{Gate Count}} &
   &
  \multicolumn{6}{c}{\textbf{2Q Gates}} &
   &
  \multicolumn{6}{c}{\textbf{Depth}} \\
  \cmidrule(lr){4-9}
  \cmidrule(lr){11-16}
  \cmidrule(lr){18-23}
\textbf{Benchmark} &
  \textbf{N} &
  \textbf{} &
  \textbf{\PCOASTone} &
  \textbf{\PCOASTFT} &
  \textbf{\qiskittwo} &
  \textbf{\qiskitthree} &
  \textbf{\tketone} &
  \textbf{\tkettwo} &
   &
  \textbf{\PCOASTone} &
  \textbf{\PCOASTFT} &
  \textbf{\qiskittwo} &
  \textbf{\qiskitthree} &
  \textbf{\tketone} &
  \textbf{\tkettwo} &
   &
  \textbf{\PCOASTone} &
  \textbf{\PCOASTFT} &
  \textbf{\qiskittwo} &
  \textbf{\qiskitthree} &
  \textbf{\tketone} &
  \textbf{\tkettwo} \\
  \midrule
\textbf{H2\_BK} &
  4 &
   &
  \cellcolor[HTML]{00B85C}\textbf{47} &
  \cellcolor[HTML]{00B85C}\textbf{47} &
  \cellcolor[HTML]{FFFFFF}117 &
  \cellcolor[HTML]{FFFFFF}117 &
  \cellcolor[HTML]{CEEACF}79 &
  \cellcolor[HTML]{BBE5C4}76 &
   &
  \cellcolor[HTML]{00B85C}\textbf{13} &
  \cellcolor[HTML]{00B85C}\textbf{13} &
  \cellcolor[HTML]{FFFFFF}38 &
  \cellcolor[HTML]{F3F8F0}33 &
  \cellcolor[HTML]{5ACE8E}18 &
  \cellcolor[HTML]{48C984}17 &
   &
  \cellcolor[HTML]{00B85C}\textbf{31} &
  \cellcolor[HTML]{00B85C}\textbf{31} &
  \cellcolor[HTML]{FFFFFF}95 &
  \cellcolor[HTML]{F6FAF4}86 &
  \cellcolor[HTML]{E4F0DD}66 &
  \cellcolor[HTML]{CCE9CE}60 \\
\textbf{H2\_JW} &
  4 &
   &
  \cellcolor[HTML]{3EC77E}58 &
  \cellcolor[HTML]{00B85C}\textbf{36} &
  \cellcolor[HTML]{FFFFFF}195 &
  \cellcolor[HTML]{F4F8F1}165 &
  \cellcolor[HTML]{9CDEB3}91 &
  \cellcolor[HTML]{7FD7A3}81 &
   &
  \cellcolor[HTML]{47C984}19 &
  \cellcolor[HTML]{00B85C}\textbf{12} &
  \cellcolor[HTML]{FFFFFF}56 &
  \cellcolor[HTML]{E9F3E4}40 &
  \cellcolor[HTML]{52CC89}20 &
  \cellcolor[HTML]{29C272}16 &
   &
  \cellcolor[HTML]{11BC65}25 &
  \cellcolor[HTML]{00B85C}\textbf{21} &
  \cellcolor[HTML]{FFFFFF}126 &
  \cellcolor[HTML]{F2F7EE}103 &
  \cellcolor[HTML]{BDE6C5}65 &
  \cellcolor[HTML]{8EDAAB}54 \\
\textbf{H2\_PM} &
  4 &
   &
  \cellcolor[HTML]{00B85C}\textbf{53} &
  \cellcolor[HTML]{00B85C}\textbf{53} &
  \cellcolor[HTML]{FFFFFF}119 &
  \cellcolor[HTML]{F3F8F0}106 &
  \cellcolor[HTML]{82D7A4}72 &
  \cellcolor[HTML]{7BD5A0}71 &
   &
  \cellcolor[HTML]{00B85C}\textbf{15} &
  \cellcolor[HTML]{00B85C}\textbf{15} &
  \cellcolor[HTML]{FFFFFF}38 &
  \cellcolor[HTML]{EAF3E5}30 &
  \cellcolor[HTML]{13BC66}16 &
  \cellcolor[HTML]{00B85C}\textbf{15} &
   &
  \cellcolor[HTML]{00B85C}\textbf{25} &
  \cellcolor[HTML]{00B85C}\textbf{25} &
  \cellcolor[HTML]{FFFFFF}99 &
  \cellcolor[HTML]{F0F6EB}80 &
  \cellcolor[HTML]{E2EFDA}63 &
  \cellcolor[HTML]{C9E9CC}58 \\
\textbf{LiH\_BK} &
  12 &
   &
  \cellcolor[HTML]{0FBB64}3143 &
  \cellcolor[HTML]{00B85C}\textbf{2372} &
  \cellcolor[HTML]{FFFFFF}24985 &
  \cellcolor[HTML]{F4F9F2}21080 &
  \cellcolor[HTML]{C0E6C7}12017 &
  \cellcolor[HTML]{A8E0B9}10786 &
   &
  \cellcolor[HTML]{14BC67}1308 &
  \cellcolor[HTML]{00B85C}\textbf{964} &
  \cellcolor[HTML]{FFFFFF}8680 &
  \cellcolor[HTML]{F1F7ED}6834 &
  \cellcolor[HTML]{8DDAAB}3385 &
  \cellcolor[HTML]{6ED399}2858 &
   &
  \cellcolor[HTML]{03B85D}1157 &
  \cellcolor[HTML]{00B85C}\textbf{1027} &
  \cellcolor[HTML]{FFFFFF}18549 &
  \cellcolor[HTML]{F3F8F0}15111 &
  \cellcolor[HTML]{E2EFDA}9799 &
  \cellcolor[HTML]{C5E7C9}8665 \\
\textbf{LiH\_JW} &
  12 &
   &
  \cellcolor[HTML]{0CBB63}3173 &
  \cellcolor[HTML]{00B85C}\textbf{2647} &
  \cellcolor[HTML]{FFFFFF}21239 &
  \cellcolor[HTML]{FBFDFA}20083 &
  \cellcolor[HTML]{A5E0B8}9471 &
  \cellcolor[HTML]{87D9A7}8225 &
   &
  \cellcolor[HTML]{0FBB64}1295 &
  \cellcolor[HTML]{00B85C}\textbf{1062} &
  \cellcolor[HTML]{FFFFFF}8064 &
  \cellcolor[HTML]{F3F8F0}6666 &
  \cellcolor[HTML]{64D093}2616 &
  \cellcolor[HTML]{3FC77F}2040 &
   &
  \cellcolor[HTML]{03B85D}1194 &
  \cellcolor[HTML]{00B85C}\textbf{1070} &
  \cellcolor[HTML]{FFFFFF}17027 &
  \cellcolor[HTML]{F6FAF4}14787 &
  \cellcolor[HTML]{BAE5C3}7654 &
  \cellcolor[HTML]{94DCAE}6310 \\
\textbf{LiH\_PM} &
  12 &
   &
  \cellcolor[HTML]{12BC66}3218 &
  \cellcolor[HTML]{00B85C}\textbf{2402} &
  \cellcolor[HTML]{FFFFFF}22214 &
  \cellcolor[HTML]{FBFCFA}20897 &
  \cellcolor[HTML]{C4E7C9}11030 &
  \cellcolor[HTML]{A8E1BA}9795 &
   &
  \cellcolor[HTML]{1CBE6B}1324 &
  \cellcolor[HTML]{00B85C}\textbf{905} &
  \cellcolor[HTML]{FFFFFF}7640 &
  \cellcolor[HTML]{F8FBF6}6893 &
  \cellcolor[HTML]{92DBAD}3092 &
  \cellcolor[HTML]{6FD39A}2565 &
   &
  \cellcolor[HTML]{05B95F}1182 &
  \cellcolor[HTML]{00B85C}\textbf{999} &
  \cellcolor[HTML]{FFFFFF}16131 &
  \cellcolor[HTML]{F9FBF7}14674 &
  \cellcolor[HTML]{E4F0DD}9263 &
  \cellcolor[HTML]{D1EBD0}8021 \\
\textbf{BeH2\_BK} &
  14 &
   &
  \cellcolor[HTML]{0ABA62}7164 &
  \cellcolor[HTML]{00B85C}\textbf{6016} &
  \cellcolor[HTML]{FFFFFF}54201 &
  \cellcolor[HTML]{F8FBF6}48857 &
  \cellcolor[HTML]{BFE6C6}26465 &
  \cellcolor[HTML]{A7E0B9}23834 &
   &
  \cellcolor[HTML]{11BC65}3049 &
  \cellcolor[HTML]{00B85C}\textbf{2407} &
  \cellcolor[HTML]{FFFFFF}18796 &
  \cellcolor[HTML]{F4F9F2}15936 &
  \cellcolor[HTML]{8BDAAA}7482 &
  \cellcolor[HTML]{69D196}6240 &
   &
  \cellcolor[HTML]{03B85D}2726 &
  \cellcolor[HTML]{00B85C}\textbf{2430} &
  \cellcolor[HTML]{FFFFFF}39889 &
  \cellcolor[HTML]{F7FAF5}34861 &
  \cellcolor[HTML]{E2EFDB}21702 &
  \cellcolor[HTML]{C8E8CB}19029 \\
\textbf{BeH2\_JW} &
  14 &
   &
  \cellcolor[HTML]{07B960}6812 &
  \cellcolor[HTML]{00B85C}\textbf{5967} &
  \cellcolor[HTML]{FFFFFF}54113 &
  \cellcolor[HTML]{FBFDFA}51123 &
  \cellcolor[HTML]{A6E0B8}23684 &
  \cellcolor[HTML]{83D8A5}19974 &
   &
  \cellcolor[HTML]{0EBB64}2970 &
  \cellcolor[HTML]{00B85C}\textbf{2363} &
  \cellcolor[HTML]{FFFFFF}21072 &
  \cellcolor[HTML]{F4F8F0}17528 &
  \cellcolor[HTML]{68D195}6669 &
  \cellcolor[HTML]{40C780}5044 &
   &
  \cellcolor[HTML]{02B85D}2571 &
  \cellcolor[HTML]{00B85C}\textbf{2372} &
  \cellcolor[HTML]{FFFFFF}44354 &
  \cellcolor[HTML]{F7FAF4}38564 &
  \cellcolor[HTML]{BAE5C4}19706 &
  \cellcolor[HTML]{8EDAAB}15619 \\
\textbf{BeH2\_PM} &
  14 &
   &
  \cellcolor[HTML]{0BBA62}7143 &
  \cellcolor[HTML]{00B85C}\textbf{5806} &
  \cellcolor[HTML]{FFFFFF}59164 &
  \cellcolor[HTML]{FBFCFA}55629 &
  \cellcolor[HTML]{C5E7C9}29069 &
  \cellcolor[HTML]{AAE1BB}25923 &
   &
  \cellcolor[HTML]{13BC66}3051 &
  \cellcolor[HTML]{00B85C}\textbf{2260} &
  \cellcolor[HTML]{FFFFFF}20392 &
  \cellcolor[HTML]{F8FBF6}18418 &
  \cellcolor[HTML]{96DCB0}8309 &
  \cellcolor[HTML]{73D49C}6913 &
   &
  \cellcolor[HTML]{03B85D}2646 &
  \cellcolor[HTML]{00B85C}\textbf{2340} &
  \cellcolor[HTML]{FFFFFF}42891 &
  \cellcolor[HTML]{F9FBF7}39019 &
  \cellcolor[HTML]{E4F0DD}24475 &
  \cellcolor[HTML]{D3EBD1}21316 \\
\textbf{qft\_5} &
  5 &
   &
  \cellcolor[HTML]{22C06F}68 &
  \cellcolor[HTML]{00B85C}\textbf{65} &
  \cellcolor[HTML]{E5F1DE}87 &
  \cellcolor[HTML]{2EC375}69 &
  \cellcolor[HTML]{FFFFFF}104 &
  \cellcolor[HTML]{E4F0DC}86 &
   &
  \cellcolor[HTML]{FFFFFF}26 &
  \cellcolor[HTML]{F5F9F2}25 &
  \cellcolor[HTML]{FFFFFF}26 &
  \cellcolor[HTML]{00B85C}\textbf{20} &
  \cellcolor[HTML]{FFFFFF}26 &
  \cellcolor[HTML]{00B85C}\textbf{20} &
   &
  \cellcolor[HTML]{5ACE8E}31 &
  \cellcolor[HTML]{00B85C}\textbf{25} &
  \cellcolor[HTML]{F1F7ED}48 &
  \cellcolor[HTML]{E3F0DC}41 &
  \cellcolor[HTML]{FFFFFF}55 &
  \cellcolor[HTML]{F3F8F0}49 \\
\textbf{qft\_10} &
  10 &
   &
  \cellcolor[HTML]{72D39B}262 &
  \cellcolor[HTML]{00B85C}\textbf{210} &
  \cellcolor[HTML]{E7F2E1}334 &
  \cellcolor[HTML]{ADE2BC}289 &
  \cellcolor[HTML]{FFFFFF}416 &
  \cellcolor[HTML]{F2F8EE}371 &
   &
  \cellcolor[HTML]{FFFFFF}112 &
  \cellcolor[HTML]{00B85C}\textbf{89} &
  \cellcolor[HTML]{EDF5E8}105 &
  \cellcolor[HTML]{13BC66}90 &
  \cellcolor[HTML]{EDF5E8}105 &
  \cellcolor[HTML]{13BC66}90 &
   &
  \cellcolor[HTML]{63D093}81 &
  \cellcolor[HTML]{00B85C}\textbf{70} &
  \cellcolor[HTML]{EBF4E5}103 &
  \cellcolor[HTML]{E3EFDB}96 &
  \cellcolor[HTML]{FFFFFF}120 &
  \cellcolor[HTML]{F8FBF6}114 \\
\textbf{qft\_20} &
  20 &
   &
  \cellcolor[HTML]{7ED6A2}922 &
  \cellcolor[HTML]{00B85C}\textbf{645} &
  \cellcolor[HTML]{E9F3E3}1269 &
  \cellcolor[HTML]{E4F0DD}1179 &
  \cellcolor[HTML]{FFFFFF}1631 &
  \cellcolor[HTML]{F9FCF8}1541 &
   &
  \cellcolor[HTML]{FFFFFF}440 &
  \cellcolor[HTML]{00B85C}\textbf{289} &
  \cellcolor[HTML]{F3F8F0}410 &
  \cellcolor[HTML]{E7F2E1}380 &
  \cellcolor[HTML]{F3F8F0}410 &
  \cellcolor[HTML]{E7F2E1}380 &
   &
  \cellcolor[HTML]{FFFFFF}270 &
  \cellcolor[HTML]{4DCA87}217 &
  \cellcolor[HTML]{2AC273}212 &
  \cellcolor[HTML]{00B85C}\textbf{206} &
  \cellcolor[HTML]{ECF5E7}250 &
  \cellcolor[HTML]{E7F2E0}244 \\
\textbf{qft\_30} &
  30 &
   &
  \cellcolor[HTML]{BFE6C6}2281 &
  \cellcolor[HTML]{00B85C}\textbf{1280} &
  \cellcolor[HTML]{EAF3E4}2804 &
  \cellcolor[HTML]{E6F1DF}2642 &
  \cellcolor[HTML]{FFFFFF}3646 &
  \cellcolor[HTML]{FAFCF9}3466 &
   &
  \cellcolor[HTML]{FFFFFF}1119 &
  \cellcolor[HTML]{00B85C}\textbf{589} &
  \cellcolor[HTML]{E8F2E2}915 &
  \cellcolor[HTML]{E2EFDA}858 &
  \cellcolor[HTML]{E8F2E2}915 &
  \cellcolor[HTML]{E2EFDA}858 &
   &
  \cellcolor[HTML]{FFFFFF}491 &
  \cellcolor[HTML]{F3F8EF}455 &
  \cellcolor[HTML]{0FBB64}322 &
  \cellcolor[HTML]{00B85C}\textbf{316} &
  \cellcolor[HTML]{A5E0B8}380 &
  \cellcolor[HTML]{95DCAF}374 \\
\textbf{qft\_50} &
  50 &
   &
  \cellcolor[HTML]{E3EFDB}6348 &
  \cellcolor[HTML]{00B85C}\textbf{3150} &
  \cellcolor[HTML]{F0F6EC}7674 &
  \cellcolor[HTML]{C9E9CC}5862 &
  \cellcolor[HTML]{FFFFFF}9236 &
  \cellcolor[HTML]{F0F6EB}7666 &
   &
  \cellcolor[HTML]{FFFFFF}3282 &
  \cellcolor[HTML]{00B85C}\textbf{1489} &
  \cellcolor[HTML]{E6F1DF}2525 &
  \cellcolor[HTML]{67D195}1898 &
  \cellcolor[HTML]{D0EAD0}2315 &
  \cellcolor[HTML]{67D195}1898 &
   &
  \cellcolor[HTML]{FFFFFF}1324 &
  \cellcolor[HTML]{FFFFFF}1165 &
  \cellcolor[HTML]{0DBB63}542 &
  \cellcolor[HTML]{00B85C}\textbf{536} &
  \cellcolor[HTML]{E2EFDA}640 &
  \cellcolor[HTML]{DBEDD6}634 \\
\textbf{grover\_5} &
  5 &
   &
  \cellcolor[HTML]{00B85C}\textbf{37} &
  \cellcolor[HTML]{00B85C}\textbf{37} &
  \cellcolor[HTML]{E7F2E1}52 &
  \cellcolor[HTML]{E7F2E1}52 &
  \cellcolor[HTML]{FFFFFF}62 &
  \cellcolor[HTML]{FFFFFF}62 &
   &
  \cellcolor[HTML]{00B85C}\textbf{13} &
  \cellcolor[HTML]{00B85C}\textbf{13} &
  \cellcolor[HTML]{00B85C}\textbf{13} &
  \cellcolor[HTML]{00B85C}\textbf{13} &
  \cellcolor[HTML]{00B85C}\textbf{13} &
  \cellcolor[HTML]{00B85C}\textbf{13} &
   &
  \cellcolor[HTML]{00B85C}\textbf{6} &
  \cellcolor[HTML]{00B85C}\textbf{6} &
  \cellcolor[HTML]{F1F7ED}35 &
  \cellcolor[HTML]{F1F7ED}35 &
  \cellcolor[HTML]{FFFFFF}44 &
  \cellcolor[HTML]{FFFFFF}44 \\
\textbf{grover\_10} &
  10 &
   &
  \cellcolor[HTML]{5ACE8E}162 &
  \cellcolor[HTML]{00B85C}\textbf{145} &
  \cellcolor[HTML]{DFEED8}187 &
  \cellcolor[HTML]{DFEED8}187 &
  \cellcolor[HTML]{FFFFFF}230 &
  \cellcolor[HTML]{FFFFFF}230 &
   &
  \cellcolor[HTML]{F0F7EC}56 &
  \cellcolor[HTML]{00B85C}\textbf{48} &
  \cellcolor[HTML]{38C57B}49 &
  \cellcolor[HTML]{38C57B}49 &
  \cellcolor[HTML]{38C57B}49 &
  \cellcolor[HTML]{38C57B}49 &
   &
  \cellcolor[HTML]{00B85C}\textbf{11} &
  \cellcolor[HTML]{00B85C}\textbf{11} &
  \cellcolor[HTML]{EFF6EB}94 &
  \cellcolor[HTML]{EFF6EB}94 &
  \cellcolor[HTML]{FFFFFF}124 &
  \cellcolor[HTML]{FFFFFF}124 \\
\textbf{grover\_30} &
  30 &
   &
  \cellcolor[HTML]{29C272}530 &
  \cellcolor[HTML]{00B85C}\textbf{504} &
  \cellcolor[HTML]{D2EBD1}637 &
  \cellcolor[HTML]{D2EBD1}637 &
  \cellcolor[HTML]{FFFFFF}790 &
  \cellcolor[HTML]{FFFFFF}790 &
   &
  \cellcolor[HTML]{F6FAF4}185 &
  \cellcolor[HTML]{00B85C}\textbf{160} &
  \cellcolor[HTML]{A2DFB6}169 &
  \cellcolor[HTML]{A2DFB6}169 &
  \cellcolor[HTML]{A2DFB6}169 &
  \cellcolor[HTML]{A2DFB6}169 &
   &
  \cellcolor[HTML]{00B85C}\textbf{26} &
  \cellcolor[HTML]{00B85C}\textbf{26} &
  \cellcolor[HTML]{EFF6EB}274 &
  \cellcolor[HTML]{EFF6EB}274 &
  \cellcolor[HTML]{FFFFFF}364 &
  \cellcolor[HTML]{FFFFFF}364 \\
\textbf{grover\_80} &
  80 &
   &
  \cellcolor[HTML]{3EC77E}1545 &
  \cellcolor[HTML]{00B85C}\textbf{1442} &
  \cellcolor[HTML]{C1E7C7}1762 &
  \cellcolor[HTML]{C1E7C7}1762 &
  \cellcolor[HTML]{F1F7EE}2190 &
  \cellcolor[HTML]{F1F7EE}2190 &
   &
  \cellcolor[HTML]{FFFFFF}559 &
  \cellcolor[HTML]{00B85C}\textbf{467} &
  \cellcolor[HTML]{09BA61}469 &
  \cellcolor[HTML]{09BA61}469 &
  \cellcolor[HTML]{09BA61}469 &
  \cellcolor[HTML]{09BA61}469 &
   &
  \cellcolor[HTML]{00B85C}\textbf{63} &
  \cellcolor[HTML]{00B85C}\textbf{63} &
  \cellcolor[HTML]{EFF6EB}721 &
  \cellcolor[HTML]{EFF6EB}721 &
  \cellcolor[HTML]{FFFFFF}960 &
  \cellcolor[HTML]{FFFFFF}960 \\
\textbf{grover\_100} &
  100 &
   &
  \cellcolor[HTML]{3AC67C}1936 &
  \cellcolor[HTML]{00B85C}\textbf{1814} &
  \cellcolor[HTML]{C0E6C7}2212 &
  \cellcolor[HTML]{C0E6C7}2212 &
  \cellcolor[HTML]{FFFFFF}2750 &
  \cellcolor[HTML]{FFFFFF}2750 &
   &
  \cellcolor[HTML]{FFFFFF}695 &
  \cellcolor[HTML]{00B85C}\textbf{587} &
  \cellcolor[HTML]{08BA60}589 &
  \cellcolor[HTML]{08BA60}589 &
  \cellcolor[HTML]{08BA60}589 &
  \cellcolor[HTML]{08BA60}589 &
   &
  \cellcolor[HTML]{00B85C}\textbf{78} &
  \cellcolor[HTML]{00B85C}\textbf{78} &
  \cellcolor[HTML]{E2EFDA}901 &
  \cellcolor[HTML]{E2EFDA}901 &
  \cellcolor[HTML]{FFFFFF}1200 &
  \cellcolor[HTML]{FFFFFF}1200 \\
\textbf{hea5\_l\_20} &
  5 &
   &
  \cellcolor[HTML]{F8FBF7}341 &
  \cellcolor[HTML]{00B85C}\textbf{307} &
  \cellcolor[HTML]{E2EFDA}326 &
  \cellcolor[HTML]{E2EFDA}326 &
  \cellcolor[HTML]{FFFFFF}345 &
  \cellcolor[HTML]{FFFFFF}345 &
   &
  \cellcolor[HTML]{F8FBF6}93 &
  \cellcolor[HTML]{00B85C}\textbf{80} &
  \cellcolor[HTML]{00B85C}\textbf{80} &
  \cellcolor[HTML]{00B85C}\textbf{80} &
  \cellcolor[HTML]{00B85C}\textbf{80} &
  \cellcolor[HTML]{00B85C}\textbf{80} &
   &
  \cellcolor[HTML]{E2EFDA}114 &
  \cellcolor[HTML]{E2EFDA}108 &
  \cellcolor[HTML]{00B85C}\textbf{106} &
  \cellcolor[HTML]{00B85C}\textbf{106} &
  \cellcolor[HTML]{00B85C}\textbf{106} &
  \cellcolor[HTML]{00B85C}\textbf{106} \\
\textbf{hea5\_c\_20} &
  5 &
   &
  \cellcolor[HTML]{FFFFFF}473 &
  \cellcolor[HTML]{00B85C}\textbf{380} &
  \cellcolor[HTML]{79D59F}405 &
  \cellcolor[HTML]{79D59F}405 &
  \cellcolor[HTML]{79D59F}405 &
  \cellcolor[HTML]{79D59F}405 &
   &
  \cellcolor[HTML]{FFFFFF}184 &
  \cellcolor[HTML]{86D8A7}125 &
  \cellcolor[HTML]{00B85C}\textbf{100} &
  \cellcolor[HTML]{00B85C}\textbf{100} &
  \cellcolor[HTML]{00B85C}\textbf{100} &
  \cellcolor[HTML]{00B85C}\textbf{100} &
   &
  \cellcolor[HTML]{DDEDD7}221 &
  \cellcolor[HTML]{00B85C}\textbf{175} &
  \cellcolor[HTML]{E2EFDA}222 &
  \cellcolor[HTML]{E2EFDA}222 &
  \cellcolor[HTML]{E2EFDA}222 &
  \cellcolor[HTML]{E2EFDA}222 \\
\textbf{hea5\_f\_20} &
  5 &
   &
  \cellcolor[HTML]{15BD68}326 &
  \cellcolor[HTML]{00B85C}\textbf{307} &
  \cellcolor[HTML]{9DDEB4}446 &
  \cellcolor[HTML]{9DDEB4}446 &
  \cellcolor[HTML]{FFFFFF}705 &
  \cellcolor[HTML]{FFFFFF}705 &
   &
  \cellcolor[HTML]{12BC66}85 &
  \cellcolor[HTML]{00B85C}\textbf{80} &
  \cellcolor[HTML]{FFFFFF}200 &
  \cellcolor[HTML]{FFFFFF}200 &
  \cellcolor[HTML]{FFFFFF}200 &
  \cellcolor[HTML]{FFFFFF}200 &
   &
  \cellcolor[HTML]{08BA60}111 &
  \cellcolor[HTML]{00B85C}\textbf{108} &
  \cellcolor[HTML]{68D196}144 &
  \cellcolor[HTML]{68D196}144 &
  \cellcolor[HTML]{FFFFFF}264 &
  \cellcolor[HTML]{FFFFFF}264 \\
\textbf{hea10\_l\_40} &
  10 &
   &
  \cellcolor[HTML]{FFFFFF}1591 &
  \cellcolor[HTML]{00B85C}\textbf{1412} &
  \cellcolor[HTML]{62CF92}1451 &
  \cellcolor[HTML]{62CF92}1451 &
  \cellcolor[HTML]{C4E7C9}1490 &
  \cellcolor[HTML]{C4E7C9}1490 &
   &
  \cellcolor[HTML]{FFFFFF}480 &
  \cellcolor[HTML]{00B85C}\textbf{360} &
  \cellcolor[HTML]{00B85C}\textbf{360} &
  \cellcolor[HTML]{00B85C}\textbf{360} &
  \cellcolor[HTML]{00B85C}\textbf{360} &
  \cellcolor[HTML]{00B85C}\textbf{360} &
   &
  \cellcolor[HTML]{FFFFFF}277 &
  \cellcolor[HTML]{0EBB64}218 &
  \cellcolor[HTML]{00B85C}\textbf{216} &
  \cellcolor[HTML]{00B85C}\textbf{216} &
  \cellcolor[HTML]{00B85C}\textbf{216} &
  \cellcolor[HTML]{00B85C}\textbf{216} \\
\textbf{hea10\_c\_40} &
  10 &
   &
  \cellcolor[HTML]{FFFFFF}4348 &
  \cellcolor[HTML]{ACE1BC}2655 &
  \cellcolor[HTML]{00B85C}\textbf{1610} &
  \cellcolor[HTML]{00B85C}\textbf{1610} &
  \cellcolor[HTML]{00B85C}\textbf{1610} &
  \cellcolor[HTML]{00B85C}\textbf{1610} &
   &
  \cellcolor[HTML]{FFFFFF}1749 &
  \cellcolor[HTML]{C1E7C7}977 &
  \cellcolor[HTML]{00B85C}\textbf{400} &
  \cellcolor[HTML]{00B85C}\textbf{400} &
  \cellcolor[HTML]{00B85C}\textbf{400} &
  \cellcolor[HTML]{00B85C}\textbf{400} &
   &
  \cellcolor[HTML]{FFFFFF}924 &
  \cellcolor[HTML]{00B85C}\textbf{659} &
  \cellcolor[HTML]{EDF5E8}842 &
  \cellcolor[HTML]{EDF5E8}842 &
  \cellcolor[HTML]{EDF5E8}842 &
  \cellcolor[HTML]{EDF5E8}842 \\
\textbf{hea10\_f\_40} &
  10 &
   &
  \cellcolor[HTML]{12BC66}1592 &
  \cellcolor[HTML]{00B85C}\textbf{1412} &
  \cellcolor[HTML]{98DCB0}2891 &
  \cellcolor[HTML]{98DCB0}2891 &
  \cellcolor[HTML]{FFFFFF}5810 &
  \cellcolor[HTML]{FFFFFF}5810 &
   &
  \cellcolor[HTML]{24C070}476 &
  \cellcolor[HTML]{00B85C}\textbf{360} &
  \cellcolor[HTML]{FFFFFF}1800 &
  \cellcolor[HTML]{FFFFFF}1800 &
  \cellcolor[HTML]{FFFFFF}1800 &
  \cellcolor[HTML]{FFFFFF}1800 &
   &
  \cellcolor[HTML]{1CBE6B}275 &
  \cellcolor[HTML]{00B85C}\textbf{218} &
  \cellcolor[HTML]{86D8A6}489 &
  \cellcolor[HTML]{86D8A6}489 &
  \cellcolor[HTML]{FFFFFF}1129 &
  \cellcolor[HTML]{FFFFFF}1129 \\
\textbf{hea20\_l\_50} &
  20 &
   &
  \cellcolor[HTML]{FFFFFF}4065 &
  \cellcolor[HTML]{00B85C}\textbf{3772} &
  \cellcolor[HTML]{4BCA86}3821 &
  \cellcolor[HTML]{4BCA86}3821 &
  \cellcolor[HTML]{97DCB0}3870 &
  \cellcolor[HTML]{97DCB0}3870 &
   &
  \cellcolor[HTML]{FFFFFF}1152 &
  \cellcolor[HTML]{00B85C}\textbf{950} &
  \cellcolor[HTML]{00B85C}\textbf{950} &
  \cellcolor[HTML]{00B85C}\textbf{950} &
  \cellcolor[HTML]{00B85C}\textbf{950} &
  \cellcolor[HTML]{00B85C}\textbf{950} &
   &
  \cellcolor[HTML]{FFFFFF}401 &
  \cellcolor[HTML]{07B960}288 &
  \cellcolor[HTML]{00B85C}\textbf{286} &
  \cellcolor[HTML]{00B85C}\textbf{286} &
  \cellcolor[HTML]{00B85C}\textbf{286} &
  \cellcolor[HTML]{00B85C}\textbf{286} \\
\textbf{hea20\_c\_50} &
  20 &
   &
  \cellcolor[HTML]{FFFFFF}20533 &
  \cellcolor[HTML]{12BC66}4679 &
  \cellcolor[HTML]{00B85C}\textbf{4020} &
  \cellcolor[HTML]{00B85C}\textbf{4020} &
  \cellcolor[HTML]{00B85C}\textbf{4020} &
  \cellcolor[HTML]{00B85C}\textbf{4020} &
   &
  \cellcolor[HTML]{FFFFFF}9048 &
  \cellcolor[HTML]{1DBF6C}1524 &
  \cellcolor[HTML]{00B85C}\textbf{1000} &
  \cellcolor[HTML]{00B85C}\textbf{1000} &
  \cellcolor[HTML]{00B85C}\textbf{1000} &
  \cellcolor[HTML]{00B85C}\textbf{1000} &
   &
  \cellcolor[HTML]{F3F8F0}2207 &
  \cellcolor[HTML]{00B85C}\textbf{1275} &
  \cellcolor[HTML]{EDF5E9}2052 &
  \cellcolor[HTML]{EDF5E9}2052 &
  \cellcolor[HTML]{EDF5E9}2052 &
  \cellcolor[HTML]{EDF5E9}2052 \\
\textbf{hea20\_f\_50} &
  20 &
   &
  \cellcolor[HTML]{05B95E}4063 &
  \cellcolor[HTML]{00B85C}\textbf{3772} &
  \cellcolor[HTML]{96DCB0}12371 &
  \cellcolor[HTML]{96DCB0}12371 &
  \cellcolor[HTML]{FFFFFF}29520 &
  \cellcolor[HTML]{FFFFFF}29520 &
   &
  \cellcolor[HTML]{0ABA61}1147 &
  \cellcolor[HTML]{00B85C}\textbf{950} &
  \cellcolor[HTML]{FFFFFF}9500 &
  \cellcolor[HTML]{FFFFFF}9500 &
  \cellcolor[HTML]{FFFFFF}9500 &
  \cellcolor[HTML]{FFFFFF}9500 &
   &
  \cellcolor[HTML]{13BC66}401 &
  \cellcolor[HTML]{00B85C}\textbf{288} &
  \cellcolor[HTML]{8EDAAB}1119 &
  \cellcolor[HTML]{8EDAAB}1119 &
  \cellcolor[HTML]{FFFFFF}2919 &
  \cellcolor[HTML]{FFFFFF}2919 \\
\textbf{qaoa\_6\_3} &
  6 &
   &
  \cellcolor[HTML]{E7F2E1}98 &
  \cellcolor[HTML]{00B85C}\textbf{83} &
  \cellcolor[HTML]{24C070}85 &
  \cellcolor[HTML]{6CD298}89 &
  \cellcolor[HTML]{FFFFFF}108 &
  \cellcolor[HTML]{FFFFFF}108 &
   &
  \cellcolor[HTML]{E2EFDA}23 &
  \cellcolor[HTML]{00B85C}\textbf{22} &
  \cellcolor[HTML]{E7F2E1}24 &
  \cellcolor[HTML]{E7F2E1}24 &
  \cellcolor[HTML]{E7F2E1}24 &
  \cellcolor[HTML]{E7F2E1}24 &
  \multicolumn{1}{c}{} &
  \cellcolor[HTML]{4FCB88}36 &
  \cellcolor[HTML]{00B85C}\textbf{33} &
  \cellcolor[HTML]{6AD197}37 &
  \cellcolor[HTML]{6AD197}37 &
  \cellcolor[HTML]{FFFFFF}50 &
  \cellcolor[HTML]{FFFFFF}50 \\
\textbf{qaoa\_6\_6} &
  6 &
   &
  \cellcolor[HTML]{8AD9A9}97 &
  \cellcolor[HTML]{00B85C}\textbf{86} &
  \cellcolor[HTML]{57CD8D}93 &
  \cellcolor[HTML]{57CD8D}93 &
  \cellcolor[HTML]{FFFFFF}122 &
  \cellcolor[HTML]{FFFFFF}122 &
   &
  \cellcolor[HTML]{71D39B}25 &
  \cellcolor[HTML]{00B85C}\textbf{24} &
  \cellcolor[HTML]{F0F7EC}28 &
  \cellcolor[HTML]{F0F7EC}28 &
  \cellcolor[HTML]{F0F7EC}28 &
  \cellcolor[HTML]{F0F7EC}28 &
   &
  \cellcolor[HTML]{20BF6E}37 &
  \cellcolor[HTML]{00B85C}\textbf{34} &
  \cellcolor[HTML]{B6E4C2}51 &
  \cellcolor[HTML]{B6E4C2}51 &
  \cellcolor[HTML]{FFFFFF}76 &
  \cellcolor[HTML]{FFFFFF}76 \\
\textbf{qaoa\_17\_3} &
  17 &
   &
  \cellcolor[HTML]{70D39A}432 &
  \cellcolor[HTML]{00B85C}\textbf{381} &
  \cellcolor[HTML]{37C57A}406 &
  \cellcolor[HTML]{37C57A}406 &
  \cellcolor[HTML]{FFFFFF}586 &
  \cellcolor[HTML]{FFFFFF}586 &
   &
  \cellcolor[HTML]{FFFFFF}148 &
  \cellcolor[HTML]{00B85C}\textbf{136} &
  \cellcolor[HTML]{FFFFFF}148 &
  \cellcolor[HTML]{FFFFFF}148 &
  \cellcolor[HTML]{FFFFFF}148 &
  \cellcolor[HTML]{FFFFFF}148 &
   &
  \cellcolor[HTML]{ADE2BC}90 &
  \cellcolor[HTML]{00B85C}\textbf{75} &
  \cellcolor[HTML]{45C882}81 &
  \cellcolor[HTML]{45C882}81 &
  \cellcolor[HTML]{FFFFFF}114 &
  \cellcolor[HTML]{FFFFFF}114 \\
\textbf{qaoa\_17\_6} &
  17 &
   &
  \cellcolor[HTML]{3BC67D}731 &
  \cellcolor[HTML]{00B85C}\textbf{659} &
  \cellcolor[HTML]{54CC8B}761 &
  \cellcolor[HTML]{54CC8B}761 &
  \cellcolor[HTML]{FFFFFF}1202 &
  \cellcolor[HTML]{FFFFFF}1202 &
   &
  \cellcolor[HTML]{8FDAAB}253 &
  \cellcolor[HTML]{00B85C}\textbf{220} &
  \cellcolor[HTML]{E4F0DC}324 &
  \cellcolor[HTML]{E4F0DC}324 &
  \cellcolor[HTML]{E4F0DC}324 &
  \cellcolor[HTML]{E4F0DC}324 &
   &
  \cellcolor[HTML]{0EBB64}128 &
  \cellcolor[HTML]{00B85C}\textbf{124} &
  \cellcolor[HTML]{87D9A7}161 &
  \cellcolor[HTML]{87D9A7}161 &
  \cellcolor[HTML]{FFFFFF}247 &
  \cellcolor[HTML]{FFFFFF}247 \\
\textbf{qaoa\_28\_3} &
  28 &
   &
  \cellcolor[HTML]{5FCF91}1007 &
  \cellcolor[HTML]{00B85C}\textbf{887} &
  \cellcolor[HTML]{33C478}952 &
  \cellcolor[HTML]{33C478}952 &
  \cellcolor[HTML]{E3EFDB}1456 &
  \cellcolor[HTML]{E3EFDB}1456 &
   &
  \cellcolor[HTML]{E2EFDA}356 &
  \cellcolor[HTML]{00B85C}\textbf{328} &
  \cellcolor[HTML]{E3EFDB}384 &
  \cellcolor[HTML]{E3EFDB}384 &
  \cellcolor[HTML]{E3EFDB}384 &
  \cellcolor[HTML]{E3EFDB}384 &
   &
  \cellcolor[HTML]{91DBAC}151 &
  \cellcolor[HTML]{00B85C}\textbf{116} &
  \cellcolor[HTML]{A5E0B8}156 &
  \cellcolor[HTML]{A5E0B8}156 &
  \cellcolor[HTML]{FFFFFF}225 &
  \cellcolor[HTML]{FFFFFF}225 \\
\textbf{qaoa\_28\_6} &
  28 &
   &
  \cellcolor[HTML]{28C172}1746 &
  \cellcolor[HTML]{00B85C}\textbf{1617} &
  \cellcolor[HTML]{4FCB88}1870 &
  \cellcolor[HTML]{4FCB88}1870 &
  \cellcolor[HTML]{E6F1DF}3052 &
  \cellcolor[HTML]{E6F1DF}3052 &
   &
  \cellcolor[HTML]{76D49E}660 &
  \cellcolor[HTML]{00B85C}\textbf{596} &
  \cellcolor[HTML]{E9F3E3}840 &
  \cellcolor[HTML]{E9F3E3}840 &
  \cellcolor[HTML]{E9F3E3}840 &
  \cellcolor[HTML]{E9F3E3}840 &
   &
  \cellcolor[HTML]{42C880}248 &
  \cellcolor[HTML]{00B85C}\textbf{218} &
  \cellcolor[HTML]{84D8A5}278 &
  \cellcolor[HTML]{84D8A5}278 &
  \cellcolor[HTML]{FFFFFF}423 &
  \cellcolor[HTML]{FFFFFF}423 \\
\textbf{qaoa\_40\_3} &
  40 &
   &
  \cellcolor[HTML]{5ECF90}2169 &
  \cellcolor[HTML]{00B85C}\textbf{1882} &
  \cellcolor[HTML]{32C478}2035 &
  \cellcolor[HTML]{32C478}2035 &
  \cellcolor[HTML]{E6F1DF}3254 &
  \cellcolor[HTML]{E6F1DF}3254 &
   &
  \cellcolor[HTML]{E6F1DF}865 &
  \cellcolor[HTML]{00B85C}\textbf{725} &
  \cellcolor[HTML]{E7F2E1}884 &
  \cellcolor[HTML]{E7F2E1}884 &
  \cellcolor[HTML]{E7F2E1}884 &
  \cellcolor[HTML]{E7F2E1}884 &
   &
  \cellcolor[HTML]{42C881}251 &
  \cellcolor[HTML]{00B85C}\textbf{222} &
  \cellcolor[HTML]{83D7A5}279 &
  \cellcolor[HTML]{83D7A5}279 &
  \cellcolor[HTML]{FFFFFF}418 &
  \cellcolor[HTML]{FFFFFF}418 \\
\textbf{qaoa\_40\_6} &
  40 &
   &
  \cellcolor[HTML]{20C06E}3456 &
  \cellcolor[HTML]{00B85C}\textbf{3220} &
  \cellcolor[HTML]{5BCE8F}3876 &
  \cellcolor[HTML]{5BCE8F}3876 &
  \cellcolor[HTML]{F7FBF5}6460 &
  \cellcolor[HTML]{F7FBF5}6460 &
   &
  \cellcolor[HTML]{73D49C}1361 &
  \cellcolor[HTML]{00B85C}\textbf{1210} &
  \cellcolor[HTML]{F3F8F0}1800 &
  \cellcolor[HTML]{F3F8F0}1800 &
  \cellcolor[HTML]{F3F8F0}1800 &
  \cellcolor[HTML]{F3F8F0}1800 &
   &
  \cellcolor[HTML]{53CC8A}413 &
  \cellcolor[HTML]{00B85C}\textbf{364} &
  \cellcolor[HTML]{5CCE8F}418 &
  \cellcolor[HTML]{5CCE8F}418 &
  \cellcolor[HTML]{E5F1DE}629 &
  \cellcolor[HTML]{E5F1DE}629 \\
  \bottomrule
\end{tabular}%
}
\end{table*}

\PCAST is implemented in C++ as the core optimization of the Intel Quantum SDK, enabled by the (-O1) flag.
In this section, we evaluate its performance against IBM's Qiskit~\cite{Qiskit} and Quantinuum's \tket~\cite{Sivarajah_2021_tket} optimizing compilers.

\subsection{Experimental Setup}
\textbf{System:} Our experiments use an Intel Xeon\registered Platinum CPU ($2.4$GHz, $2$TB RAM) and Python $3.10$.

\textbf{Framework Setup:} We compare against the Qiskit transpiler with optimization levels $2$ (\qiskittwo) and $3$ (\qiskitthree). For \tket, we compare with two predefined optimization sequences: \tketone, comprised mainly of the \textit{SynthesiseTket} pass; and \tkettwo, comprising of the \textit{FullPeepholeOptimise} pass.\footnote{In both cases we add other passes like \textit{PauliSimp} and \textit{OptimisePhaseGadgets} when they improve \tket's performance.}

For \PCOAST we compare optimization level 1 (\PCOASTone) and a version where we fine-tune some parameters of the search functions (\PCOASTFT), as discussed in \cref{sec:pcast_ft}.

\textbf{Benchmarks:} We analyze the compilation performance for a total of $36$ benchmarks of different configurations and sizes: the Unitary Coupled-Cluster Single and Double excitations (UCCSD) ansatz~\cite{uccsd}, Quantum Fourier Transform (QFT)~\cite{nielsen2002quantum}, Grover's Diffusion operator~\cite{grover}, Hardware-Efficient Ansatz (HEA)~\cite{Kandala2017}, and the Quantum Approximate Optimization Algorithm (QAOA) ansatz~\cite{farhi2014quantum}.

For UCCSD, we construct ansatz for $3$ molecules (H$_2$, LiH, and BeH$_2$) obtained through $3$ fermionic mapping techniques: Jordan-Wigner (JW)~\cite{Jordan1928}, Bravyi–Kitaev (BK)~\cite{Bravyi_2002}, and Parity Mapping (PM)~\cite{parity}. We test HEA for different circuit sizes and entanglement arrangements: linear, circular, and full. Finally, we construct QAOA ansatz for graph MaxCut on Erdős-Rényi random graphs~\cite{erdos_graph} of different sizes with edge probabilities of $0.3$ and $0.6$. 

\textbf{Gate Set Conversions:} To guarantee fairness, we map both Qiskit and \tket's compilation workflows to Intel's native gate set $\{ \RXY, \CZ, \RZ, \texttt{I}\}$. This is done in Qiskit by adding equivalence rules to the SessionEquivalenceLibrary class, such as connecting $\RXY$ (\textit{RGate} in Qiskit) to other gates and decomposing CX into $\RXY; \CZ; \RXY$.
For \tket this is done via a \textit{CustomRebase} pass that maps the gate set to \tket's TK1 and CX gates, and applied at the end of \tketone and \tkettwo.

\subsection{Gate Counts \& Circuit Depth}
\cref{tab:gate_counts_combined} shows the total gate count, two-qubit gates, and depth for \PCAST, Qiskit, and \tket across all benchmarks. Results assume an all-to-all connected machine as a backend and are obtained from multiple runs to guarantee consistency.

The untuned implementation, \PCOASTone, reduces the total gate count by $22.51\%$ (respectively $33.70\%$), single-qubit gates by $22.06\%$ ($43.58\%$), two-qubit gates by $13.46\%$ ($2.62\%$), and depth by $36.02\%$ ($45.52\%$) compared to the best Qiskit (\tket) performance across all benchmarks.

\PCOASTone performs exceptionally well with UCCSD, with average reductions in total gate count by $76.42\%$ ($55.51\%$), 
%single-qubit gates by $77.78\%$ ($62.68\%$), 
two-qubit gates by $72.52\%$ ($32.45\%$), and depth by $84.82\%$ ($74.30\%$) compared to the best Qiskit (\tket) results. This can be attributed to the fact that all fermionic mapping methods are equivalent up to conjugation by a Clifford, and thus are naturally captured by Pauli frames. 
\begin{ExtendedVersion}
In particular, one can go from one mapping to the other via an $\mathcal{O}(N)$ circuit. Moreover, because of how synthesis works, if one mapping is dramatically worse than the other, the process will generically gravitate to the more optimal mapping as the circuit is being synthesized.
\end{ExtendedVersion}

In some cases, \PCOASTone is outperformed by Qiskit and \tket in two-qubit gate counts for QFT ($34.71\%$) and HEA ($34.14\%$). This is a result of \PCOASTone's circuit synthesis cost function optimizing mainly for circuit depth.
%, which can be adjusted or fine-tuned to get optimized results per benchmark (Section~\ref{sec:pcast_ft}).

\begin{figure*}[h!]
    \centering
    \includegraphics[width=\linewidth, keepaspectratio]{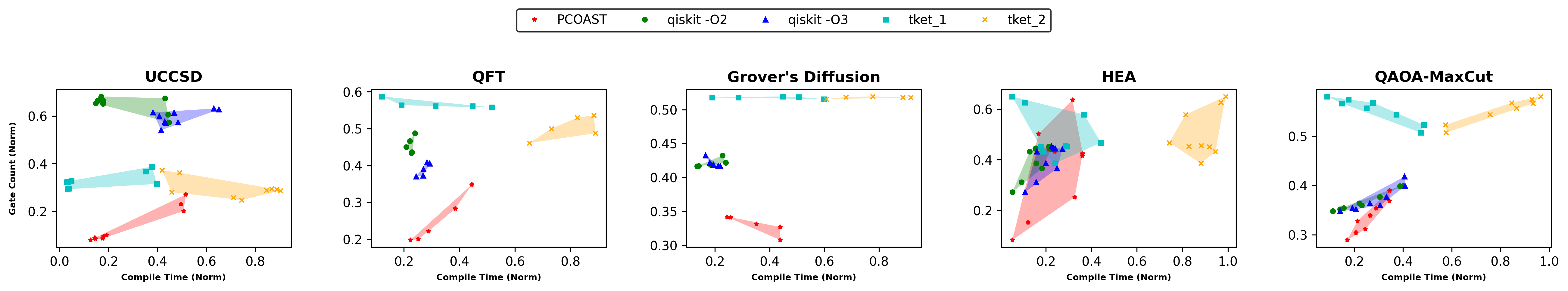}
    \caption{Normalized gate counts vs. compile time for each framework to evaluate the scalability across different benchmarks.}
    \label{fig:scalability}
\end{figure*}

\subsection{Fine-tuning \PCOAST's Cost Function}\label{sec:pcast_ft}
% \albert{do we want references to the flags themselves? These flag are not currently available in v1.0 but they could be later. ALSO, if we widen the table, I think we can combine sections.}
%From Section~\ref{sec:synthesis}, we learn that the $\gateCost$ function includes a circuit depth credit, which controls how much the gate's pace (the difference between the scheduled time and the latest scheduled gate) contributes to the cost (Eqn. (26) in~\cite{Schmitz2021}). This credit value is referred to as a \textit{parallelization credit} $C$, and it's set to $1.0$ by default. Further, the $\gateCost$ calculation can be modified to assign greater weights to nodes free of dependencies (see footnote~\ref{foot:rflag}), achieved by compiling \PCAST with the flag ``-$r$''. Thus, we fine-tune the cost function by scanning multiple values of $C$ and adjusting the weighting of free nodes to achieve better performance for each benchmark. Additionally, we utilize \textit{release outcome} whenever a workload represents a full computation (UCCSD, HEA, and QAOA), not a subroutine.
%\input{src/tables/pcoast_ft}

To investigate the cases where \PCOASTone performs poorly (\textit{qft\_50}, \textit{hea10\_c} and \textit{hea20\_c}), we modify the parallelization credit, as described in \cref{foot:Cflag}, away from its default value of $1.0$, and we add the weighting modification discussed in \cref{foot:rflag}. Additionally, we utilize a release outcome whenever a workload represents a full computation (UCCSD, HEA, and QAOA), not a subroutine. Table~\ref{tab:gate_counts_combined}'s \PCOASTFT shows the results of these modifications. We see that by fine-tuning the cost function, we are able to achieve better performance in all cases, where \PCOASTFT reduces the gate count by $50.33\%$, $38.24\%$, and $76.82\%$ for the respective workloads, as compared to \PCOASTone.
%This can be attributed to the nature of the workload itself. By looking at the numbers from both Qiskit and \tket, we see that both frameworks result in the exact number of gates, which is expected for this type of ansatz, as it's designed with efficiency in mind. If we want \PCAST to match the efficiency of how a HEA is implemented, the synthesis needs to be able to apply a TQE in such a manner that results in two rotations (simultaneously reducing the cost of two rotations from support 2 to support 1). In the standard \PCAST implementation, rotations at the beginning of the circuit don't weigh any higher than rotations at the end. This means that when a TQE gate is applied, even if it reduces the two rotations to support 1, this can be offset by an operation happening further in the dependency graph, making \PCAST inefficient in capturing the repeated pattern of ``TQE, two rotations, TQE, etc.'' within HEAs with large depths. However, by assigning more weight to elements occurring at the beginning of the circuit (through the ``r" flag), the optimization is more likely to capture this pattern.
\begin{ExtendedVersion}
This can be attributed to the initial implementation of workload. Both Qiskit and \tket result in the same number of gates, owing to the already efficient design of the circuit. In particular, it contains a common repeated pattern of ``TQE gate followed by two rotations.'' For \PCAST to match that efficiency, the synthesis needs to select TQE gates which result in the reduction of two available rotations simultaneously. By assigning more weight to elements occurring at the beginning of the DAG, the optimization is more likely to capture this pattern.
\end{ExtendedVersion}

%Our empirical analysis reveals that tuning the parallelization credit at a value less than $0.1$ generally achieves better results across different benchmarks. Additionally, we notice that as we increase its value beyond $1$, \PCOASTFT's performance shows limited variation, due to the parallelization part of the cost function dominating all other considerations.
Our empirical analysis reveals that setting the parallelization credit less than $0.1$ tends to achieve the best results. As we increase its value beyond $1$, performance plateaus due to the dominance of the parallelization part of the cost function.

Overall, \PCOASTFT reduces total gate count, two-qubit gates, and depth by $16.79\%$, $20.5\%$, and $11.28\%$ respectively compared to \PCOASTone, and by $32.53\% (43.33\%)$, $29.22\% (20.58\%)$, and $42.05\% (51.27\%)$ compared to the best Qiskit (\tket) performance across all benchmarks.

% \begin{figure}[!t]
%     \centering
%     \includegraphics[width=\linewidth, keepaspectratio]{src/figures/c_vals_v1.png}
%     \caption{Total gate count reduction as a result of fine-tuning \PCAST's cost function's \textit{Parallelization Credit}. Values are normalized for each benchmark, with $1$ denoting the lowest (best) gate count.}
%     \label{fig:c_vals}
% \end{figure}

\subsection{Scalability} \label{sec:scalability}

\ifExtended{\cref{sec:timeComplexity}}{Ref~\cite{PaykinSchmitz2023PCOAST}} argues that the complexity of \PCOAST's search algorithm is $\mathcal{O}(N^3|G|^2)$, where $N$ is the number of qubits and $|G|$ the number of nodes.
%We also showed \PCAST's performance for circuits with up to 100 qubits in \cref{tab:gate_counts_combined}.
\cref{tab:gate_counts_combined} shows that \PCAST can be applied to circuits with up to 100 qubits.
To gain a comprehensive understanding of scalability and quality, we compared compilation time and performance (gate count) across various toolchains (Fig.~\ref{fig:scalability}), revealing that \PCAST's scalability is superior to other frameworks. Its data points consistently reside in the lower left region, indicating better results, with a few exceptions in HEA. %While \textit{tket\_1} offers fast compilation, its gate count is high; \textit{tket\_2} experiences slow compile times and diminished performance with growing workloads. Qiskit configurations show good compile times but produce higher gate counts for UCCSD compared to other frameworks.

\section{Related Work}
\label{sec:related}
%Quantum circuit optimizations fall in two main classes: local, peephole-style optimizations~\citep{prasad2006data,kliuchnikov2013optimization, abdessaied2014quantum, nam2018automated,Qiskit,Sivarajah_2021_tket, pointing2021optimizing, xu2022quartz}, where local patterns of gates are replaced by other patterns; and global optimizations, where circuits are converted to an intermediate mathematical structure that highlights some semantic equivalence, simplified according to the rules of that structure, and synthesized back into a circuit that could be significantly different from the original. While peephole optimizations can be effective for recognizing common, simple patterns, including between unitary and non-unitary gates, they tend to miss richer opportunities for optimizations, both because they only look locally at neighboring gates, and because of the number and complexity of patterns needed to identify such opportunities.

Global quantum circuit optimizations, such as phase polynomials~\citep{amy2014polynomial, amy2018towards, nam2018automated, Meijer2020architecture}, the ZX-calculus~\citep{kissinger2020,kissinger2020reducing, deBeaudrap2020fusion, cowtan2019phase}, Pauli strings~\citep{li2022paulihedral}, and Pauli rotations~\citep{Zhang2019,Schmitz2021}, leverage mathematical structures to reduce gate count. Most focus on unitary optimizations, with the exception of some ZX variants~\citep{borgna2021hybrid,deBeaudrap2020fusion}. Most similar to \PCAST graphs is \citep{Zhang2019}'s DAGs of Pauli rotations and \citep{cowtan2019phase}'s ZX-based Pauli gadgets. Both overlap with \cref{sec:toGraph} when restricted to unitary gates, but neither use Pauli frames to represent Cliffords, nor address efficient Pauli gadget synthesis. With ZX-based approaches in particular, optimizations must maintain a ``circuit-like'' form, as not all ZX diagrams can be directly synthesized into gates~\citep{kissinger2020, deBeaudrap2020fusion, deBeaudrap2022circuit}. In contrast, all \PCOAST nodes are synthesizable, as synthesis is built into the framework itself.

Bottom-up synthesis methods construct parameterized circuits by iteratively adding gates and using numerical optimization algorithms for parameter determination~\citep{davis2020towards,younis2021qfast,rakyta2022efficient,madden2022best}. 
%Notable techniques  include QSearch~\cite{davis2020towards}, which formulates synthesis as a search across circuit templates; QFAST~\cite{younis2021qfast}, which unifies search and optimization through innovative circuit encoding; SQUANDER~\cite{rakyta2022efficient}, employing a straightforward strategy of enlarging and compressing circuits; and approaches in~\cite{madden2022best} that develop approximate circuits through instantiation. 
Compilation algorithms like QGo~\cite{ wu2021reoptimization} and QUEST~\cite{patel2022quest} leverage bottom-up synthesis for larger circuits, though scalability remains a concern because their search space increases exponentially with circuit size. In contrast, \PCAST scales well to large circuits, as demonstrated in \cref{sec:scalability}.

\citet{Schmitz2021} and \citet{li2022paulihedral} both address circuit synthesis from Pauli strings in the context of Hamiltonian simulation.
%with the goal of starting with a Pauli string representation of a Hamiltonian, and ending with a realizable circuit. 
%\citet{li2022paulihedral} goes so far as to propose treating its mathematical representation Paulihedral (based on Pauli strings) as an IR, or intermediate representation, to enable hardware-aware optimization and scheduling passes. In a similar way, \PCAST graphs could too be seen as an IR for the internal optimizations in \cref{sec:opts}, although IRs are more typically text-based rather than graph-based structures. \albert{that's not my understanding of IR...} 
\citep{Schmitz2021} is the basis of the \PCAST synthesis algorithm, extended to support non-unitary gates and custom cost functions that allow the ultra-greedy search to be fine-tuned.
\citeauthor{li2022paulihedral} incorporate hardware-aware optimization and scheduling passes into Hamiltonian synthesis and, though out of scope of this work, we will extend \PCOAST search functions with such hardware-aware considerations in the near future.

\section{Conclusion}
\label{sec:conclusion}
\label{sec:conc}
\PCAST is a novel optimization framework for mixed unitary and non-unitary quantum circuits that adapts the commutativity properties of Cliffords and Pauli strings to preparation and measurement gates in the \PCAST graph. Internal optimizations simplify the graph depending on whether the quantum state needs to be preserved (\hold) or can be released (\release) after circuit execution. Finally, a customizable greedy search algorithm finds an efficient gate implementation for the optimized \PCAST graph.

Implemented in the Intel Quantum SDK, \PCAST significantly reduces gate count, two-qubit gates, and depth in key benchmarks. With minor tuning, it reduces total gate count by between 32\% (resp. 43\%) two-qubit gates by 29\% (21\%), and depth by 42\% (51\%) compared to the best performance of Qiskit (resp. \tket). On applications for quantum chemistry, it reduces gate count by 79\% (62\%), two-qubit gates by 77\% (54\%), and depth by 85\% (76\%).

The framework leaves many avenues for future work.
\PCOAST can be used as an IR beyond circuit conversion for Hamiltonian simulation~\citep{Schmitz2021} and higher-order circuit transformations~\citep{schmitz2023functional}.
%\PCAST nodes allow for higher-order transformations like quantum control and unitary inverse of entire subgraphs, as enabled in Intel's quantum kernel expressions~\citep{schmitz2023functional},
%and is a natural candidate for Hamiltonian simulation as in~\citep{Schmitz2021}.
Future internal optimizations could include more advanced unitary optimizations such as singlet node to factor node merging and incorporating other representations like phase polynomials into \PCAST.
%For synthesis, the search functions can be generalized to capture connectivity constraints and noise models by adapting state-of-the-art [adjective] methods.
For synthesis, immediate next steps will adapt state-of-the-art methods for limited connectivity in NISQ architectures by incorporating connectivity and noise into the search functions.
%Additional refinements to \gateCost including a weighting of the node averaging based on distance in the partial ordering of the \PCAST graph. 
%\todo{Application to circuit cutting?}

\bibliography{references}

\end{document}